%% file: main.tex
\documentclass{article}

\usepackage{microtype}
\usepackage{graphicx}
\usepackage{subcaption}
\usepackage{booktabs} 
\usepackage{multirow}
\usepackage{bm}
\usepackage{makecell}
\usepackage[table]{xcolor}
\usepackage{adjustbox}

\usepackage{amsmath}
\usepackage{amssymb}
\usepackage{mathtools}
\usepackage{amsthm}

\newenvironment{icompact}{
  \begin{list}{$\bullet$}{
    \itemindent -.05em
    \parsep 0pt plus 1pt
    \partopsep 0pt plus 1pt
    \topsep 2pt plus 2pt minus 2pt
    \itemsep 0pt plus 1.3pt
    \parskip 0pt plus 2pt
    \leftmargin 0.13in}
      }
{\normalsize
\end{list}
}

\usepackage{url}
\usepackage{hyperref}

\usepackage[nolinenum]{icml2026}

\usepackage[capitalize,noabbrev]{cleveref}

\theoremstyle{plain}
\newtheorem{theorem}{Theorem}[section]
\newtheorem{proposition}[theorem]{Proposition}

\theoremstyle{definition}

\theoremstyle{remark}

\definecolor{darkgrey}{HTML}{434343}
\definecolor{lightgray}{gray}{0.9}

\usepackage[disable,textsize=tiny]{todonotes}

\icmltitlerunning{Depth Charge: Jailbreak Large Language Models from Deep Safety Attention Heads}

\begin{document}

\twocolumn[
  \icmltitle{Depth Charge: Jailbreak Large Language Models from \\ Deep Safety Attention Heads}
  
  \icmlsetsymbol{equal}{*}
  
  \begin{icmlauthorlist}
    \icmlauthor{Jinman Wu}{xidian}
    \icmlauthor{Yi Xie}{tsu}
    \icmlauthor{Shiqian Zhao}{nyu}
    \icmlauthor{Xiaofeng Chen}{xidian}
  \end{icmlauthorlist}

  \icmlaffiliation{xidian}{Xidian University, Xi'an, China}
  \icmlaffiliation{tsu}{Tsinghua University, Shenzhen, China}
  \icmlaffiliation{nyu}{Nanyang Technological University, Singapore}

  \icmlcorrespondingauthor{Xiaofeng Chen}{sce-wjm@stu.xidian.edu.cn}

  \icmlkeywords{Machine Learning, ICML, Large Language Models, Jailbreak, Safety Attention}

  \vskip 0.3in
]

\printAffiliationsAndNotice{}
\input{sections/0-abstract}

\input{sections/1-introduction}
\input{sections/2-related}
\input{sections/3-Preliminaries}
\input{sections/4-method}
\input{sections/5-experiment}

\input{sections/6-conclusion}

\input{sections/10-limitations_ethics}
\clearpage

\bibliography{sections/11-reference}
\bibliographystyle{icml2026}

\newpage
\appendix
\onecolumn 
\input{sections/8-appendix}


\end{document}

%% file: sections/0-abstract.tex
\begin{abstract}
Currently, open-sourced large language models (OSLLMs) have demonstrated remarkable generative performance. However, as their structure and weights are made public, they are exposed to jailbreak attacks even after alignment. Existing attacks operate primarily at shallow levels, such as the prompt or embedding level, and often fail to expose vulnerabilities rooted in deeper model components, which creates a false sense of security for successful defense. In this paper, we propose \textbf{\underline{S}}afety \textbf{\underline{A}}ttention \textbf{\underline{H}}ead \textbf{\underline{A}}ttack (\textbf{SAHA}), an attention-head-level jailbreak framework that explores the vulnerability in deeper but insufficiently aligned attention heads. SAHA contains two novel designs. 
Firstly, we reveal that deeper attention layers introduce more vulnerability against jailbreak attacks. 
Based on this finding, \textbf{SAHA} introduces \textit{Ablation-Impact Ranking} head selection strategy to effectively locate the most vital layer for unsafe output. 
Secondly, we introduce a boundary-aware perturbation method, \textit{i.e. Layer-Wise Perturbation}, to probe the generation of unsafe content with minimal perturbation to the attention. This constrained perturbation guarantees higher semantic relevance with the target intent while ensuring evasion. 
Extensive experiments show the superiority of our method: SAHA improves ASR by 14\% over SOTA baselines, revealing the vulnerability of the attack surface on the attention head. 
Our code is available at \url{https://anonymous.4open.science/r/SAHA}.

\end{abstract}

%% file: sections/1-introduction.tex
\section{Introduction}
\label{sec:intro}

Lately, open-sourced large language models (OSLLMs), including Llama~\cite{32llama3} and Qwen~\cite{30qwen1.5}, have demonstrated remarkable capabilities in text generation tasks~\cite{1Fei25,2Guo25,3Hu25,4Sinha25}. However, as their model structure and weights are all accessible to all parties, they are exposed to jailbreak attacks~\cite{5andriushchenko25,6chen25,7zhao-etal-2025-sql,8huang-etal-2025-efficient, 9Ahi,10singh,11park,12panaitescu-liess-etal-2025-poisonedparrot}, where malicious users can leverage these OSLLMs for unsafe content generation, \textit{e.g.,} nudity and sexuality. 
These realistic threat makes pre-release penetration testing (\textit{e.g.} jailbreak) a must to find the vulnerabilities of open-sourced LLMs, which provide an alignment corpus for compensating these safety holes. 

Existing jailbreak attacks can be classified into two categories according to attack surface: \textit{prompt-level} and \textit{embedding-level attacks}. Specifically, prompt-level attacks~\cite{17DBLP,18saiem} manipulate the input by searching for adversarial prompts with gradient or LLMs~\cite{61wang2025comprehensive}. For example, GCG~\cite{60GCG} utilizes greedy search to find adversarial substitutes in the token space, and PAIR\cite{59PAIR} leverages the strong generation ability of LLMs to craft jailbreak prompts. More deeply, embedding-level attacks~\cite{19ha-etal-2025-one,20li-etal-2025-revisiting,21li-etal-2025-cavgan} operate in the latent continuous space, directly manipulating the embedding for unsafe content generation. For instance, SCAV~\cite{25DBLP:conf/nips/XuHCW24} conducts embedding-level attacks by introducing targeted perturbations with automatically selected hyperparameters. 


Despite the great attack performance achieved, existing jailbreak methods fail to show resistance: they launch attacks from a shallow attack surface and are easily tackled by superficial alignment. For example, \cite{68cao} shows that random token dropping-based robust alignment check could defend prompt-level attacks, and \cite{69zhao} shows that layer-specific editing of semantic representations can easily counter embedding-level attacks. 
The common reason for these attacks' failure to resist mitigation strategies is that they focus on manipulating the shallow layers of OSLLMs, and thus are brittle to simple shallow-level safety alignment~\cite{13DBLP:conf/acl/ChenLSH0SH25,14zhang-etal-2025-defense,15yi2025probe,15-1zhang-etal-2024-defending}. 
This kind of penetration testing and remediation provides an appearance of security, yet it overlooks a question that is worth reflecting on: \textit{Are the open-sourced LLMs safe against attacks that are launched from the deeper layers?}

Our answer is negative: we identified the vulnerability at the overlooked attention-head level, which can be used to mount highly effective jailbreaks, revealing a critical vulnerability in deployed OSLLMs. To thoroughly study this flaw, in this paper, we propose \textbf{\emph{Safety Attention Head Attack} (SAHA)}, a novel jailbreak from the new attack surface-attention head. In general, \textbf{SAHA} mainly contains two modules. 
First, we propose \textbf{Ablation-Impact Ranking (AIR)}, a strategy that identifies the safety-related attention heads. We introduce a safety classifier, which is trained on the attention activation, and monitor the loss when ablating the attention head. This loss variation provides a signal for locating the important safety-related attention head. 
Second, we introduce \textbf{Layer-Wise Perturbation (LWP)} for manipulating attention heads, which strategically allocates the perturbation budget to the most critical head within each layer. 
We generate the perturbation vector with a closed-form solution derived from the linearized decision boundary of the safety probe, 
yielding a perturbation of minimal magnitude that flips the safety label.
Such a practice makes our method computationally efficient and can be scaled to larger language models. 



We extensively evaluate SAHA across popular safety-aligned OSLLMs (including Llama3.1, Qwen, and DeepSeek) and compare against seven state-of-the-art baselines. The experimental results show that SAHA \textbf{consistently outperforms all baselines, including both prompt-level and embedding-level ones by substantial margin}, in terms of both attack success rate and semantic relevance. 
This indicates that attention level possesses a higher safety risk than shallower-level attack surfaces. 
Beyond performance surpass, we conduct thorough analysis to delve deeper into the mechanism via ablation studies on locating and perturbing strategies and statistical testing.
We also validate that our method is robust even under low perturbation budgets, and consistently show effectiveness against composite defenses. 
The superior attack performance and resistance to existing defenses raise the alarm for OSLLMs and highlight the urgent need for more effective safeguards.

In summary, we make the following contributions:
\begin{icompact}
	\item We identify the limitations of existing jailbreak attacks, which launch attacks from a shallow level and are thus easily defended by safety alignment. 
    \item We investigate the vulnerability of OSLLMs from a deeper perspective and propose a novel safety attention-head attack to exploit its vulnerability. 
    \item We propose AIR attention head selection strategy and LWP perturbation method to effectively locate and probe the model for unsafe content generation.
    \item We conducted extensive experiments to validate our method. Results on Llama and Qwen show that it largely outperforms existing methods, shedding light on the importance of deeper safety for OSLLMs.
\end{icompact}

%% file: sections/2-related.tex
\section{Related Work}
\label{sec:related_work}

\subsection{Jailbreak Attacks}

Existing jailbreak attacks can be divided into two types based on their attack surfaces, \textit{i.e.}, \textit{prompt-level attacks} and \textit{embedding-level attacks}. 

\textbf{Prompt-level Attack}. The prompt-level attacks craft adversarial tokens or templates at the input of LLMs to elicit malicious outputs~\cite{36mu-etal-2025-stealthy,37ding-etal-2024-wolf,38ramesh2024gpt,39yuan2024gpt}. 
For example, the gradient-based method GCG adopts greedy search to find the substitute of an unsafe word in the token dictionary~\cite{60GCG}, and PAIR uses an attacker LLM to find the revision of the blocked prompt that could bypass the safety boundary of the victim OSLLMs~\cite{59PAIR}. 
More recently, AutoDAN-Turbo designs a pipeline that can automatically find the effective strategies, which provide guidance for generating adversarial prompts~\cite{57autodanturbo}. 
As these methods manipulate only the model's input, they show less effectiveness when simple alignment is imposed on the encoder. For example, 
Liu et al.\cite{70liu2024formalizing} formalize and benchmark prompt-injection attacks and show that encoder-targeted defenses can substantially reduce attack success. Li et al.\cite{71li2024evaluating} construct a benchmark demonstrating that instruction-tuned LLMs are vulnerable to injected instructions in retrieved context, while retrieval- and encoder-side mitigations lower exploitability.

\textbf{Embedding-level Attack}. Embedding-level attacks manipulate the continuous hidden representations or optimize soft prompts~\cite{40schwinn2024soft,41hase2024smoothed,42zhou-etal-2024-alignment,43zhou} at the more covert embedding level. 
Recent efforts have further conceptualized ``steering vectors'' in activation space to directly steer OSLLMs~\cite{44o'brien2025steering,45scalena-etal-2024-multi,46nguyen-etal-2025-multi,47deng-etal-2025-unveiling} for unsafe content generation. 
For example, SCAV~\cite{25DBLP:conf/nips/XuHCW24} utilizes a concept activation vector (CAV) as a steering vector to perturb embeddings, and CAA~\cite{28rimsky-etal-2024-steering} derives the mean difference of contrastive activations as a steering vector. 
This line of work demonstrates that directly altering a model's activations is more effective and stealthier, as it functions at a deeper level than the prompt-level attacks. 
Despite their better attack performance and defense resistance, they are still fragile against safety alignment that touches the embedding level. 
For instance, 
Arditi et al.~\cite{72arditi2024refusal} show that a single residual-stream “refusal” direction mediates safety—ablating this direction disables refusal, while amplifying it induces refusal, implying that embedding interventions can block activation steering.
Yu et al.~\cite{73yu2024refusal} propose ReFAT and demonstrate that strengthening/purifying the refusal feature in hidden representations substantially reduces success rates of both activation- and input-level jailbreaks.

\textbf{Our Attack}. From prompt-level to embedding-level attack, as the attack surface penetrates deeper into the model, the rule is clear: the jailbreak attacks tend to become more effective. This is caused by the limited depth touched by the existing safety alignment. Motivated by this, in this work, we investigate the deeper level of OSLLMs, \textit{i.e.} the attention-head level, to verify the vulnerability of LLM when facing attacks from these blind spots for safety alignment. 

\subsection{Attention Head}
Previous works validate that LLMs' behaviors are governed by a sparse set of attention heads~\cite{48marks2025sparse,49frankle2018the}. 
They proved that these heads are redundant, while a critical few play the most important role for the specific functions (\textit{e.g.}, syntactic parsing, factual recall) by studying \textbf{attention head importance}\cite{50voita-etal-2019-analyzing,51kovaleva-etal-2019-revealing}. 
For example, PASTA\cite{66zhang2023tell} selects the intersection of top-k heads as the ``global effective attention head set'' by evaluating the steering performance of each attention head across multiple tasks, and guides the model to focus on key information by adjusting their attention scores. 
Han et al.~\cite{65han2025heads} screens out collaborative attention heads via the selective attention head masking method, organizing them into a ``functional pathway''. 
Sahara~\cite{79sahara} utilizes a heuristic attribution algorithm to identify safety-critical attention heads by measuring the representational shifts caused by targeted attention head ablations.
By retaining heads within the pathway and masking out the ones outside the pathway, it enables guiding LLMs to automatically exhibit specific task functionalities. 

Inspired by the controllability of attention heads, in this paper, we study this important but unnoticed question: \textit{whether it's possible to probe the OSLLMs to generate unsafe content by pruning a few attention heads only?} 
Based on this insight, we make the first attempt to investigate OSLLM vulnerability at a ``deeper'' mechanistic level—targeting individual attention heads as structural, atomic causal units rather than merely deeper topological layers or the aggregated residual stream. 

%% file: sections/3-Preliminaries.tex
\section{Preliminaries}

\subsection{Architectural Foundation and Attention Mechanism}
Contemporary open-source large language models are typically implemented as decoder-only Transformers and trained in an autoregressive manner. Given a prefix, the model predicts a distribution over the next token by projecting the final hidden state at the last position through a linear layer, followed by a softmax normalization. Training maximizes the likelihood of the observed next token conditioned on its preceding context, which enables the model to learn rich contextual representations.

Within each Transformer layer, multi-head attention~\cite{24vaswani} serves as the primary mechanism for contextual aggregation. For an input sequence representation, the model computes attention scores by scaled dot products between queries and keys, normalizes them with a softmax function, and uses the resulting weights to aggregate the corresponding values. Outputs from multiple attention heads, each equipped with separate query, key, and value projections, are concatenated and linearly transformed to form the layer output. This head-wise structure is particularly relevant to our analysis, since it provides a natural unit for fine-grained attribution and intervention.

\subsection{Safeguarding Mechanisms and Adversarial Subversion}
To reduce the risk of harmful or unethical generation, open-source large language models are usually aligned with safety objectives before deployment. Common alignment procedures include supervised fine-tuning on safe responses, reinforcement learning from human feedback, and constitutional AI\cite{63ouyang2022training,64bai2022constitutionalaiharmlessnessai}. The intended effect is to bias the model toward rejecting unsafe prompts, for example by assigning high probability to refusal templates when faced with harmful queries.

Nevertheless, jailbreak attacks aim to bypass these safeguards and elicit disallowed outputs. This issue is particularly pronounced in the open-source setting, where access to model weights and architectural details facilitates stronger adversarial strategies. Existing jailbreak methods often exploit weaknesses in prompt design or internal representations, suggesting that alignment may mainly alter superficial behavioral patterns rather than eliminate deeper vulnerabilities. Consequently, standard evaluation protocols may overestimate the robustness of safety mechanisms, as they do not fully capture the model's susceptibility to adversarial manipulation \cite{74NEURIPS2023_fd661313,75NEURIPS2023_ac662d74,76ICLR2024_83b7da3e}.

\begin{figure*}[ht!]
  \centering
  \adjustbox{trim=0 0 0 8, clip}{
    \includegraphics[width=0.98\linewidth]{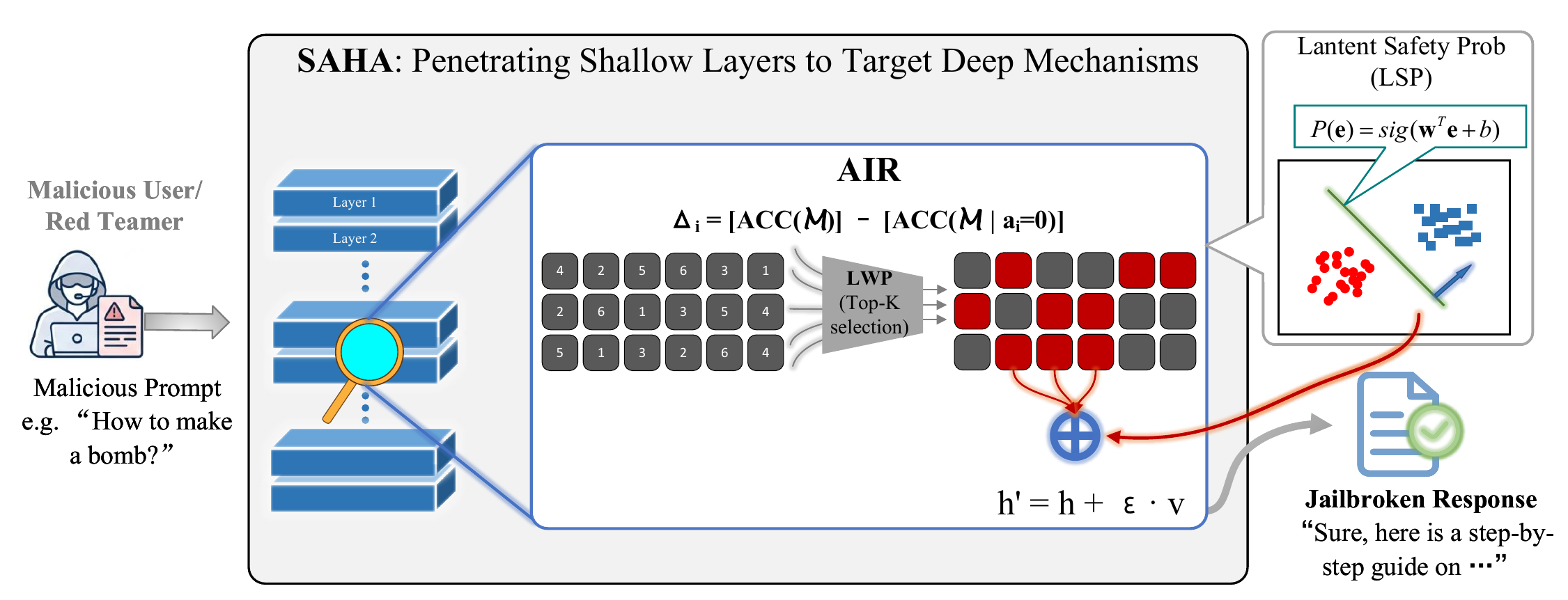}
  }
  \vspace{-0.1in} 
  \caption{Overview of SAHA. Given a victim LLM, we derive a steering vector that captures the linear separability between benign and malicious samples. In the selection stage, AIR identifies the attention heads whose ablation most degrades safety; in the perturbation stage, LWP allocates a layer-wise budget and ranks heads within each layer. The steering vector for each chosen head is scaled by its assigned perturbation magnitude and injected into the model, yielding the final perturbed embeddings that drive the jailbreak.} 
  \label{fig:overview}
  \vspace{-0.1in} 
\end{figure*}


\subsection{Problem formulation}
We formulate the jailbreak as a constrained optimization against a pretrained language model \(\mathcal{M}_{\theta}\). The adversary's ultimate goal is to induce the model to produce a response \(y\) that is semantically aligned with a malicious target prompt \(p_{t}\) and is judged unsafe by an external ground-truth judge \(\mathcal{J}\), that is \(\mathcal{J}(y)=0\). Directly presenting the malicious embedding \(\mathbf{e}\in\mathbb{R}^{d}\) typically activates the model's intrinsic safety refusal; we abstract this internal mechanism as a classifier \(f_{\mathrm{cls}}\) for which \(f_{\mathrm{cls}}(\mathbf{x})=1\) denotes an input perceived as safe. To bypass this refusal the attacker injects an additive perturbation \(\mathbf{p}\in\mathbb{R}^{d}\), producing the adversarial embedding \(\mathbf{e}'=\mathbf{e}+\mathbf{p}\). The core intuition is to force the internal classifier to misclassify the perturbed embedding as benign, thereby preventing the defensive trigger and allowing the model to produce the intended content. Balancing stealth and semantic fidelity, we therefore choose \(\mathbf{p}\) to maximize a utility measure of semantic alignment under a perturbation budget while ensuring the internal classifier accepts the perturbed input and the external judge labels the resulting output as unsafe. The attack objective is
\begin{equation}
\begin{aligned}
&\underset{\|\mathbf{p}\|\le\epsilon}{\mathrm{maximize}}\quad && U\big(\mathcal{M}_{\theta}(p;\,\mathbf{e}+\mathbf{p}),\,p_{t}\big)\\[4pt]
&\mathrm{subject\ to}\quad && f_{\mathrm{cls}}(\mathbf{e}+\mathbf{p})=1,\\[2pt]
& && \mathcal{J}\big(\mathcal{M}_{\theta}(p;\,\mathbf{e}+\mathbf{p})\big)=0,\\[2pt]
& && \|\mathbf{p}_{\mathcal{H}_{\mathrm{safety}}}\|\le\eta,
\end{aligned} \nonumber
\end{equation}
where \(\epsilon\) bounds the overall perturbation magnitude and \(\eta\ll\epsilon\) imposes a tighter limit on components of \(\mathbf{p}\) applied to attention heads strongly coupled to the safety mechanism.




%% file: sections/4-method.tex
\section{Methodology}
\label{sec:method}


\subsection{Overview}
\label{subsec:overview}
In this section, we present the design of our deeper-level jailbreak attack, termed \textbf{SAHA} (Figure~\ref{fig:overview}). The primary objective of \textbf{SAHA} is to identify attention heads that are critical to a model’s safety mechanisms and to deliberately manipulate their activations to induce unsafe generation. This objective raises two key challenges: (1) how to accurately localize safety-related attention heads within the transformer architecture, and (2) how to effectively probe these heads to trigger unsafe behavior while preserving semantic fidelity. 
We address these challenges via two components:
\begin{icompact}
	\item[(i)] \textbf{Ablation-Impact Ranking (AIR).} AIR identifies safety-critical attention heads using signals from a safety classifier. By selectively ablating individual heads and measuring the resulting degradation in classifier performance, AIR assigns an importance score to each head based on its causal contribution to safety. This ablation-based criterion enables fine-grained and reliable selection among the large number of attention heads.
	
	\item[(ii)] \textbf{Layer-Wise Perturbation (LWP).} After localizing safety-critical heads, LWP probes these heads by injecting optimized additive perturbations into their activations. Perturbation budgets are allocated in a layer-aware manner, allowing SAHA to concentrate its intervention on the most impactful layers while minimizing unintended degradation of output semantic fidelity.
\end{icompact}

In the following sections, we detail the design of our AIR and LWP strategies.

\subsection{Ablation-Impact Ranking}
\label{subsec:head_selection_localization}

This section introduces the \emph{Ablation-Impact Ranking} (AIR) strategy, which is designed to identify attention heads that are critical to model safety. AIR is motivated by the observation that safety mechanisms are often implemented implicitly through a small subset of internal components. By measuring the degradation in safety performance caused by selectively disabling individual heads, AIR enables fine-grained localization of safety-relevant attention heads.

\noindent\textbf{Safety Classifier Construction.}
We assume access to a trained safety classifier $f_{\mathrm{cls}}$ that operates on the internal representations of the target language model $\mathcal{M}$. Given a dataset $\mathcal{D}=\{(X_i, y_i)\}_{i=1}^M$, where $X_i$ denotes a model input and $y_i \in \{0,1\}$ indicates whether the corresponding output is safe, the classifier is trained to predict $y_i$ from the model’s hidden activations. This classifier serves as an external probe for evaluating the contribution of individual attention heads to safety behavior.

\noindent\textbf{Impact of Head Ablation.}
Let $\text{Acc}_{\mathrm{orig}}$ denote the classification accuracy of $f_{\mathrm{cls}}$ when applied to the unmodified model $\mathcal{M}$. For each attention head $i \in \{1,\ldots,N\}$, we construct an ablated variant of the model by zeroing out the output of that head:
\begin{equation}
\mathcal{M}_{(i)}(X) = \mathcal{M}(X) \mid \mathbf{a}_i = \mathbf{0},
\end{equation}
where $\mathbf{a}_i$ denotes the activation produced by head $i$. The ablated model is then evaluated using the same classifier to obtain accuracy $\text{Acc}_{(i)}$.

\noindent\textbf{Ablation-Impact Ranking (AIR).}
The importance of head $i$ for safety is quantified by the performance drop induced by its ablation:
\begin{equation}
\Delta_i = \text{Acc}_{\mathrm{orig}} - \text{Acc}_{(i)}.
\end{equation}
Larger values of $\Delta_i$ indicate that the corresponding head plays a more significant role in maintaining safe behavior. All attention heads are ranked in descending order according to $\Delta_i$, and the top-$k$ heads are selected:
\begin{equation}
\mathcal{S}_{\mathrm{AIR}} = \arg \operatorname{top}\text{-}k \left\{ \Delta_1, \Delta_2, \ldots, \Delta_N \right\}.
\end{equation}
Compared to gradient-based or heuristic selection methods, AIR directly measures causal impact on safety, enabling precise identification of safety-critical heads.

\noindent\textbf{Spatial Localization via Frequency Analysis.}
To robustly mitigate sensitivity to a particular sparsity level, we extend AIR with a frequency-based localization framework. Specifically, we repeat the AIR procedure under multiple selection ratios $\alpha \in \mathcal{A} = \{0.25, 0.3, \ldots, 1.0\}$. For each attention head indexed by layer $\ell$ and head position $h$, we compute its empirical selection frequency:
\begin{equation}
f_{\ell,h} = \frac{|\{\alpha \in \mathcal{A} : (\ell,h) \in \mathcal{S}_{\alpha}\}|}{|\mathcal{A}|}.
\end{equation}
This frequency captures the intrinsic importance of a head across different ablation budgets and is less sensitive to individual hyperparameter choices.

\noindent\textbf{Global Critical Head Set.}
Treating all attention heads across layers as a unified candidate pool, we define the final safety-critical head set as:
\begin{equation}
\mathcal{H}_{\mathrm{critical}} =
\arg \operatorname{top}\text{-}k
\left\{
f_{\ell,h} \;\middle|\;
\ell \in [1,L],\; h \in [1,H_\ell]
\right\},
\end{equation}
where $L$ is the total layers and $H_\ell$ is the total heads in layer $\ell$. This procedure yields both a ranked list of safety-relevant attention heads and their precise architectural locations.

\subsection{Layer-Wise Perturbation}
\label{subsec:perturbation_distribution_magnitude}

Having identified safety-critical attention heads using AIR, we next describe the \emph{Layer-Wise Perturbation} (LWP) strategy, which governs how adversarial perturbations are distributed across these heads. LWP is motivated by the hierarchical structure of transformer architectures and aims to maximize safety degradation while minimizing distortion to overall model behavior.

\noindent\textbf{Layer-Wise Perturbation Strategy.}
Rather than allocating a single global perturbation budget, LWP enforces independent budget constraints at each layer. For a given selection ratio $\alpha \in (0,1]$ and a layer $\ell$ containing $n_\ell$ attention heads, the number of perturbed heads is defined as:
\begin{equation}
k_\ell = \left\lfloor \alpha \cdot n_\ell \right\rfloor.
\end{equation}
Within each layer, attention heads are ranked according to AIR importance scores, and the top-$k_\ell$ heads are selected:
\begin{equation}
\mathcal{P}_\ell =
\arg \operatorname{top}\text{-}k_\ell
\left\{ \text{score}(i) \mid i \in \mathcal{S}_\ell \right\}.
\end{equation}
The final perturbation set is obtained by aggregating selections across all layers:
\begin{equation}
\mathcal{P}_{\mathrm{LWP}} = \bigcup_{\ell=1}^{L} \mathcal{P}_\ell.
\end{equation}
This layer-aware allocation ensures that perturbations are distributed throughout the network depth, preventing over-concentration in shallow or deep layers.

\noindent\textbf{Perturbation Formulation.}
Perturbations are applied additively to the activations of the selected heads. Let $\mathbf{e}$ denote the original concatenated activation vector, and let $\mathbf{v} \in \mathbb{R}^{N \cdot d_a}$ be a unit perturbation direction with support restricted to the selected head set $\mathcal{S}$, i.e., $\operatorname{supp}(\mathbf{v}) = \mathcal{S}$. The perturbed representation is given by:
\begin{equation}
\widetilde{\mathbf{e}} = \mathbf{e} + \epsilon \mathbf{v},
\end{equation}
where $\epsilon$ controls the perturbation magnitude.

\noindent\textbf{Minimal Perturbation Analysis.}
Let $f_{\mathrm{cls}}$ be the trained linear safety classifier with parameters $(\mathbf{w}, b)$. A successful attack corresponds to driving the classifier loss below a threshold $L_0$, which represents the decision boundary:
\begin{equation}
\mathcal{L}_{\mathrm{CE}}\!\left(f_{\mathrm{cls}}(\widetilde{\mathbf{e}}), y=1\right) < L_0.
\end{equation}
As established in prior work \cite{27zhang,26Zhang,25DBLP:conf/nips/XuHCW24}, the linearity of the classifier allows the minimal required perturbation magnitude to be expressed in closed form. For a target probability $P_0 = e^{-L_0}$, the sufficient condition is:
\begin{equation}
\label{eq:epsilon-bound}
\epsilon \ge
\frac{
\log\!\left(\frac{P_0}{1-P_0}\right)
-
(\mathbf{w}^\top \mathbf{e} + b)
}{
\mathbf{w}^\top \mathbf{v}
}.
\end{equation}
To minimize $\epsilon$, the optimal perturbation direction aligns with the classifier’s weights projected onto the selected subspace:
\begin{equation}
\mathbf{v} =
\frac{\mathbf{w}_{\mathcal{S}}}{\|\mathbf{w}_{\mathcal{S}}\|},
\end{equation}
where $\mathbf{w}_{\mathcal{S}}$ denotes the restriction of $\mathbf{w}$ to the coordinates indexed by $\mathcal{S}$. Full derivation in Appendix~\ref{app:minimal perturbation}.

%% file: sections/5-experiment.tex
\section{Experiment}
\label{sec:core-q3}

\subsection{Experiment Setup}

\noindent \textbf{Datasets.} We consider two popular benchmarks used in previous works, \textit{i.e}, JailbreakBench~\cite{34chao2024jailbreakbench} and MaliciousInstruct~\cite{35huang2024catastrophic}. 
JailbreakBench provides a standardized dataset of 100 distinct harmful behaviors covering a broad spectrum of misuse scenarios, designed to facilitate reproducible evaluation of adversarial attacks. 
MaliciousInstruct comprises 100 harmful instructions evenly distributed across ten specific malicious intent categories, such as sabotage, theft, and hacking, enabling a comprehensive assessment of model robustness against diverse threats.
This benchmark is designed to enable comprehensive evaluation of approach adaptability and effectiveness through a broader range of malicious instructions. 

\noindent \textbf{Victim Models.} We evaluate our approach on three widely used open-sourced LLMs: Qwen1.5-7B-Chat \cite{30qwen1.5}, Llama3.1-8B-Instruct \cite{32llama3}, and Deepseek-LLM-7B-Chat \cite{33Deepseek}. All three models are instruction-tuned and aligned prior to release, making them representative targets for studying safety mechanisms in modern LLMs. Architecturally, all victim models adopt a decoder-only Transformer backbone with multi-head self-attention as the core computational component, while differing in model scale, attention head configurations, and alignment strategies. These shared architectural foundations, together with their design variations, make the selected models well-suited for evaluating the generality and robustness of our method.

\noindent \textbf{Baselines.} We extensively consider 7 baselines, including PAIR~\cite{59PAIR}, GCG~\cite{60GCG}, AutoDAN~\cite{58autodan}, AutoDAN-Turbo~\cite{57autodanturbo}, SCAV~\cite{25DBLP:conf/nips/XuHCW24}, CAA~\cite{28rimsky-etal-2024-steering}, ConVA~\cite{29jin-etal-2025-internal}, from both prompt-level and embedding-level attacks. 
Among them, PAIR deploys an attacker LLM to rework blocked prompts, letting them slip past the safety guardrails of victim OSLLMs.
GCG and AutoDAN craft attack prompts by tapping into LLMs’ output logit distributions and gradients. 
SCAV conducts attacks by leveraging concept activation vectors at the layer level. CAA modulates LLM behaviors by computing the mean difference of contrastive activations. ConVA employs context-controlled value vectors and a gating mechanism to perturb embeddings for LLM internal value alignment.

\noindent \textbf{Metrics.} We use two metrics to evaluate the performance: \textit{Attack Success Rate (ASR)} and \textit{BERTScore}. We assess the safety of the output using the Judge model Llama-Guard-3-8B~\cite{77inan2023llama}, and compute ASR as the proportion of output that is identified as unsafe by the Judge. 
For evaluating task completion, we adopt the BERTScore, which is evaluated on Sentence-BERT\cite{78reimers2019sentence}, to measure the semantic similarity. BERTScore quantifies the relevance between a prompt and a model-generated response by computing semantic similarity over their contextual (BERT) embeddings and aggregating token-level similarity scores using the arithmetic mean to produce a single relevance score. 

\subsection{Main Results}

Table~\ref{tab:asr_main} presents a comprehensive comparison of jailbreak effectiveness across three safety-aligned open-source LLMs. Overall, SAHA consistently achieves the strongest performance in terms of both ASR and semantic fidelity, as measured by BERTScore, substantially outperforming prompt-level baselines (PAIR, GCG, AutoDAN, and AutoDAN-Turbo) as well as embedding-level methods (SCAV, CAA, and ConVA). For instance, SAHA attains an ASR/BERTScore of \textbf{0.85/0.76} on Llama3.1 evaluated on JailbreakBench, and \textbf{0.86/0.81} on Qwen1.5 evaluated on MaliciousInstruct. These results demonstrate that SAHA is able to achieve high jailbreak success rates while preserving strong semantic coherence in outputs. 

\begin{table*}[h]
  \centering
  \caption{Comparison with baselines on JailbreakBench and MaliciousInstruct. ASR is evaluated by Llama-Guard-3-8B, and BERTScore is evaluated by SentenceBert. The best results are in \textbf{bold} and the second best are \underline{underlined}.}
  \vspace{-0.1in}
  \label{tab:asr_main}
  \scriptsize
  \resizebox*{\linewidth}{!}{
  \begin{tabular}{lccccccc}
    \toprule
    \multirow{3}{*}{\textbf{Methods}} & \multicolumn{2}{c}{Llama3.1-8B-Instruct}
    & \multicolumn{2}{c}{Qwen1.5-7B-Chat}
    & \multicolumn{2}{c}{Deepseek-llm-7B-Chat} \\
    \cmidrule(lr){2-3} \cmidrule(lr){4-5} \cmidrule(lr){6-7} 
    & \makecell{JailbreakBench \\ ASR $\uparrow$ / BERTScore$\uparrow$}
    & \makecell{MaliciousInstruct \\ ASR $\uparrow$ / BERTScore$\uparrow$}
    & \makecell{JailbreakBench \\ ASR $\uparrow$ / BERTScore$\uparrow$}
    & \makecell{MaliciousInstruct \\ ASR $\uparrow$ /BertScore$\uparrow$}
    & \makecell{JailbreakBench \\ ASR $\uparrow$ / BertScore$\uparrow$}
    & \makecell{MaliciousInstruct \\ ASR $\uparrow$ / BertScore$\uparrow$} \\
    \midrule

    \textbf{Based}
    & 0.06 / 0.62   & 0.01 / \textbf{0.84}
    & 0.06 / 0.74   & 0.03 / 0.66
    & 0.29 / 0.64 & 0.10 / \underline{0.79} \\
    \midrule
    \multicolumn{6}{l}{\emph{Prompt-level attacks}} \\

    \textbf{PAIR}
    & \underline{0.57} / 0.07  & 0.6 / 0.06
    & 0.1 / 0.23   & 0.03 / 0.4
    & 0.3 / 0.59 & 0.3 / 0.51 \\

    \textbf{GCG}
    & 0.10 / 0.29 & 0.03 / 0.69
    & 0.47 / 0.57 & 0.60 / 0.66
    & 0.37 / 0.64 & 0.37 / 0.67  \\

    \textbf{AutoDAN}
    & 0.40 / 0.58 & 0.67 / 0.66
    & 0.57 / 0.62 & 0.83 / 0.64
    & 0.60 / 0.55 & 0.63 / 0.70\\

    \textbf{AutoDAN-Turbo}
    & 0.13 / 0.7 & 0.17 / 0.68
    & 0.63 / \textbf{0.75} & 0.63 / 0.72
    & 0.13 / 0.11 & 0.20 / 0.19 \\
    \midrule
    \multicolumn{6}{l}{\emph{Embedding-level attacks}} \\

    \textbf{SCAV}
    & 0.55 / \textbf{0.77} & \underline{0.68} / 0.82
    & \underline{0.71} / 0.73 & 0.80 / \textbf{0.81}
    & \underline{0.85} / \textbf{0.71} & \underline{0.74} / \textbf{0.80}\\

    \textbf{CAA}
    & 0.14 / 0.35 & 0.22 / 0.50
    & 0.39 / 0.38 & 0.7 / 0.37
    & 0.37 / 0.15 & 0.34 / 0.07 \\

    \textbf{ConVA}
    & 0.19 / 0.39 & 0.09 / 0.46
    & 0.59 / 0.4 & \underline{0.85} / 0.4
    & 0.03 / 0.4  & 0 / 0 \\
    \midrule
    
     \cellcolor{lightgray}\textbf{SAHA} (Ours) 
    & \cellcolor{lightgray}\textbf{0.85} / \underline{0.76} & \cellcolor{lightgray}\textbf{0.87} / \textbf{0.84}
    & \cellcolor{lightgray}\textbf{0.82} / \textbf{0.75} & \cellcolor{lightgray}\textbf{0.86} / \textbf{0.81}
    & \cellcolor{lightgray}\textbf{0.91} / \underline{0.70} & \cellcolor{lightgray}\textbf{0.81} / 0.77 \\
    \bottomrule
  \end{tabular}
  }
  \end{table*}

The empirical results elucidate \textit{why} SAHA is effective. Prompt-level methods exhibit brittle and model-dependent behavior: while some achieve moderate ASR on specific architectures, they fail to generalize across models, suggesting a reliance on surface-level token manipulations that are readily mitigated by model-specific defenses. Embedding-level methods, by contrast, suffer from a systematic trade-off between attack success and semantic fidelity. Certain approaches increase automated ASR at the expense of substantially degraded BERTScore, whereas others preserve semantic relevance but achieve only limited ASR. In comparison, SAHA operates through targeted mechanistic intervention. By combining AIR to identify safety-critical attention heads with a LWP allocation strategy, SAHA directly intervenes at the loci of internal safety reasoning. This focused perturbation concentrates attack capacity where it most effectively disrupts safety-related computation, while minimizing unintended interference with semantic processing. As a result, SAHA achieves both high ASR and high BERTScore, with stable performance across diverse model architectures.

These findings have direct implications for alignment and defense design. The consistent vulnerability to head-level perturbations indicates that \textit{defenses limited to input inspection or shallow representations are insufficient}. Effective mitigation strategies should instead distribute safety mechanisms across the transformer’s internal computational pathways or explicitly monitor and reinforce attention heads identified as safety-critical. Moreover, the relatively low variance of SAHA’s performance across models suggests that the identified heads constitute shared computational substrates for safety reasoning. \textbf{Recognizing and leveraging these substrates offers a principled foundation for architecture-aware alignment and defense strategies.}

\subsection{Ablation Study}

To isolate the contributions of individual design choices in our proposed SAHA, we conduct ablation studies along three dimensions: the attention head localization strategy, the perturbation allocation strategy, and the retention ratio $\alpha$. In addition to the full SAHA configuration (AIR+LWP), we consider two alternative components: \emph{Accuracy-Probing Ranking} (APR) for head selection and \emph{Global-Wise Perturbation} (GWP) for perturbation allocation.

\begin{figure}[ht!] 
	\centering 
	\begin{subfigure}[b]{0.48\textwidth}
		\centering
		\includegraphics[width=\textwidth]{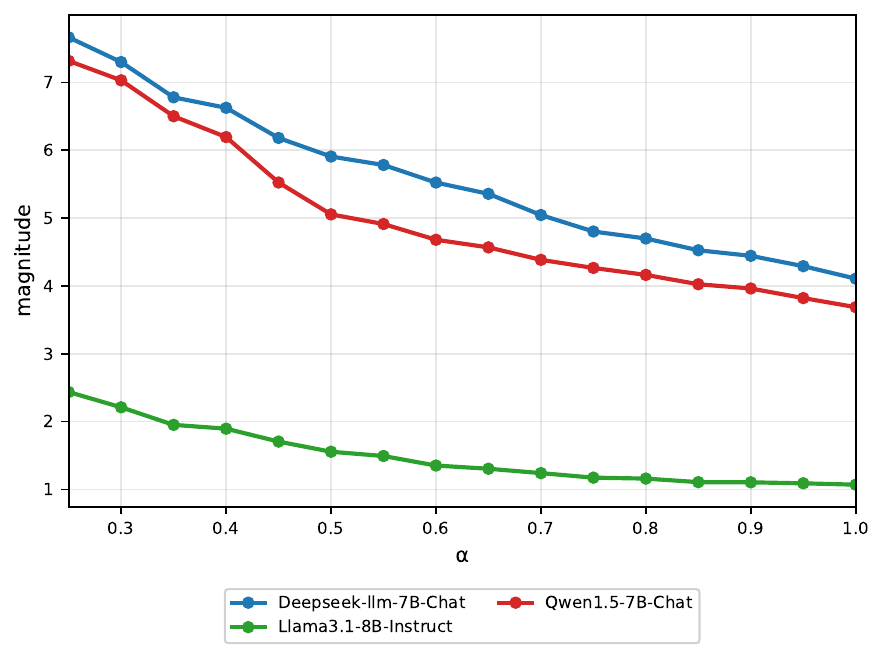} 
		\caption{JailbreakBench}
		\label{fig:sub_c}
	\end{subfigure}
    \vspace{0.5cm}
	\begin{subfigure}[b]{0.48\textwidth}
		\centering
		\includegraphics[width=\textwidth]{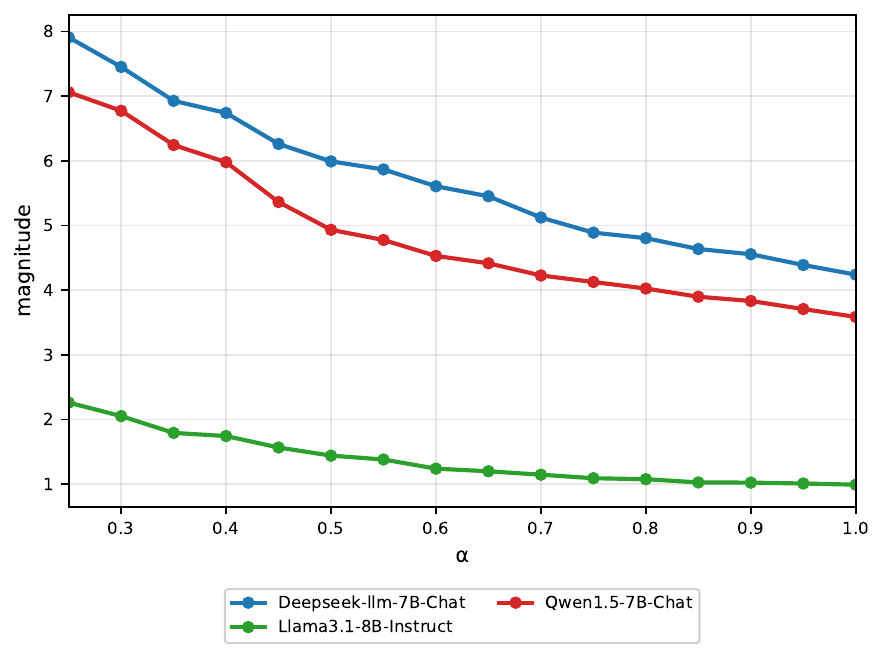}
		\caption{MaliciousInstruct}
		\label{fig:sub_d}
	\end{subfigure}
    \vspace{-25pt}
	\caption{\emph{Layer-averaged} perturbation magnitude \(\varepsilon(\alpha)\) plotted across layers for representative $\alpha$ values.}
	\label{fig:magnitude curves on JB} 
\end{figure}

\noindent\textbf{Impact of Head Localization Strategy.}
APR is an alternative head-ranking method measuring the standalone predictive power of individual heads via linear probing:
\begin{equation}
\mathcal{S}_{\text{APR}} =
\arg \operatorname{top}\text{-}k
\left\{ \text{Acc}_1, \text{Acc}_2, \ldots, \text{Acc}_N \right\}, \nonumber
\end{equation}
where $\text{Acc}_i$ denotes the accuracy of a linear classifier trained on the output of head $i$. While APR is computationally efficient and captures the correlational signal, it does not measure the causal impact of a head on safety behavior. In contrast, AIR explicitly quantifies the degradation in safety performance induced by head ablation. As shown in Table~\ref{tab:asr_ablation}, variants using APR consistently underperform those using AIR, even when paired with the same perturbation allocation strategy (LWP). Although APR+LWP achieves a competitive BERTScore in some settings, its lower ASR indicates that \textit{probing-based selection is less effective at identifying the mechanistic loci of safety reasoning than ablation-based ranking}. These results underscore the importance of causal head localization for reliable jailbreak success.

\begin{table*}[ht!]
  \centering
  \caption{Ablation Study of SAHA Strategy Combinations on JailbreakBench and MaliciousInstruct.}
  \vspace{-0.1in}
  \label{tab:asr_ablation} 
  \scriptsize
  \resizebox*{\linewidth}{!}{
  \begin{tabular}{ccccccc}
    \toprule
    \multirow{3}{*}{\textbf{Strategies}} & \multicolumn{2}{c}{Llama3.1-8B-Instruct}
    & \multicolumn{2}{c}{Qwen1.5-7B-Chat}
    & \multicolumn{2}{c}{Deepseek-llm-7B-Chat} \\
    \cmidrule(lr){2-3} \cmidrule(lr){4-5} \cmidrule(lr){6-7}
    & \makecell{JailbreakBench \\ ASR $\uparrow$ / BertScore $\uparrow$}
    & \makecell{MaliciousInstruct \\ ASR $\uparrow$ / BertScore $\uparrow$}
    & \makecell{JailbreakBench \\ ASR $\uparrow$ / BertScore $\uparrow$}
    & \makecell{MaliciousInstruct \\ ASR $\uparrow$ / BertScore $\uparrow$}
    & \makecell{JailbreakBench \\ ASR $\uparrow$ / BertScore $\uparrow$}
    & \makecell{MaliciousInstruct \\ ASR $\uparrow$ / BertScore $\uparrow$} \\
    \midrule

    \textbf{APR+LWP} & \underline{0.82} / \underline{0.75} & \underline{0.83} / \underline{0.80}
    & 0.75 / \underline{0.74} & 0.83 / \underline{0.81}
    & \underline{0.88} / \underline{0.71} & \textbf{0.82} / \underline{0.78} \\

    \textbf{APR+GWP} & 0.81 / \underline{0.75} & \underline{0.83} / 0.79
    & \underline{0.76} / \underline{0.74} & 0.83 / \underline{0.81}
    & \underline{0.88} / \underline{0.71} & 0.77 / \underline{0.78} \\

    \textbf{AIR+GWP} & 0.76 / 0.69 & 0.62 / 0.75
    & 0.75 / 0.69 & \underline{0.85} / 0.79
    & \underline{0.88} / 0.70 & 0.80 / \underline{0.78} \\

   \cellcolor{lightgray}\textbf{AIR+LWP (SAHA)} &\cellcolor{lightgray} \textbf{0.85} / \textbf{0.76} &\cellcolor{lightgray} \textbf{0.87} / \textbf{0.84}
    &\cellcolor{lightgray} \textbf{0.82} / \textbf{0.75} &\cellcolor{lightgray} \textbf{0.86} / \textbf{0.82}
    &\cellcolor{lightgray} \textbf{0.91} / \textbf{0.72} &\cellcolor{lightgray} \underline{0.81} / \textbf{0.79} \\
    \bottomrule
  \end{tabular}
  }
  \end{table*}

\noindent\textbf{Impact of Perturbation Allocation Strategy.}
To evaluate the effect of perturbation distribution, we replace the layer-wise allocation in LWP with a global selection scheme (GWP), which treats all attention heads as a single pool:
\begin{equation}
\mathcal{P}_{\text{GWP}} = \operatorname{TopK}_k(\mathcal{S}),
\quad
k = \lfloor \alpha \cdot N \rfloor. \nonumber
\end{equation}

While GWP simplifies optimization, it ignores the hierarchical structure of the transformer and may over-select heads from shallow layers at the expense of safety-critical heads in deeper layers. Empirically, replacing LWP with GWP (AIR+GWP) leads to degraded semantic fidelity and, in several cases, reduced ASR (e.g., on Qwen1.5). Figure~\ref{appfig:heatmap on JB-deepseek} --\ref{appfig:heatmap on MI-qwen}further illustrates that LWP distributes perturbation more evenly across layers, whereas global allocation concentrates perturbation disproportionately in early layers. These findings indicate that \textit{respecting layer-wise structure is crucial for concentrating perturbation on safety-relevant computation while minimizing semantic distortion}.

\noindent\textbf{Impact of Retention Ratio.} Figure~\ref{appfig: ASR curves of SAHA}--\ref{appfig: ASR curves} plots ASR against the retention ratio ($\alpha$) for JailbreakBench and MaliciousInstruct. ASR increases monotonically with $\alpha$ for all models and strategies, but the rate of improvement is model-dependent. Exclude AIR+GWP, Qwen1.5, and Deepseek start from relatively high baselines and improve gradually toward saturation as $\alpha$ grows. Llama3.1 exhibits a pronounced threshold near ($\alpha\approx0.45$): increases below this point yield limited gains, whereas exceeding the threshold produces a rapid rise to parity with the other models. Throughout the sweep, AIR+LWP achieves the highest ASR and the steadiest ascent. These trends imply two practical conclusions: increasing $\alpha$ reliably improves jailbreak success, and moderate coverage ($\alpha\in[0.5,0.7]$) is sufficient to approach peak performance for Qwen and Deepseek while Llama requires surpassing a model-specific threshold. The consistent advantage of AIR+LWP across retention regimes motivates its use as SAHA’s default configuration.


\subsection{Analysis}
\label{subsec:cq3}

\noindent \noindent\textbf{Empirical Validation of Budget Redistribution.}

We empirically validate the budget-redistribution effect by measuring the realized, layer-averaged perturbation magnitude ($\varepsilon(\alpha)$) as the effective number of perturbed attention heads is varied via the retention ratio ($\alpha$). Figure~\ref{fig:magnitude curves on JB} and Figure~\ref{appfig:magnitude curves on MI}reports $\varepsilon(\alpha)$ for the four strategies combinations and shows a monotonic increase in required per-head magnitude as $\alpha$ decreases, confirming the theoretical expectation that concentrating a fixed budget onto fewer heads requires stronger individual interventions to flip the safety classifier. Complementing this aggregate view, the ($L\times\alpha$) heatmaps of $R\_{\ell}(\alpha)$ in Figure~\ref{appfig:heatmap on JB-deepseek}--\ref{appfig:heatmap on MI-qwen} reveal how those increased magnitudes are redistributed across layers. The heatmaps expose distinct, strategy-dependent allocation regimes: AIR+GWP concentrates large-magnitude perturbations in deeper layers as $\alpha$ shrinks, reflecting GWP’s global selection of a few high-impact heads together with AIR’s deeper-head ranking, whereas other configurations produce different layer-wise profiles. These allocation modes are consistent across architectures and scales, indicating that the redistribution behaviors induced by our head-selection and allocation choices are largely model-agnostic and robust.


\noindent \noindent\textbf{Safety-Critical Head Localization Analysis.}
Figure~\ref{fig:layer_head_heatmap} reports per-head selection frequencies, which measure how often each attention head is flagged as vulnerable and thus its contribution to the model’s safety behavior. A consistent signal across models is the prominence of the final head, Head 31, suggesting a common role in aggregating safety-relevant features just before decoding. Model-specific patterns diverge: Deepseek shows a broad concentration of critical heads in middle and upper layers; Llama3.1 localizes critical heads around layer 7; Qwen1.5 concentrates around layer 5 with additional low-layer activations. \textit{These differences imply heterogeneous internal routing of safety signals and motivate layer-aware, model-specific hardening rather than uniform or purely shallow defenses.}

%% file: sections/6-conclusion.tex
\section{Conclusion}
\label{sec:conclusion}

In this work, we introduce SAHA, a novel safety attention-head attack framework designed to uncover critical vulnerabilities within the deep architecture of safety-aligned OSLLMs. By integrating AIR head-selection strategy with LWP perturbation methods, SAHA demonstrates that targeted perturbations on a minimal subset of attention heads can effectively bypass safety guardrails while preserving the semantic coherence of outputs. Extensive experiments on prominent OSLLMs, including Llama and Qwen, reveal that SAHA consistently outperforms established baselines, with the AIR+LWP variant achieving the highest attack success rates. These findings underscore the urgent need for robust alignment techniques that defend against vulnerabilities rooted in the model's deeper representational layers.



%% file: sections/10-limitations_ethics.tex
\section*{Limitations and Ethics Statements }

\paragraph{Limitations.} While SAHA demonstrates superior performance in revealing safety vulnerabilities, our framework operates under a white-box assumption, requiring access to the model's internal attention mechanisms and gradients. Although this setting restricts direct applicability to black-box commercial APIs, it is intentionally designed for model developers and safety researchers to conduct thorough pre-release red-teaming and mechanistic analysis of open-source weights. Furthermore, our current investigation is tailored to Transformer-based architectures, leveraging the specific interpretability of attention heads; as the field evolves toward alternative architectures (e.g., state-space models), the manifestation of safety-critical components may shift, necessitating potential adaptations of our head selection strategy.

\paragraph{Ethics Statements.} The primary objective of this work is to advance the safety alignment of Large Language Models by exposing overlooked vulnerabilities in deeper network layers. By demonstrating how safety mechanisms can be bypassed through targeted attention head perturbations, we provide the community with a rigorous "stress-test" tool to identify and harden architectural blind spots before deployment. We strictly adhere to ethical research standards: our experiments utilize established safety benchmarks, and we refrain from generating or disseminating harmful content beyond the scope necessary for academic evaluation. We believe that transparently revealing these mechanistic weaknesses is a prerequisite for developing more robust, verifiable, and secure AI systems.

%% file: sections/8-appendix.tex
\section*{Appendix}

\subsection*{A Derivation of Minimal Perturbation}
\label{app:minimal perturbation}

In this section, we derive the optimal perturbation required to bypass the safety classifier. based on SCAV~\cite{25DBLP:conf/nips/XuHCW24}, we extend the formulation to incorporate a target-class loss perspective and, crucially, a sparsity constraint on the perturbation direction.

\subsubsection*{A.1 Problem Setup via Loss with a Target Threshold}

Let \(f_{\text{cls}}(\mathbf{e}) = \sigma(\mathbf{w}^\top \mathbf{e} + b)\) be the safety classifier. The cross-entropy loss for a target label \(y_t\) is:
\begin{equation}
	\mathcal{L}(\mathbf{e}) = -y_t \log(f_{\text{cls}}(\mathbf{e})) - (1 - y_t)\log(1 - f_{\text{cls}}(\mathbf{e})).
\end{equation}
Our goal is to perturb a malicious prompt to be classified as safe so that the desired target label for each prompt is \(y_t=1\) whether it is safe or malicious. 
\begin{equation}
	\mathcal{L}(\mathbf{e}) = - \log(f_{\text{cls}}(\mathbf{e}))\leq L_0.
\end{equation}
This corresponds to a target loss threshold \(L_0\), which is the cross-entropy loss value when the predicted probability is exactly \(P_0\) for the target class \(y_t=1\):
\begin{equation}
	L_0 = -\log(P_0).
\end{equation}
Here, \(P_0\) is a hyperparameter. A larger \(P_0\) represents a more confident attack success. Our objective is to find the smallest perturbation magnitude \(|\epsilon|\) and the optimal sparse direction \(\mathbf{v}\) (\(\|\mathbf{v}\|=1\), \(\operatorname{supp}(\mathbf{v}) = \mathcal{S}\)) that reduces the loss below this threshold:
\begin{equation}
	\label{eq:loss-constraint}
	\mathcal{L}(\mathbf{e} + \epsilon \mathbf{v}) \leq L_0.
\end{equation}

\subsubsection*{A.2 Linearization of the Loss}

We approximate the loss at the perturbed point using a first-order Taylor expansion:
\begin{equation}
\mathcal{L}(\mathbf{e} + \epsilon \mathbf{v}) \approx \mathcal{L}(\mathbf{e}) + \epsilon \cdot \nabla_{\mathbf{e}} \mathcal{L}(\mathbf{e})^{\top} \mathbf{v}.
\end{equation}
The gradient of the cross-entropy loss with respect to the embedding is:
\begin{equation}
\nabla_{\mathbf{e}} \mathcal{L}(\mathbf{e}) = \left(f_{\text{cls}}(\mathbf{e}) - 1\right) \cdot \mathbf{w}.
\end{equation}
Substituting into the approximation:
\begin{equation}
\mathcal{L}(\mathbf{e}) + \epsilon \cdot \left(f_{\text{cls}}(\mathbf{e}) - 1\right) \cdot \mathbf{w}^{\top} \mathbf{v} \leq L_0.
\end{equation}

\subsubsection*{A.3 Solving the Optimization Problem}

Rearranging the inequality:
\begin{equation}
\epsilon \cdot \left(1 - f_{\text{cls}}(\mathbf{e})\right) \cdot \mathbf{w}^{\top} \mathbf{v} \geq \mathcal{L}(\mathbf{e}) - L_0.
\end{equation}
For an input initially classified as malicious, \(\mathcal{L}(\mathbf{e}) > L_0\) (since \(P_{\text{initial}} < P_0\)), making the right-hand side \((\mathcal{L}(\mathbf{e})-L_0  ) > 0\). Furthermore, \((1-f_{\text{cls}}(\mathbf{e})) > 0\). Therefore, to minimize the perturbation magnitude \(|\epsilon|\) that satisfies the inequality, we must \textbf{maximize} the term \(\mathbf{w}^{\top} \mathbf{v}\). This leads to the optimization problem:
\begin{equation}
\begin{aligned}
	& \underset{\mathbf{v}}{\text{maximize}}
	& & \mathbf{w}^{\top} \mathbf{v} \\
	& \text{subject to}
	& & \|\mathbf{v}\|_2 = 1, \\
	&&& \operatorname{supp}(\mathbf{v}) = \mathcal{S}.
\end{aligned}
\end{equation}
The solution is \(\mathbf{v}^* = \mathbf{w_{\mathcal{S}}} / \|\mathbf{w_{\mathcal{S}}}\|_2\), giving:
\begin{equation}
\mathbf{w}^{\top} \mathbf{v}^* = \|\mathbf{w_{\mathcal{S}}}\|_2.
\end{equation}
Substituting back:
\begin{equation}
\epsilon \cdot \left(1 - f_{\text{cls}}(\mathbf{e})\right) \cdot \|\mathbf{w_{\mathcal{S}}}\|_2 \geq \mathcal{L}(\mathbf{e}) - L_0.
\end{equation}
Solving for the minimal \(\epsilon\):
\begin{equation}
\epsilon^* = \frac{\mathcal{L}(\mathbf{e}) - L_0}{\left(1 - f_{\text{cls}}(\mathbf{e})\right) \cdot \|\mathbf{w_{\mathcal{S}}}\|_2}.
\end{equation}

\subsubsection*{A.4 Transformation to Logit Space}

We now can control the perturbation magnitude  \(\epsilon\) according to loss threshold \(L_0\), but it is inconvenient, we transform the loss-based formulation to the logit-based formulation. Recall that:
\begin{equation}
f_{\text{cls}}(\mathbf{e}) = \sigma(s) = \frac{1}{1+e^{-s}}, \quad \text{where} \quad s = \mathbf{w}^\top \mathbf{e} + b,
\end{equation}
\begin{equation}
\mathcal{L}(\mathbf{e}) = -\log(\sigma(s)) = \log(1+e^{-s}),
\end{equation}
\begin{equation}
L_0 = -\log(P_0), \quad \text{and} \quad S_0 = \log\left(\frac{P_0}{1-P_0}\right).
\end{equation}

First, we express the numerator:
\begin{equation}
\mathcal{L}(\mathbf{e}) - L_0 = \log(1+e^{-s}) + \log(P_0) = \log(P_0(1+e^{-s})).
\end{equation}

Next, we simplify the denominator:
\begin{equation}
1 - f_{\text{cls}}(\mathbf{e}) = 1 - \sigma(s) = \frac{e^{-s}}{1+e^{-s}} = \sigma(-s).
\end{equation}

Substituting these into our expression for \(\epsilon^*\):
\begin{equation}
\epsilon^* = \frac{\log(P_0(1+e^{-s}))}{\sigma(-s) \cdot \|\mathbf{w_{\mathcal{S}}}\|_2}
= \frac{\log(P_0(1+e^{-s})) \cdot (1+e^{-s})}{e^{-s} \cdot \|\mathbf{w_{\mathcal{S}}}\|_2}.
\end{equation}

Now, we show the connection to the logit difference \(S_0 - s\). Note that:
\begin{equation}
S_0 - s = \log\left(\frac{P_0}{1-P_0}\right) - s = \log\left(\frac{P_0 e^{-s}}{1-P_0}\right).
\end{equation}

For the typical case where the linear approximation holds well, we can make the following approximation:
\begin{equation}
\log(P_0(1+e^{-s})) \cdot (1+e^{-s}) \approx e^{-s} \cdot (S_0 - s).
\end{equation}

This approximation is valid because:
\begin{icompact}
	\item When \(s\) is in the linear region of the sigmoid function, \(1+e^{-s} \approx e^{-s}\)
	\item \(\log(P_0(1+e^{-s})) \approx \log(P_0) + \log(e^{-s}) = \log(P_0) - s\)
	\item \(S_0 - s = \log(\frac{P_0}{1-P_0}) - s \approx \log(P_0) - s\) when \(P_0\) is not too small
\end{icompact}

Thus, we have:
\begin{equation}
\epsilon^* \approx \frac{e^{-s} \cdot (S_0 - s)}{e^{-s} \cdot \|\mathbf{w_{\mathcal{S}}}\|_2}
= \frac{S_0 - s}{\|\mathbf{w_{\mathcal{S}}}\|_2}.
\end{equation}

Substituting back \(s = \mathbf{w}^\top \mathbf{e} + b\), we obtain the final expression:

\begin{equation}
	\label{eq:final-epsilon-corrected}
	\epsilon^* = \frac{ S_0 - (\mathbf{w}^\top \mathbf{e} + b) }{ \|\mathbf{w_{\mathcal{S}}}\|_2 }.
\end{equation}

\subsection*{B Budget-Redistribution Effect}

The relationship between the number of perturbed attention heads and magnitude is governed by the fundamental budget-redistribution effect, which we formalize through rigorous optimization analysis.

\noindent \textbf{Theoretical Foundation}: For a selected head set $\mathcal{S}$, the minimal perturbation magnitude required to flip the safety classifier is derived as:
\begin{equation}
\epsilon^*(\mathcal{S}) = \frac{C}{\|\mathbf{w}_{\mathcal{S}}\|_2}, \nonumber
\end{equation}
where $C$ denotes an input-dependent constant and $\|\mathbf{w}_{\mathcal{S}}\|_2$ represents the L2-norm of classifier weights projected onto the $\mathcal{S}$-subspace.


\begin{proposition}[Budget-Redistribution Effect]
The minimal achievable perturbation magnitude, denoted by $\epsilon^*_k$, is a non-increasing function of the number of perturbed heads $k$. Formally, for any $k_1 < k_2$,
\begin{equation}
\epsilon^*_{k_1} \geq \epsilon^*_{k_2}. \nonumber
\end{equation}
\end{proposition}

\begin{proof}
Let $\mathcal{S}^*_{k_1}$ and $\mathcal{S}^*_{k_2}$ denote the optimal head sets corresponding to perturbation budgets $k_1$ and $k_2$, respectively. Without loss of generality, construct $\mathcal{S}^*_{k_2}$ such that it contains $\mathcal{S}^*_{k_1}$, thereby introducing additional degrees of freedom for perturbation alignment. Under this construction, it follows that
\begin{equation}
\|\mathbf{w}_{\mathcal{S}^*_{k_2}}\|_2 \geq \|\mathbf{w}_{\mathcal{S}^*_{k_1}}\|_2, \nonumber
\end{equation}
and consequently,
\begin{equation}
\epsilon^*_{k_2} \leq \epsilon^*_{k_1}, \nonumber
\end{equation}
since $\epsilon^* \propto 1 / \|\mathbf{w}_{\mathcal{S}}\|_2$. 
\end{proof}

This proposition formalizes the intuition of ``force concentration.'' A smaller $k$ constrains the attack to a narrower channel, forcing a larger perturbation per head to achieve the same objective. We now test this prediction by measuring the actual perturbation magnitude $\varepsilon(\alpha)$ as we vary the effective $k$.

\subsection*{C Additional Experiments}

This section presents additional experimental results that complement the main findings in the paper.

\subsubsection*{C.1 Lantent Safety Probe}
\label{sec:classifier_implementation}

To instantiate the theoretical safety classifier $f_{\text{cls}}$ defined in subsection \ref{subsec:head_selection_localization}, we construct a \textbf{Latent Safety Probe (LSP)} for each victim LLM. Following standard mechanistic interpretability protocols \cite{67probing}, the LSP is designed to capture the linear separability of safety-aligned representations.

\paragraph{Dataset Construction.}
We curated a balanced probing dataset, $\mathcal{D}_{\text{probe}}$, comprising 200 samples (100 benign and 100 malicious instructions).
To represent standard safe interactions, the benign samples were randomly drawn from the \texttt{alpaca-cleaned} dataset.
The malicious samples were sourced from the \texttt{strongreject} dataset.
Crucially, to prevent data leakage and ensure the probe generalizes to unseen attacks, we rigorously filtered the training data: we explicitly excluded any prompts that appear in our evaluation benchmarks.
For each input, we extracted the internal hidden states $\mathbf{h} \in \mathbb{R}^d$ at the final token position, which serves as the aggregate representation for the autoregressive generation.

\paragraph{Training and Validation.}
For each target LLM, we trained a specific LSP using Logistic Regression to model the probability $P(\text{safe} \mid \mathbf{h})$. The binary labels were assigned as $y=1$ for benign inputs (indicating safety) and $y=0$ for malicious inputs. We opted for a linear probe rather than a non-linear MLP to validate our assumption that safety alignment creates a linearly separable manifold in the latent space.
\clearpage
\subsubsection*{C.2 The ASR changed with \(\alpha\)}

Figure~\ref{appfig: ASR curves of SAHA}--\ref{appfig: ASR curves} validates the variation trends of ASR with respect to \(\alpha\). The consistent balance observed across models with varying alignment strengths validates that SAHA’s performance is not confined to a narrow subset of OSLLMs but exhibits universality, further supporting the claims regarding its practical utility for OSLLMs safety evaluation.

\begin{figure}[ht!]
	\centering
	\begin{subfigure}[b]{0.48\textwidth}
		\centering
		\includegraphics[width=\textwidth]{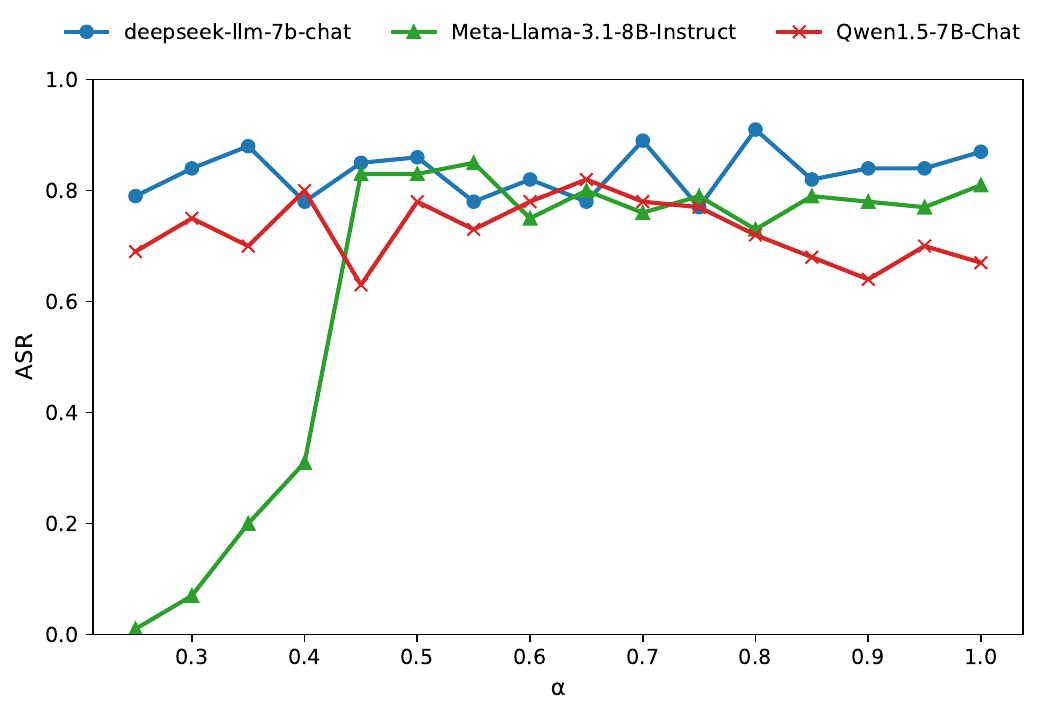}
		\caption{JailbreakBench}
		\label{fig:sub_a}
	\end{subfigure}
	\hfill
	\begin{subfigure}[b]{0.48\textwidth}
		\centering
		\includegraphics[width=\textwidth]{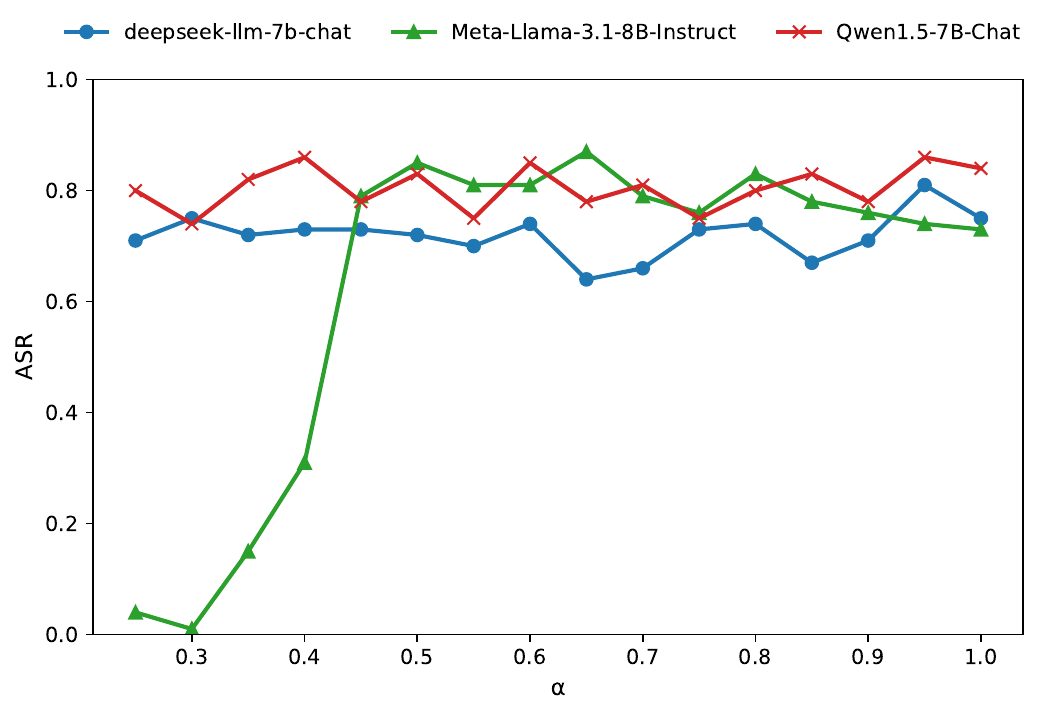}
		\caption{MaliciousInstruct}
		\label{fig:sub_c}
	\end{subfigure}
	\caption{The ASR tendency with \(\alpha\) of AIR+LWP.}
    \vspace{-0.1in}
	\label{appfig: ASR curves of SAHA}
\end{figure}

\begin{figure}[htbp] 
	\centering 
	\begin{subfigure}[b]{0.32\textwidth}
		\centering  
		\includegraphics[width=\textwidth]{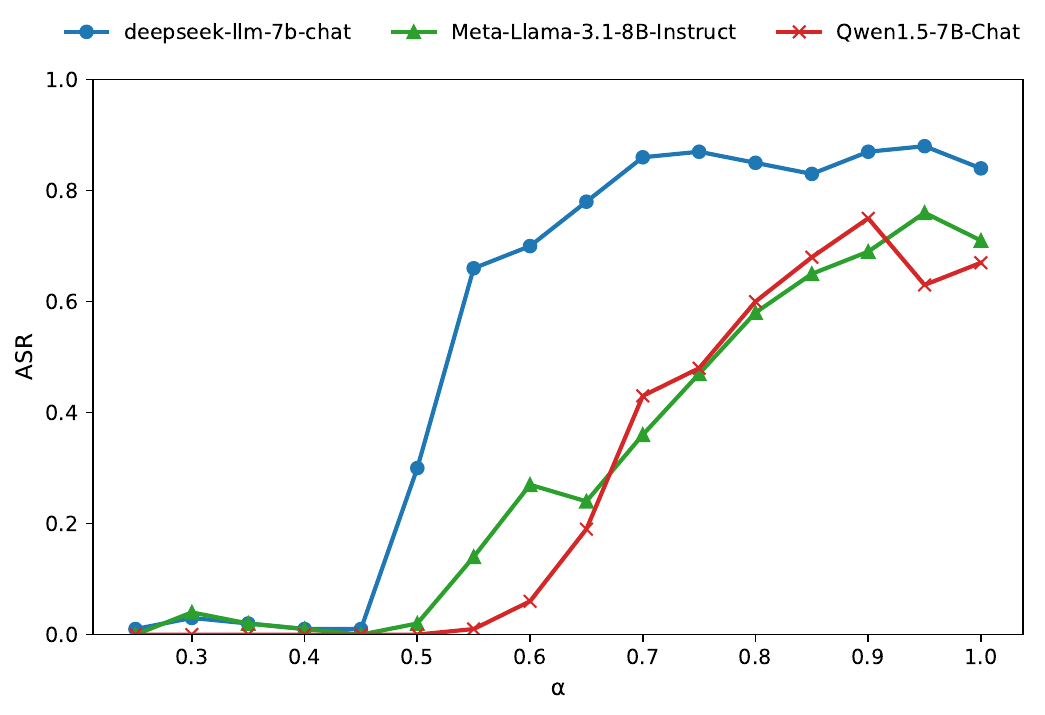} 
		\caption{AIR+GWP, JailbreakBench}
		\label{fig:sub_a}
	\end{subfigure}
    \begin{subfigure}[b]{0.32\textwidth}
		\centering
		\includegraphics[width=\textwidth]{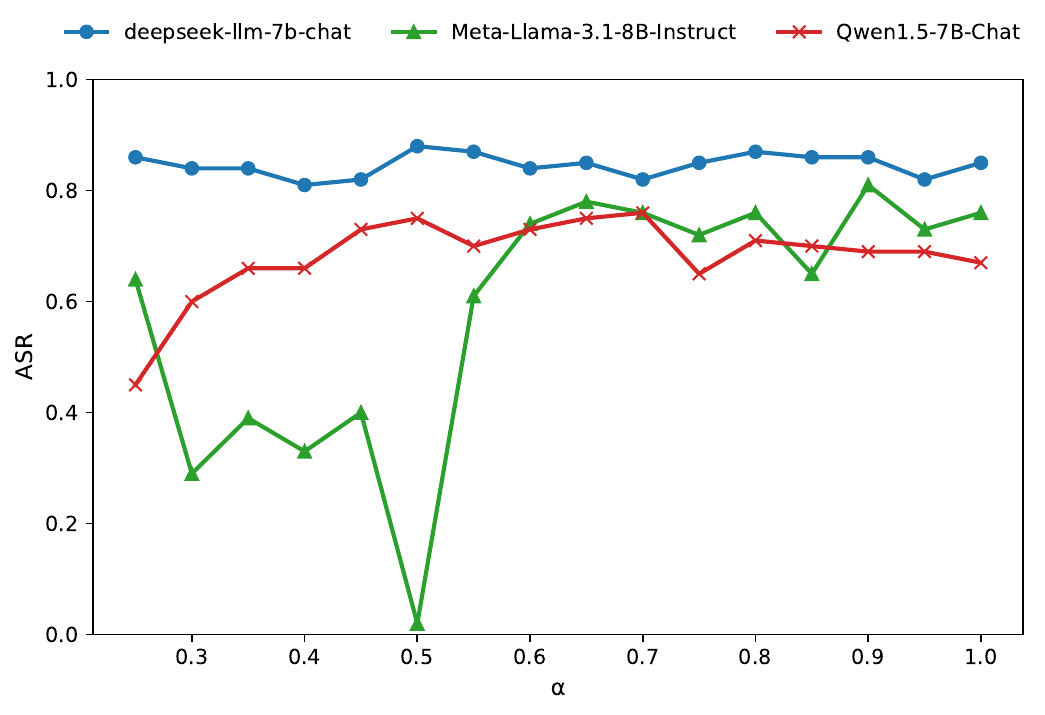}
		\caption{APR+GWP, JailbreakBench}
		\label{fig:sub_d}
	\end{subfigure}
    \begin{subfigure}[b]{0.32\textwidth}
		\centering
		\includegraphics[width=\textwidth]{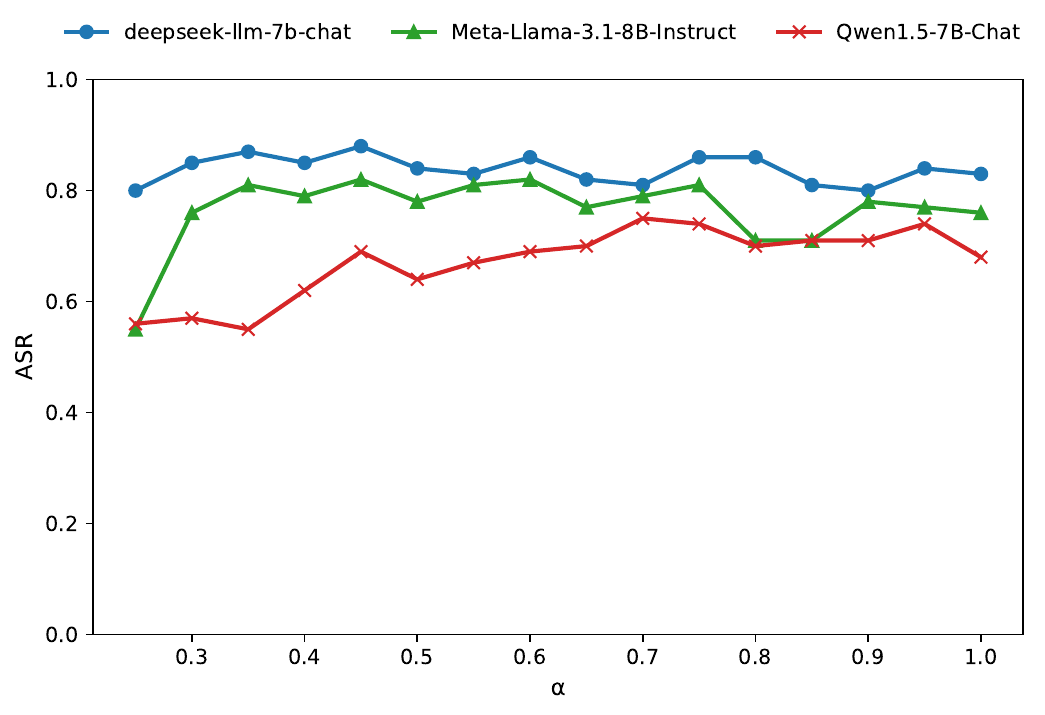}
		\caption{APR+LWP, JailbreakBench}
		\label{fig:sub_c}
	\end{subfigure}
    \vspace{0.5cm}
	\begin{subfigure}[b]{0.32\textwidth}
		\centering
		\includegraphics[width=\textwidth]{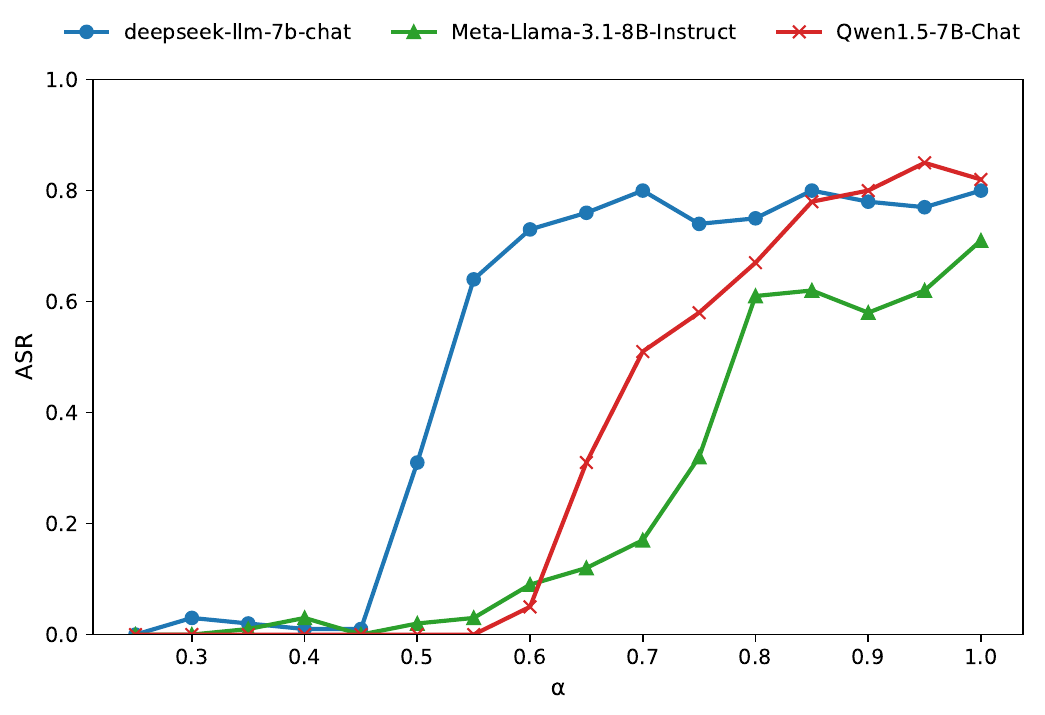}
		\caption{AIR+GWP, MaliciousInstruct}
		\label{fig:sub_c}
	\end{subfigure}
    \begin{subfigure}[b]{0.32\textwidth}
		\centering  
		\includegraphics[width=\textwidth]{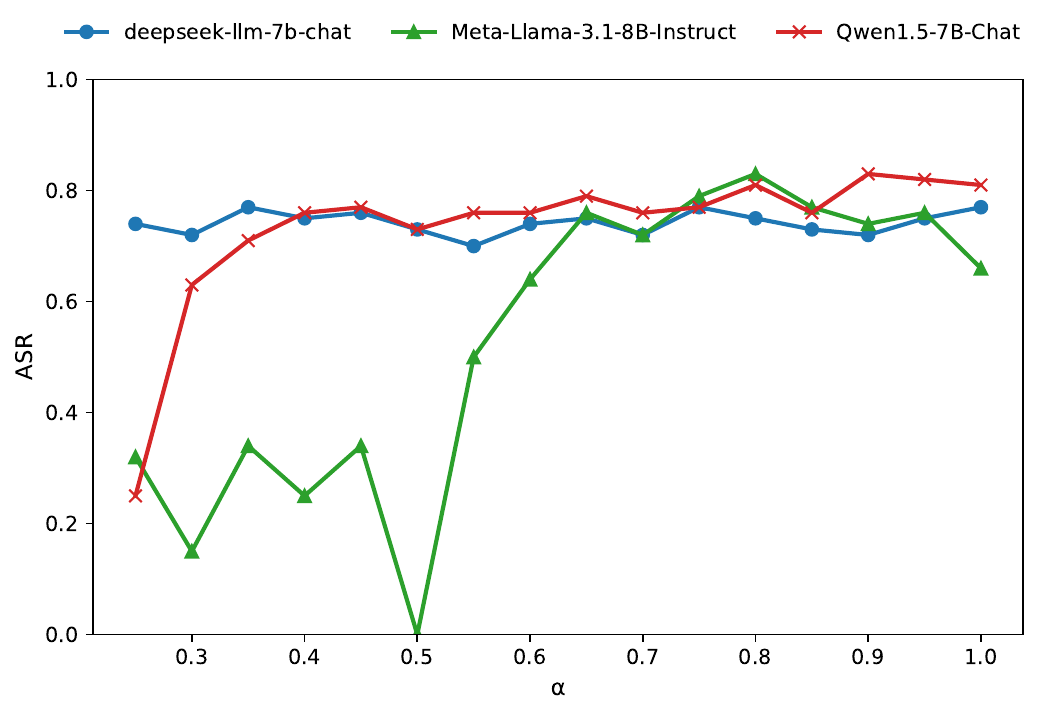} 
		\caption{APR+GWP, MaliciousInstruct}
		\label{fig:sub_a}
	\end{subfigure}
	\hfill  
	\begin{subfigure}[b]{0.32\textwidth}
		\centering
		\includegraphics[width=\textwidth]{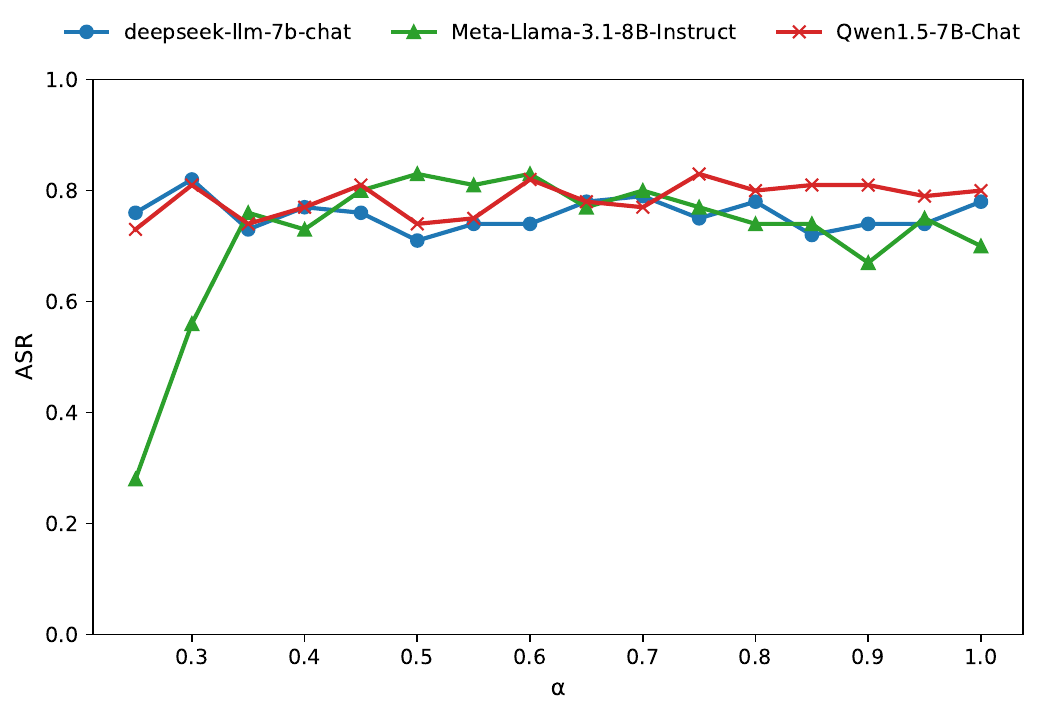}
		\caption{APR+LWP, MaliciousInstruct}
		\label{fig:sub_d}
	\end{subfigure}
	\caption{The ASR tendency with \(\alpha\).}   
	\label{appfig: ASR curves} 
\end{figure}
\clearpage

\subsubsection*{C.3 perturbation magnitudes}

Figure~\ref{appfig:magnitude curves on MI} validate the variation trends of perturbation magnitude with respect to \(\alpha\). Figure~\ref{appfig:heatmap on JB-deepseek}--\ref{appfig:heatmap on MI-qwen} are the heatmaps of different models and different datasets. These cross-architecture findings strengthen the mechanistic understanding of how SAHA interacts with LLM internal components to induce safety failures.

\begin{figure}[htbp] 
	\centering 
    \begin{subfigure}[b]{0.32\textwidth}
		\centering
		\includegraphics[width=\textwidth]{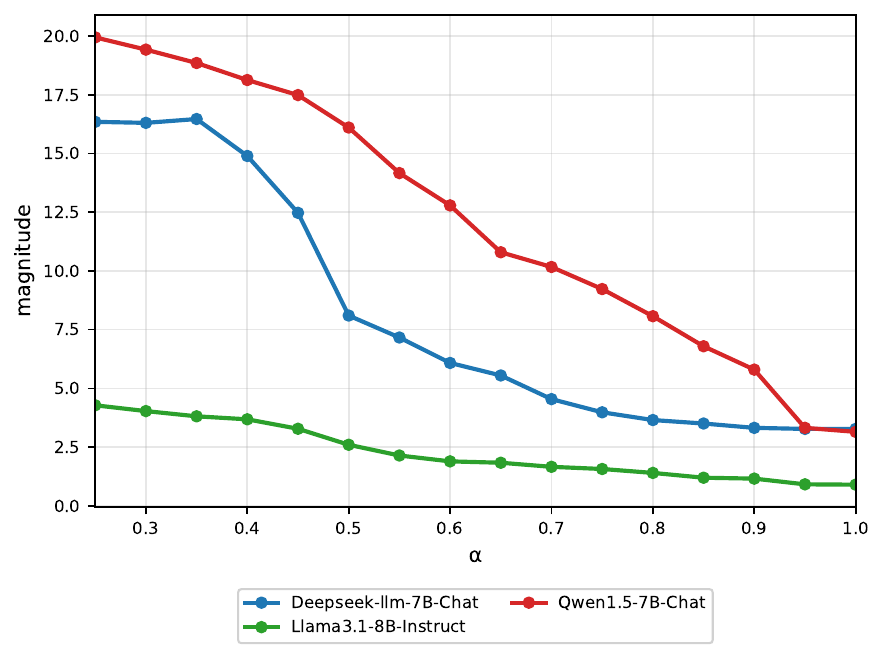}
		\caption{AIR+GWP on JailbreakBench}
		\label{fig:sub_d}
	\end{subfigure}
    \begin{subfigure}[b]{0.32\textwidth}
		\centering
		\includegraphics[width=\textwidth]{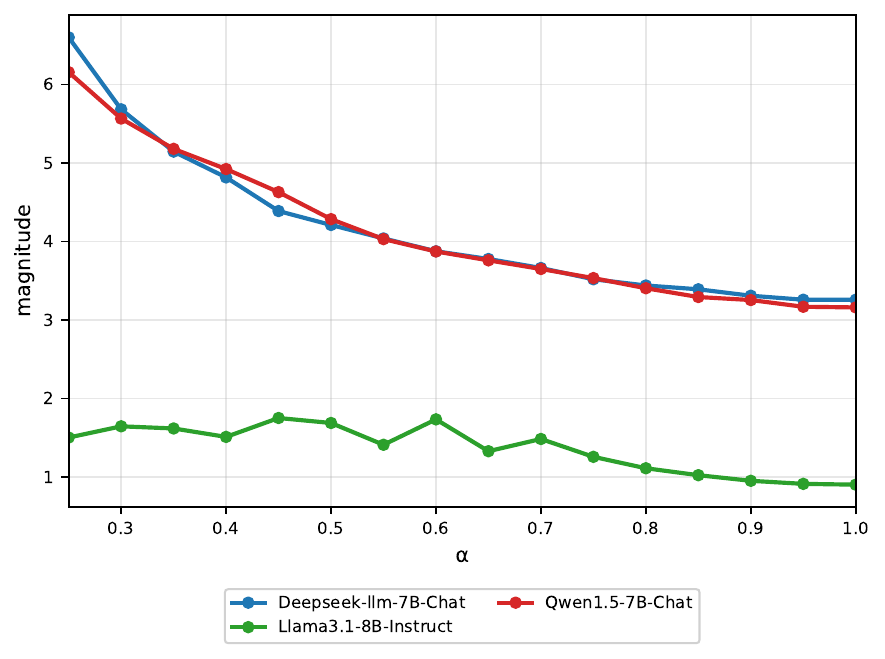}
		\caption{APR+GWP on JailbreakBench}
		\label{fig:sub_d}
	\end{subfigure}
    \hfill
    \begin{subfigure}[b]{0.32\textwidth}
		\centering
		\includegraphics[width=\textwidth]{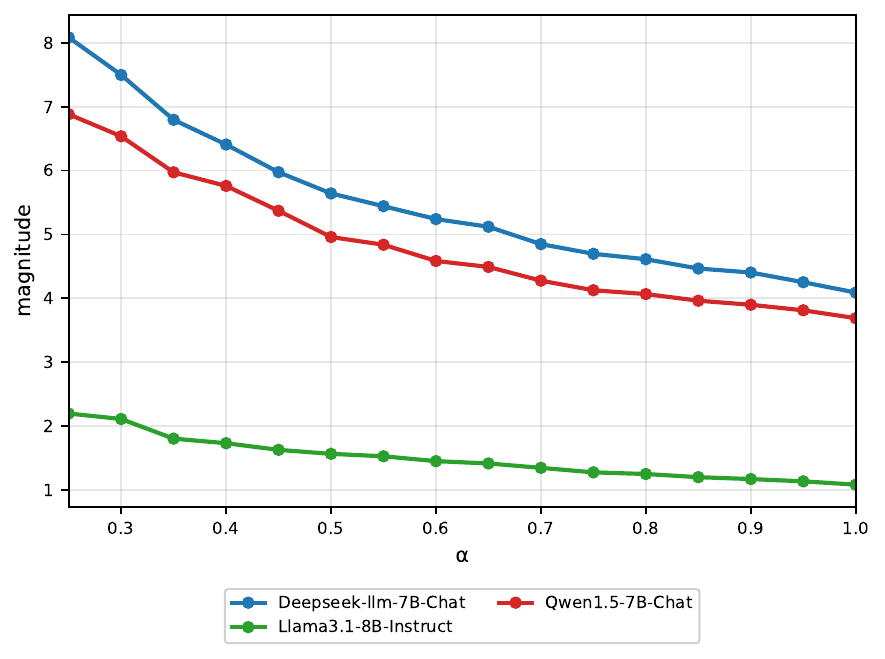}
		\caption{APR+LWP on JailbreakBench}
		\label{fig:sub_c}
	\end{subfigure}
    \vspace{0.5cm}
    \begin{subfigure}[b]{0.32\textwidth}
		\centering  
		\includegraphics[width=\textwidth]{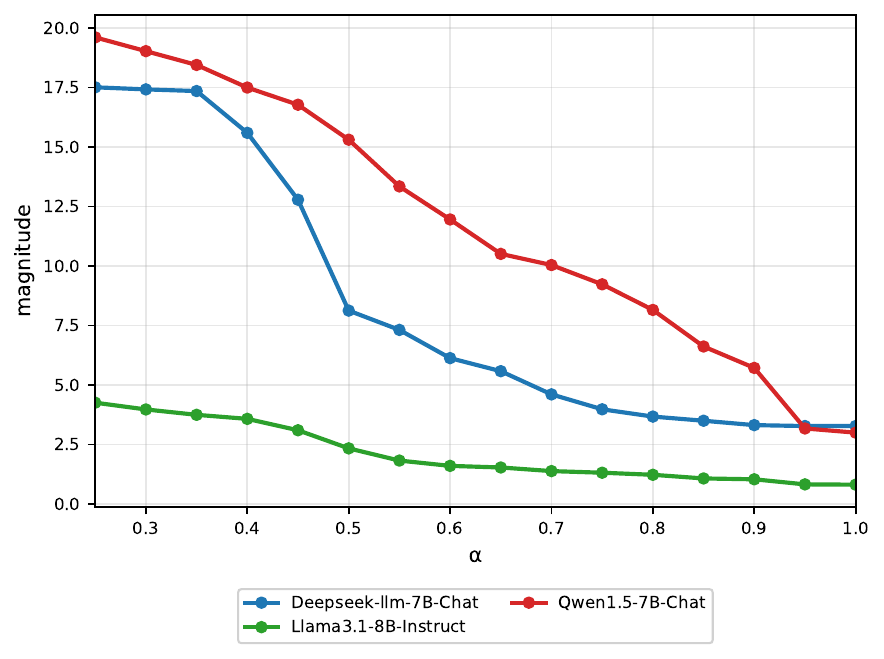} 
		\caption{AIR+GWP on MaliciousInstruct}
		\label{fig:sub_a}
	\end{subfigure}
	\hfill  
	\begin{subfigure}[b]{0.32\textwidth}
		\centering
		\includegraphics[width=\textwidth]{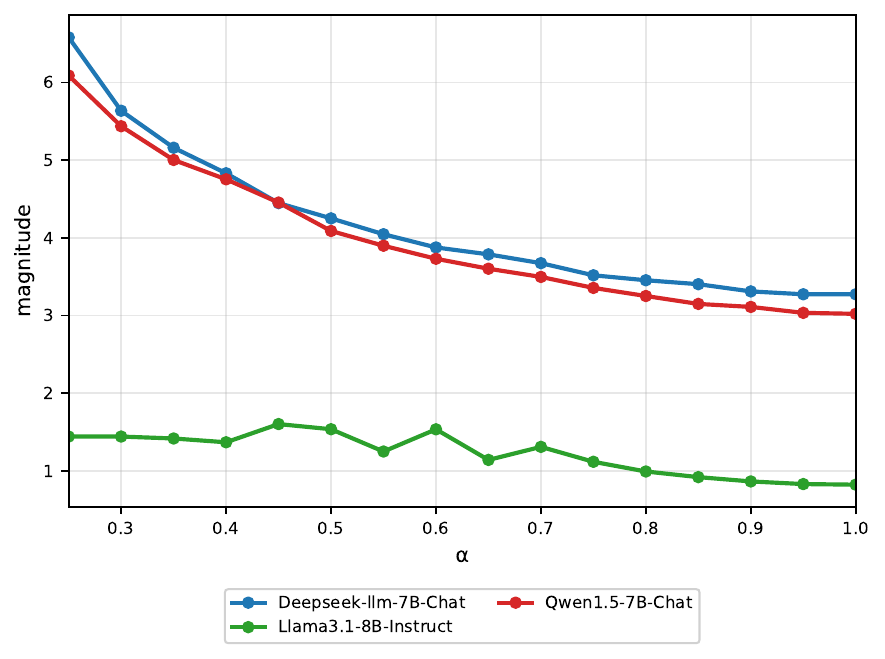}
		\caption{APR+GWP on MaliciousInstruct}
		\label{fig:sub_b}
	\end{subfigure}
	\hfill
	\begin{subfigure}[b]{0.32\textwidth}
		\centering
		\includegraphics[width=\textwidth]{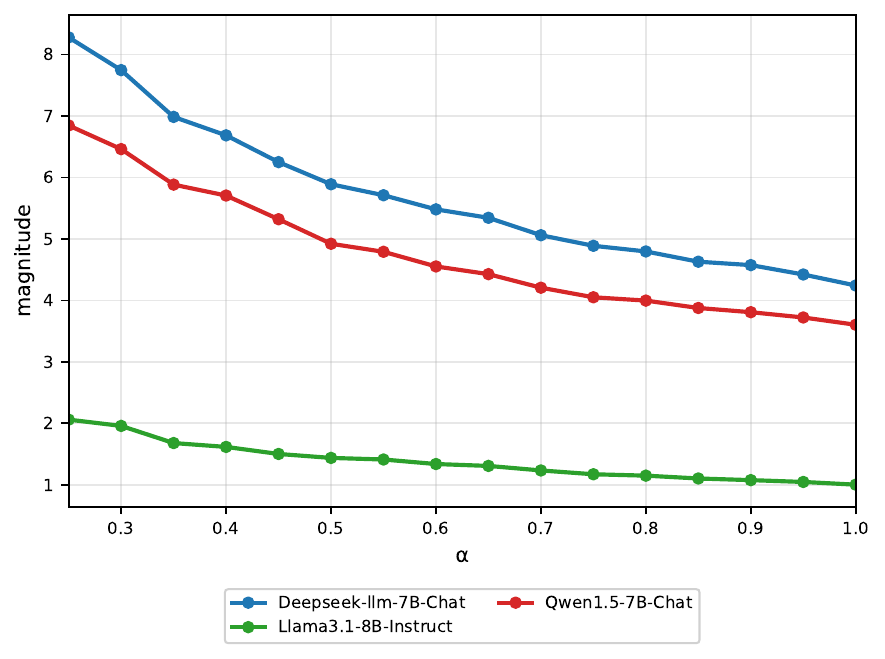}
		\caption{APR+LWP on MaliciousInstruct}
		\label{fig:sub_d}
	\end{subfigure}
	\caption{Layer-wise average perturbation magnitude $\varepsilon_\ell(\alpha)$ plotted across layers for representative $\alpha$ values.}
	\label{appfig:magnitude curves on MI} 
\end{figure}

\begin{figure}[htbp]
	\centering  
	\setlength{\tabcolsep}{2pt}  
	\begin{subfigure}[b]{0.48\textwidth}  
		\centering
		\includegraphics[width=\textwidth, keepaspectratio]{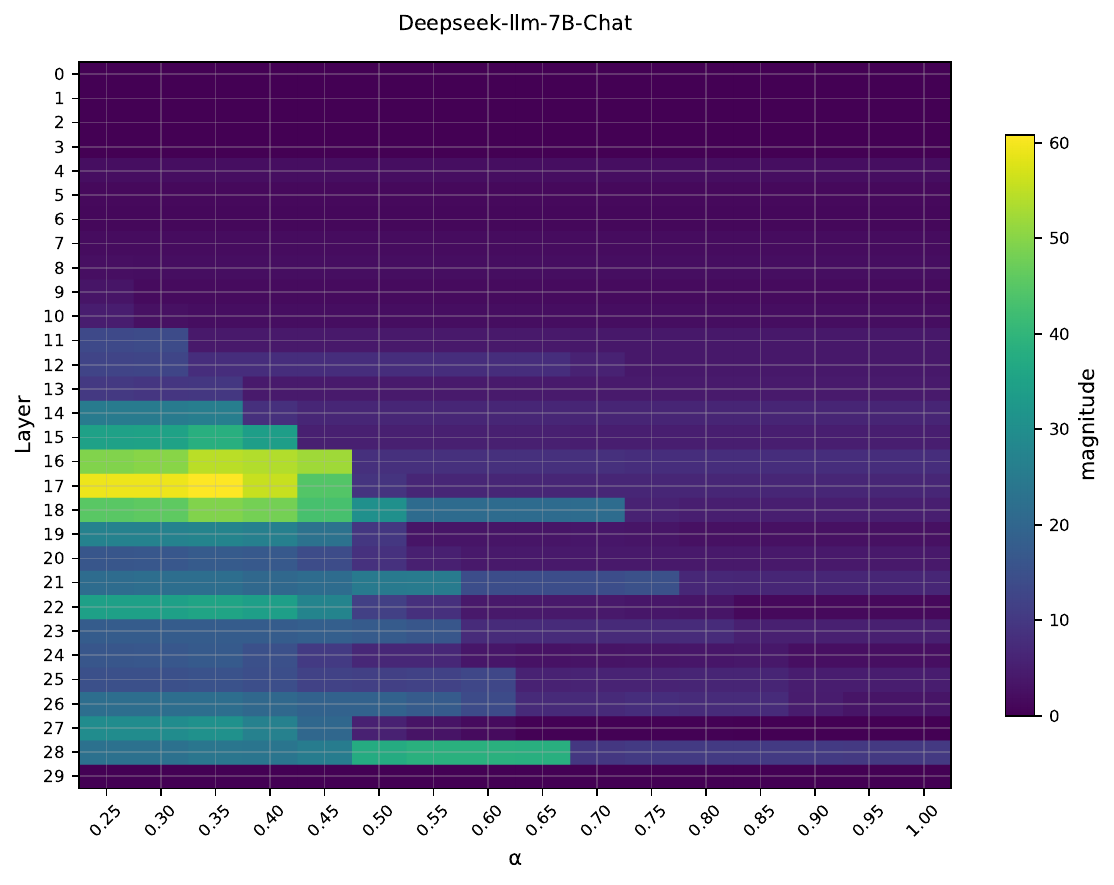} 
		\subcaption{Deepseek AIR+GWP} 
		\label{subfig:1}
	\end{subfigure}
	\hfill 
	\begin{subfigure}[b]{0.48\textwidth}
		\centering
		\includegraphics[width=\textwidth, keepaspectratio]{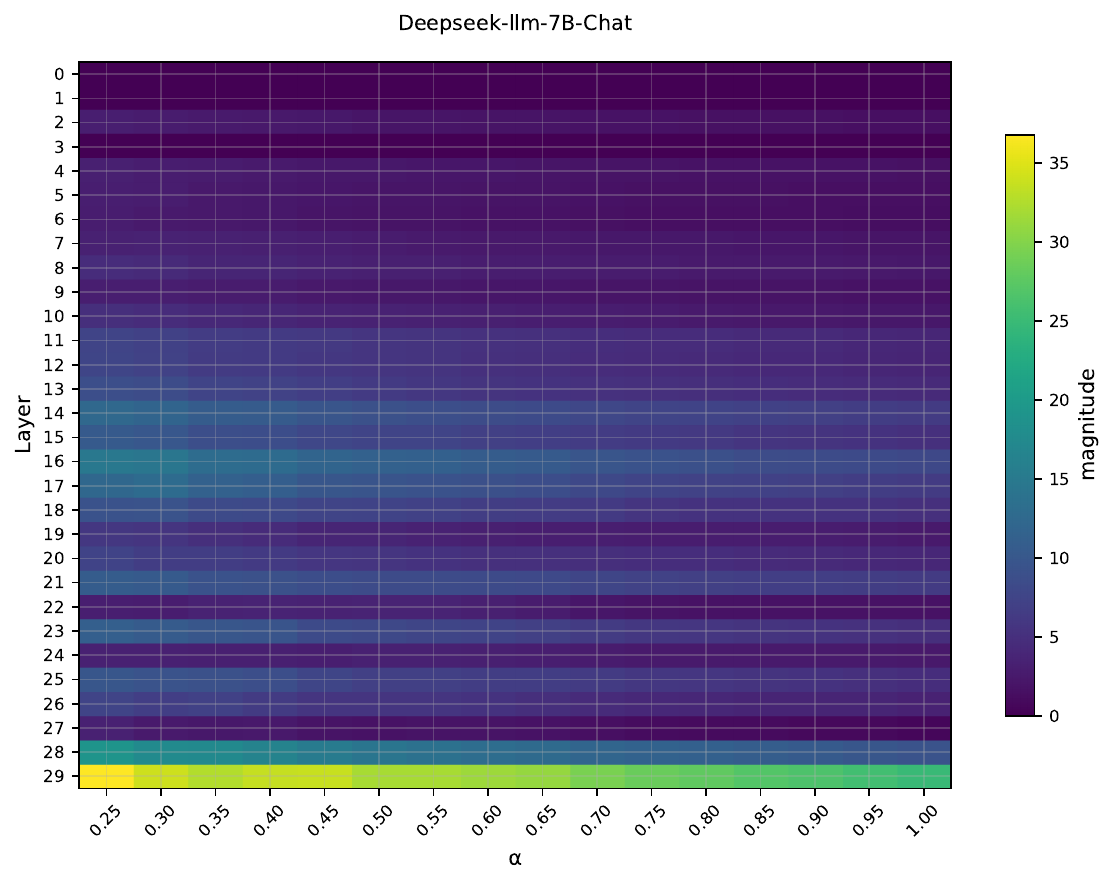}
		\subcaption{Deepseek AIR+LWP}
		\label{subfig:2}
	\end{subfigure}
    \vspace{0.5cm}
	\begin{subfigure}[b]{0.48\textwidth}
		\centering
		\includegraphics[width=\textwidth, keepaspectratio]{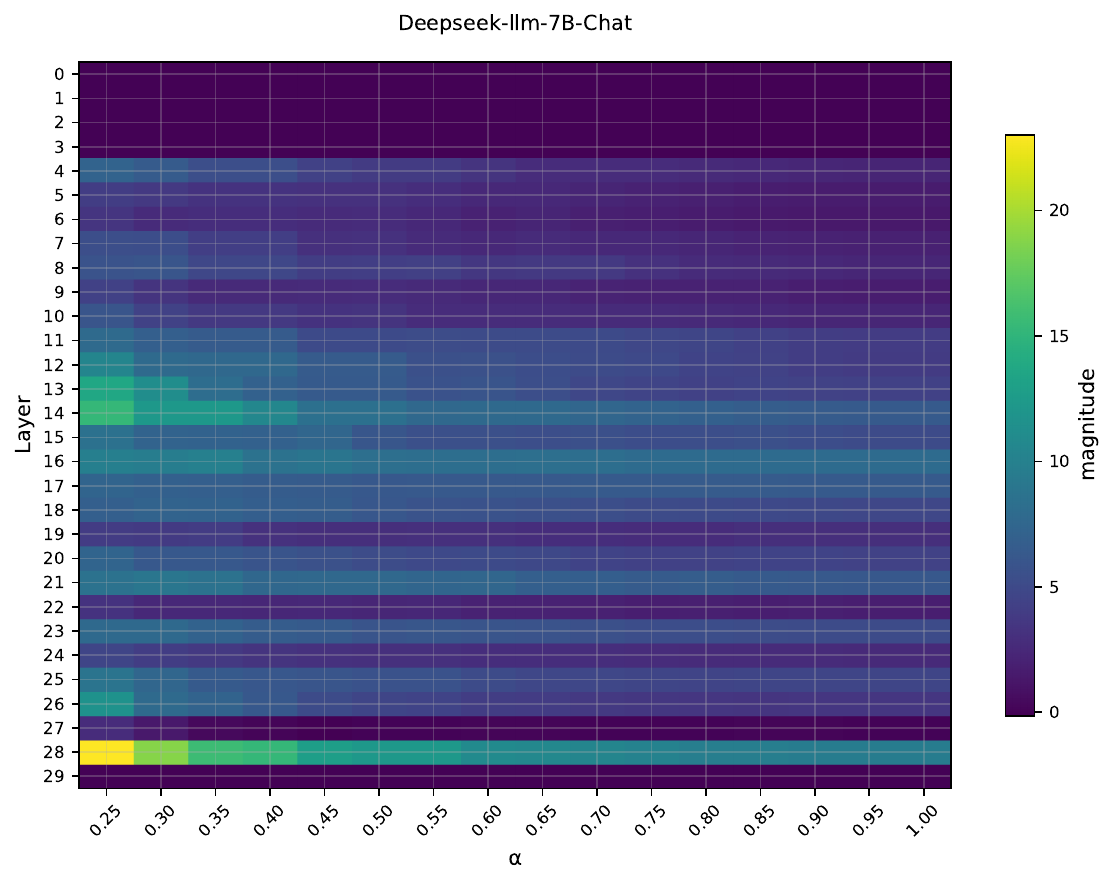}
		\subcaption{Deepseek APR+GWP}
		\label{subfig:3}
	\end{subfigure}
	\hfill
	\begin{subfigure}[b]{0.48\textwidth}
		\centering
		\includegraphics[width=\textwidth, keepaspectratio]{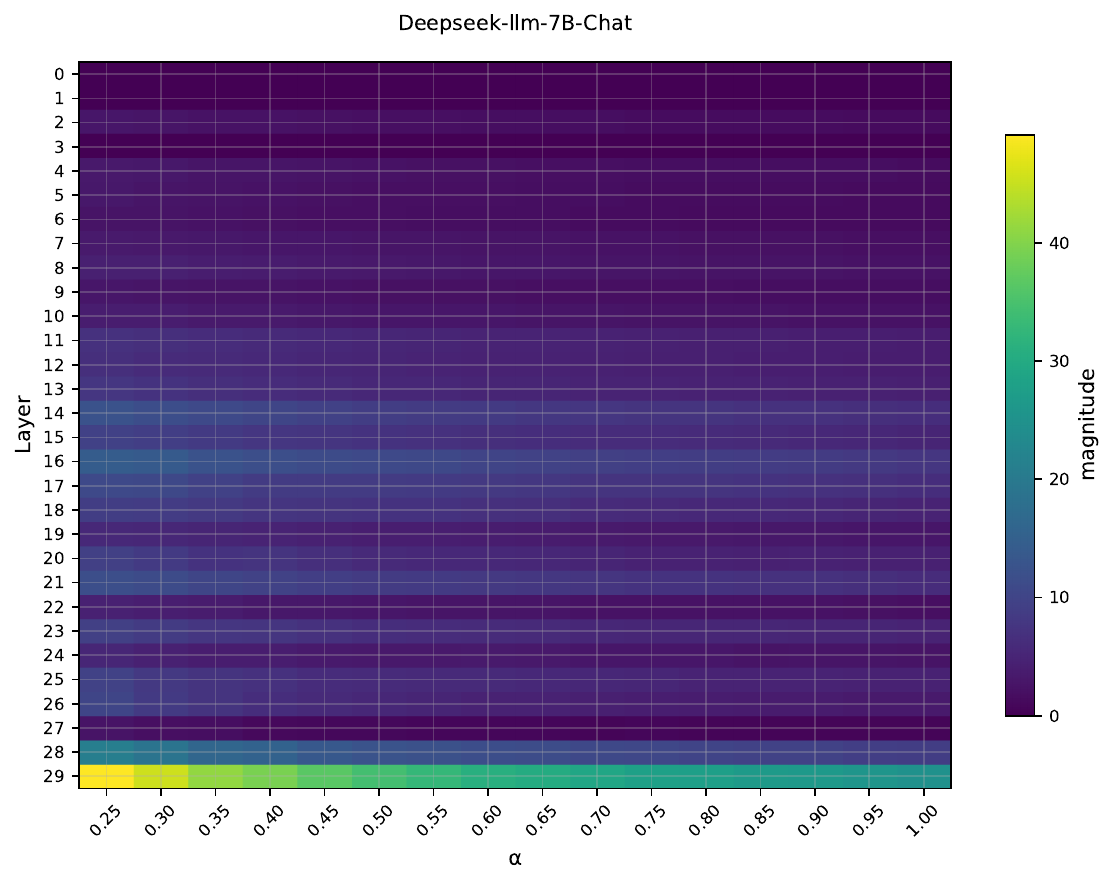}
		\subcaption{Deepseek APR+LWP}
		\label{subfig:4}
	\end{subfigure}	
	\caption{Layer~\(\times\)~\(\alpha\) heatmap of average perturbation magnitudes \(R_l(\alpha)\) on JailbreakBench. The model is Deepseek-llm-7B-Chat.}
	\label{appfig:heatmap on JB-deepseek}  
\end{figure}

\begin{figure}[htbp]
	\centering 
	\setlength{\tabcolsep}{2pt} 
	\begin{subfigure}[b]{0.48\textwidth} 
		\centering
		\includegraphics[width=\textwidth, keepaspectratio]{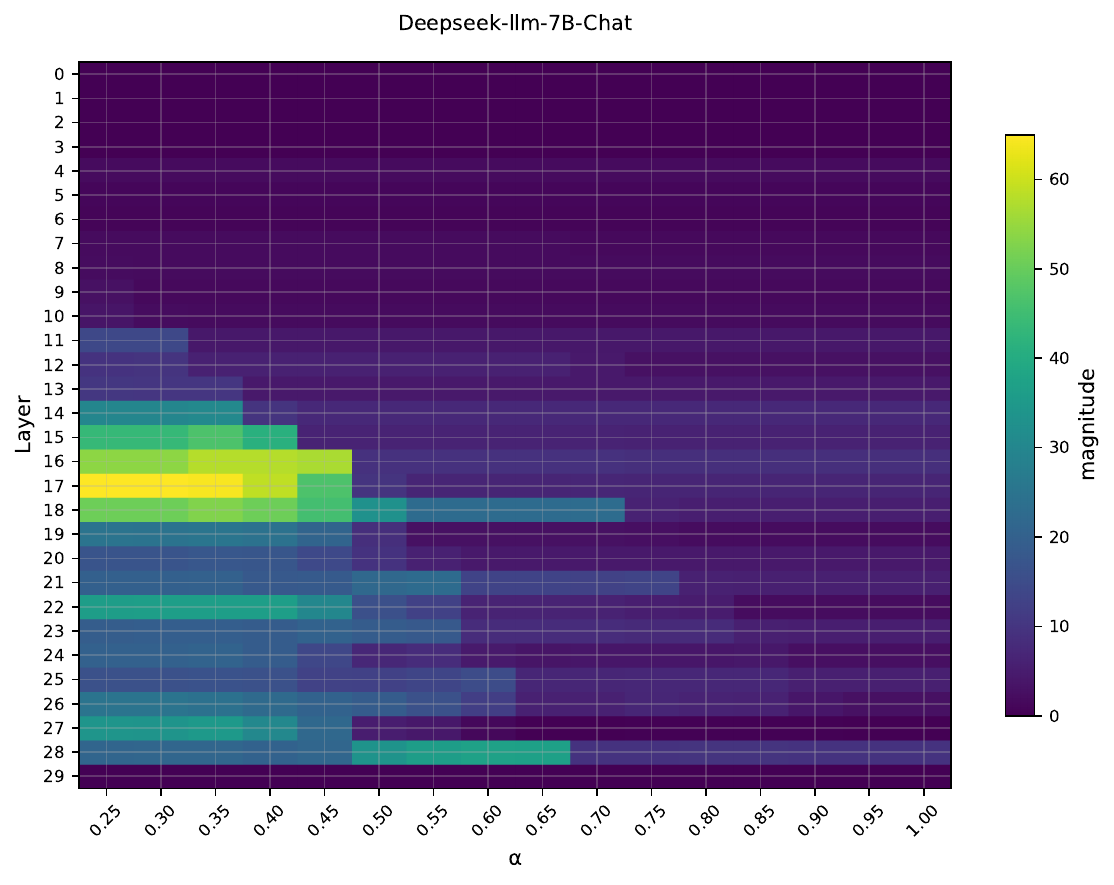}
		\subcaption{Deepseek AIR+GWP}  
		\label{subfig:1}
	\end{subfigure}
	\hfill  
	\begin{subfigure}[b]{0.48\textwidth}
		\centering
		\includegraphics[width=\textwidth, keepaspectratio]{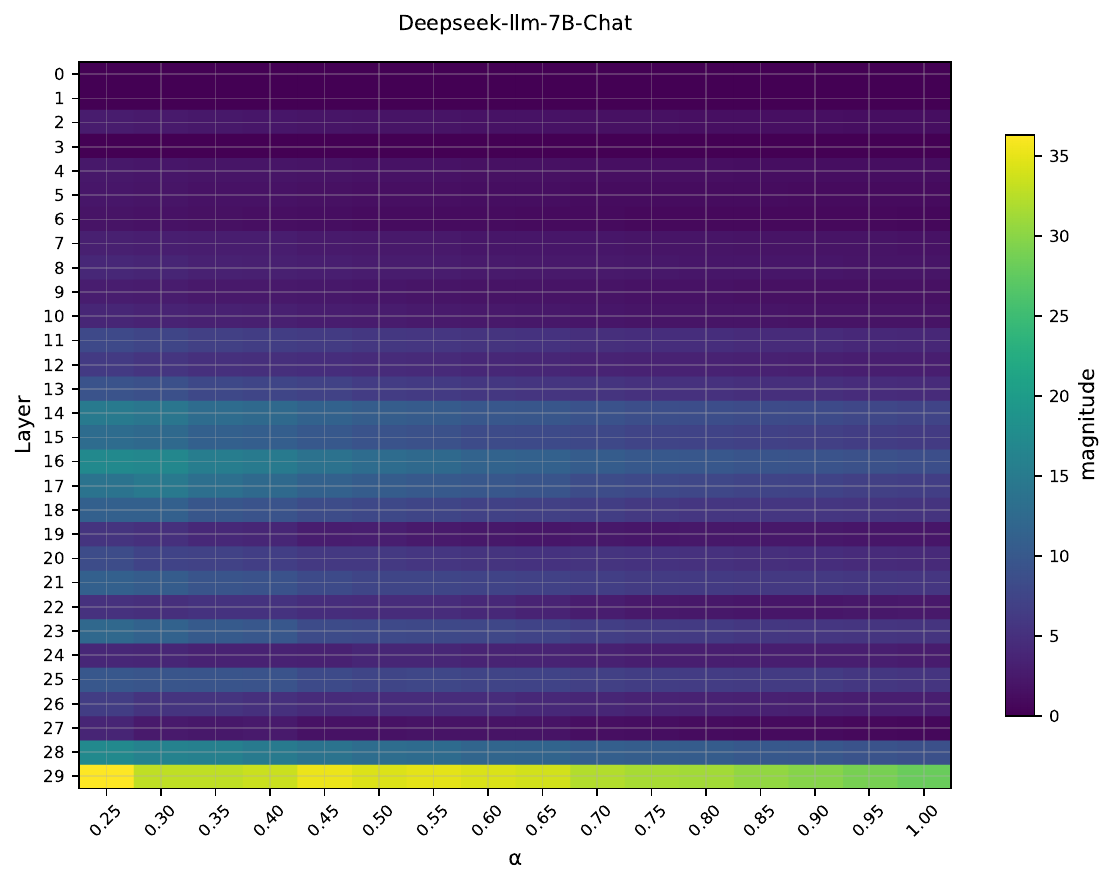}
		\subcaption{Deepseek AIR+LWP}
		\label{subfig:2}
	\end{subfigure}
    \vspace{0.5cm}
	\begin{subfigure}[b]{0.48\textwidth}
		\centering
		\includegraphics[width=\textwidth, keepaspectratio]{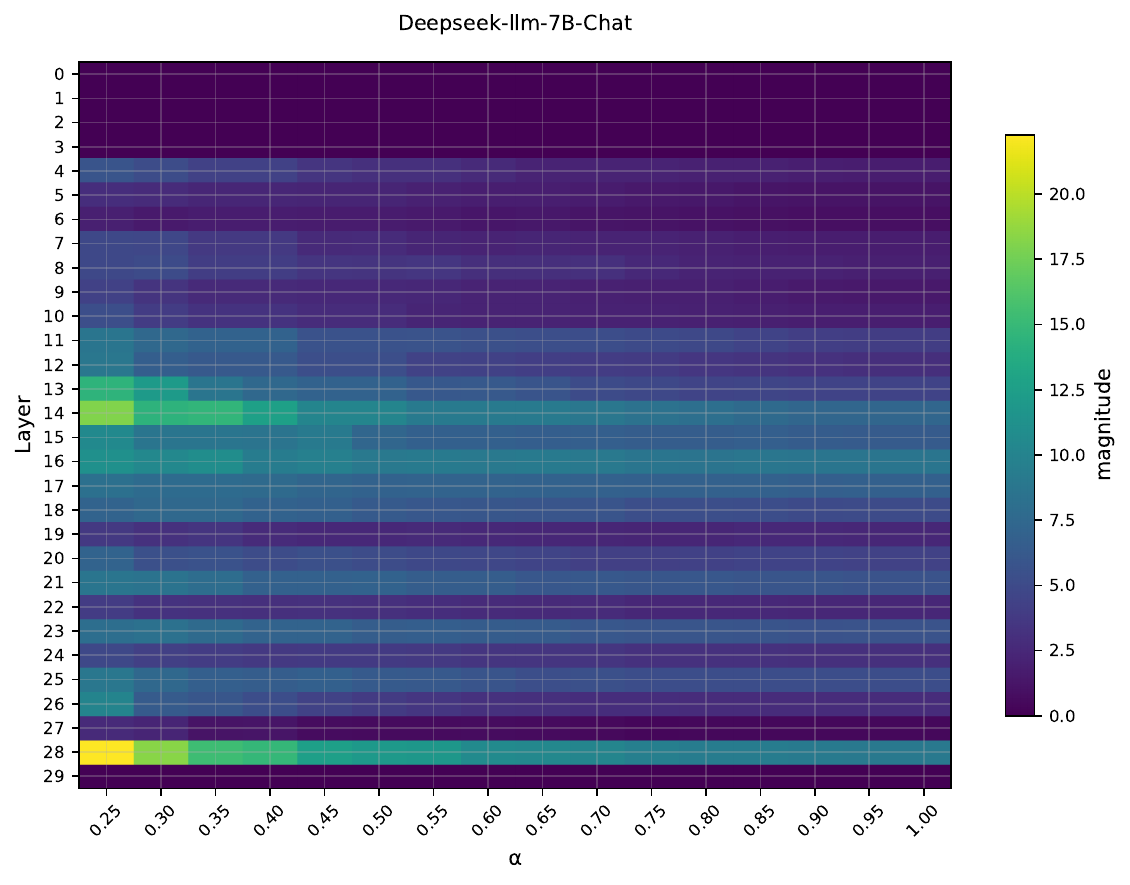}
		\subcaption{Deepseek APR+GWP}
		\label{subfig:3}
	\end{subfigure}
	\hfill
	\begin{subfigure}[b]{0.48\textwidth}
		\centering
		\includegraphics[width=\textwidth, keepaspectratio]{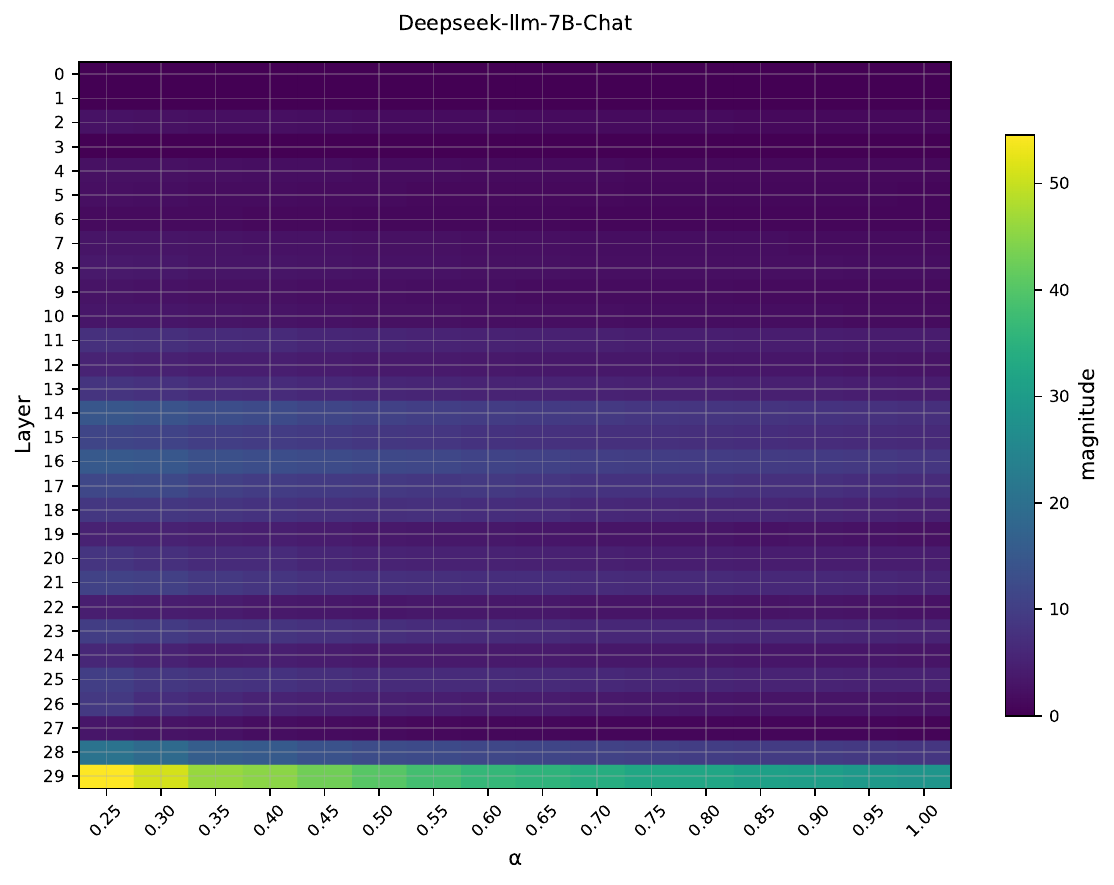}
		\subcaption{Deepseek APR+LWP}
		\label{subfig:4}
	\end{subfigure}
	\caption{Layer~\(\times\)~\(\alpha\) heatmap of average perturbation magnitudes \(R_l(\alpha)\) on MaliciousInstruct. The model is Deepseek-llm-7B-Chat.}
	\label{appfig:heatmap on MI-deepseek}  
\end{figure}

\begin{figure}[htbp]
	\centering 
	\setlength{\tabcolsep}{2pt}  
		\begin{subfigure}[b]{0.48\textwidth}
				\centering
				\includegraphics[width=\textwidth, keepaspectratio]{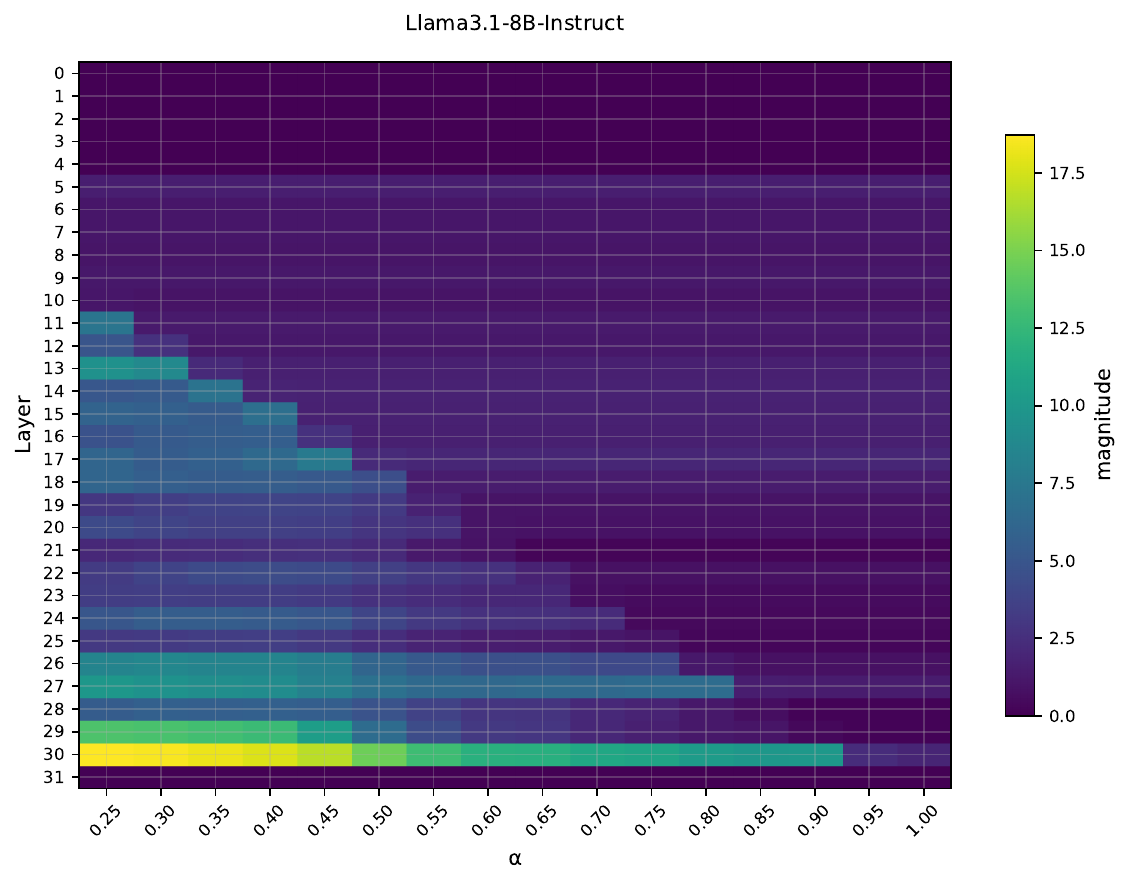} 
				\subcaption{Llama3.1 AIR+GWP}
				\label{subfig:9}
			\end{subfigure}
		\hfill
		\begin{subfigure}[b]{0.48\textwidth}
				\centering
				\includegraphics[width=\textwidth, keepaspectratio]{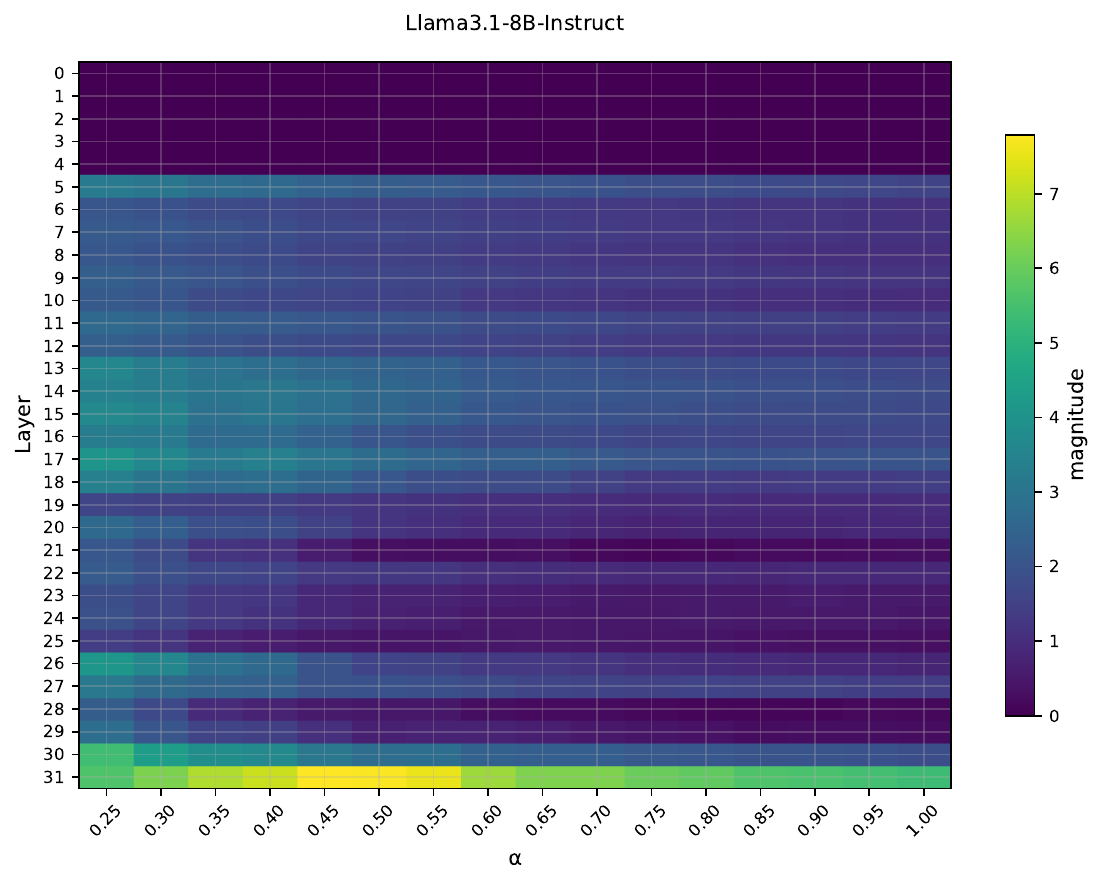}
				\subcaption{Llama3.1 AIR+LWP}
				\label{subfig:10}
			\end{subfigure}
        \vspace{0.5cm}
		\begin{subfigure}[b]{0.48\textwidth}
				\centering
				\includegraphics[width=\textwidth, keepaspectratio]{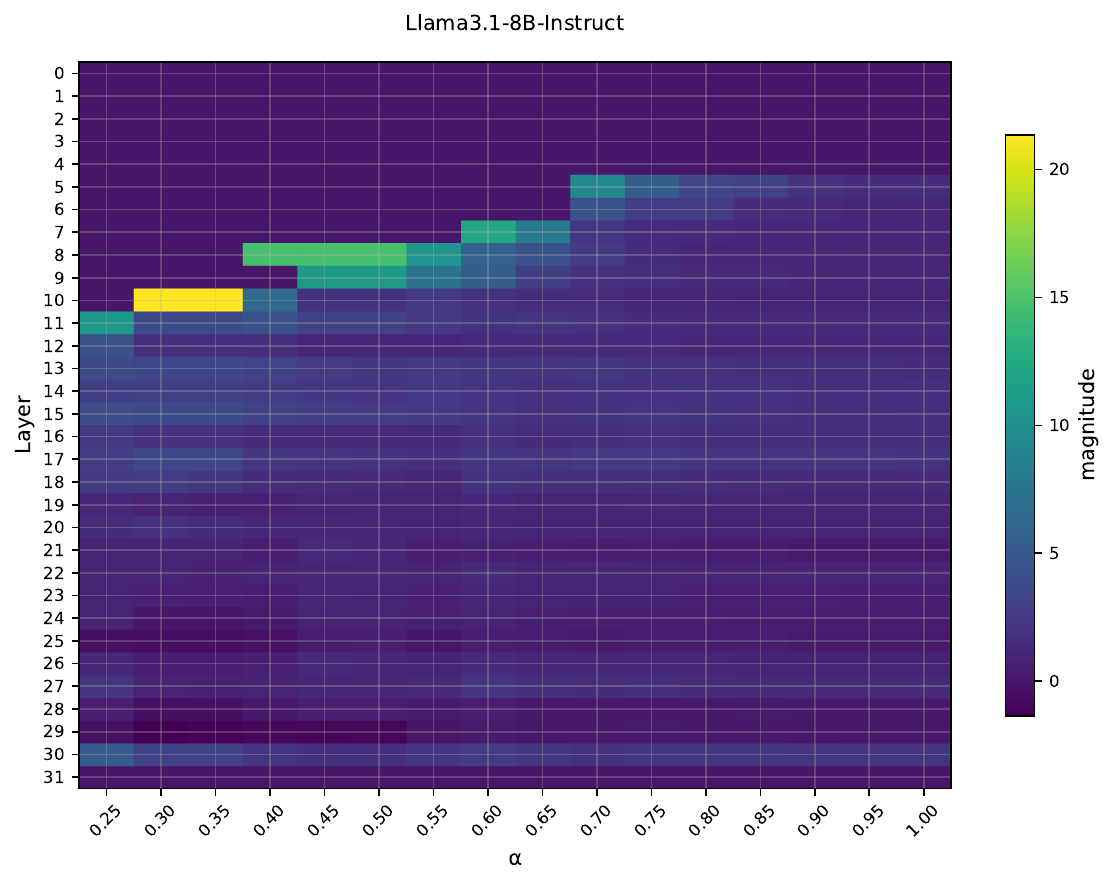}
				\subcaption{Llama3.1 APR+GWP}
				\label{subfig:11}
			\end{subfigure}
		\hfill
		\begin{subfigure}[b]{0.48\textwidth}
				\centering
				\includegraphics[width=\textwidth, keepaspectratio]{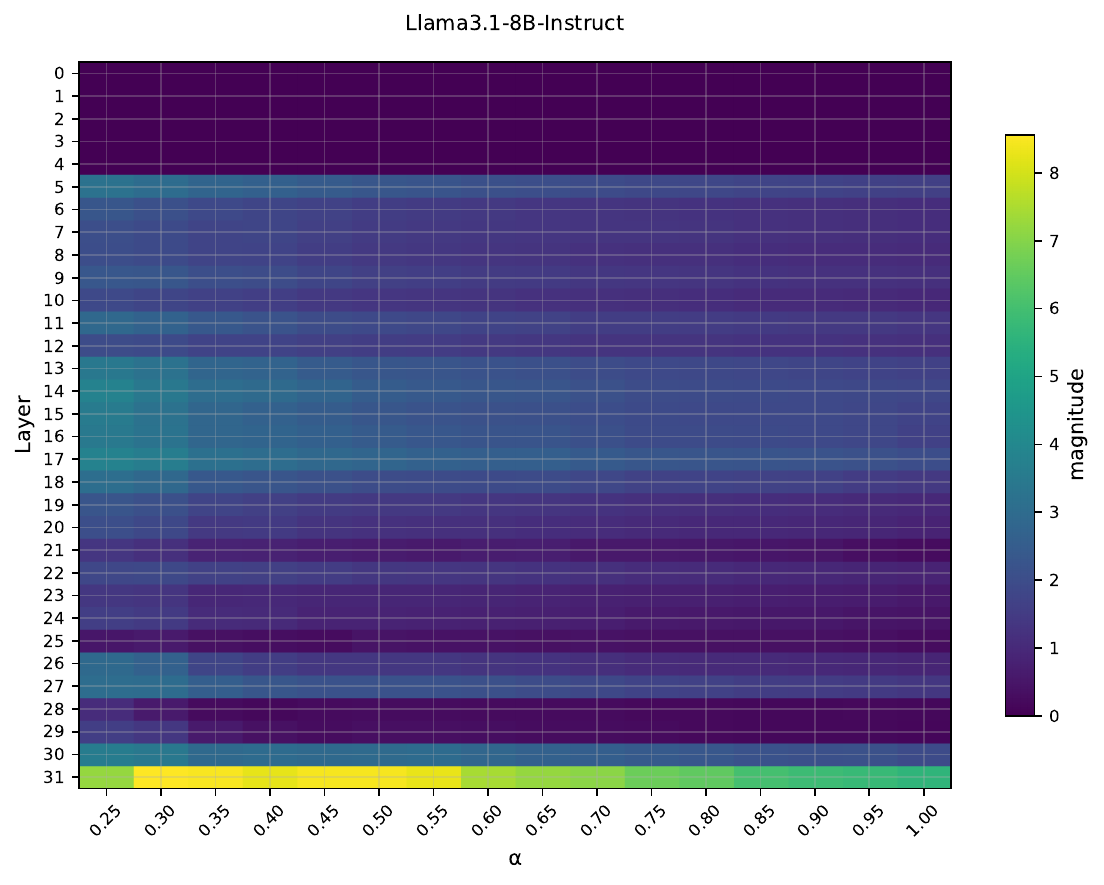}
				\subcaption{Llama3.1 APR+LWP}
				\label{subfig:12}
			\end{subfigure}
	\caption{Layer~\(\times\)~\(\alpha\) heatmap of average perturbation magnitudes \(R_l(\alpha)\) on JailbreakBench. The model is Llama3.1-8B-Instruct.}
	\label{appfig:heatmap on JB-llama}  
\end{figure}

\begin{figure}[htbp]
	\centering 
	\setlength{\tabcolsep}{2pt}  
		\begin{subfigure}[b]{0.48\textwidth}
				\centering
				\includegraphics[width=\textwidth, keepaspectratio]{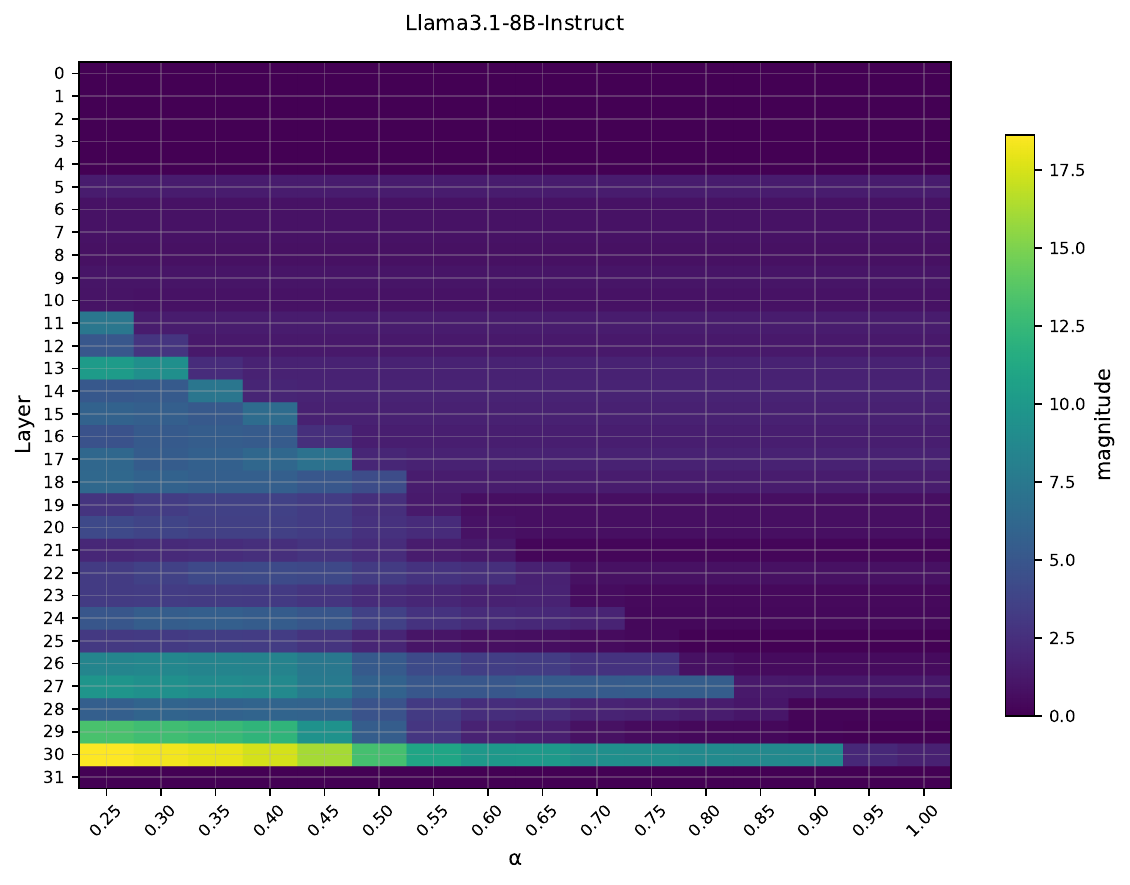} 
				\subcaption{Llama3.1 AIR+GWP}
				\label{subfig:9}
			\end{subfigure}
		\hfill
		\begin{subfigure}[b]{0.48\textwidth}
				\centering
				\includegraphics[width=\textwidth, keepaspectratio]{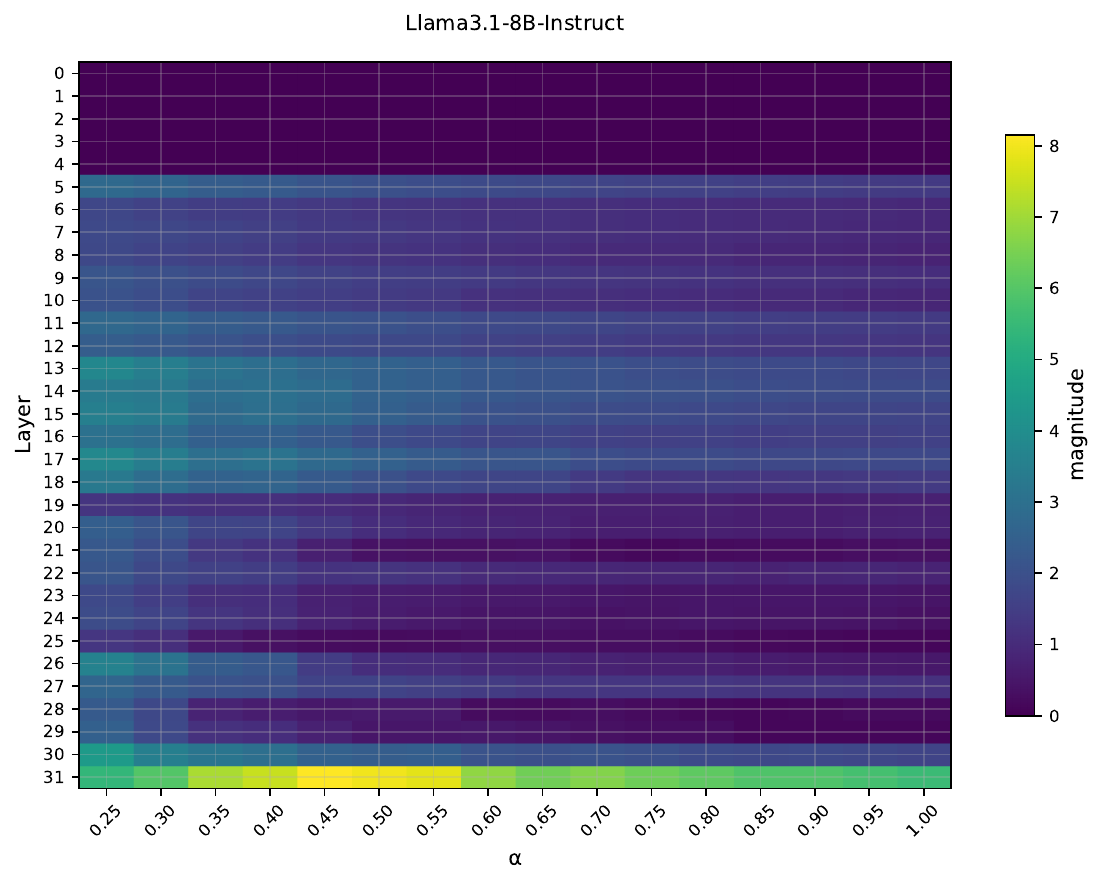}
				\subcaption{Llama3.1 AIR+LWP}
				\label{subfig:10}
			\end{subfigure}
        \vspace{0.5cm}
		\begin{subfigure}[b]{0.48\textwidth}
				\centering
				\includegraphics[width=\textwidth, keepaspectratio]{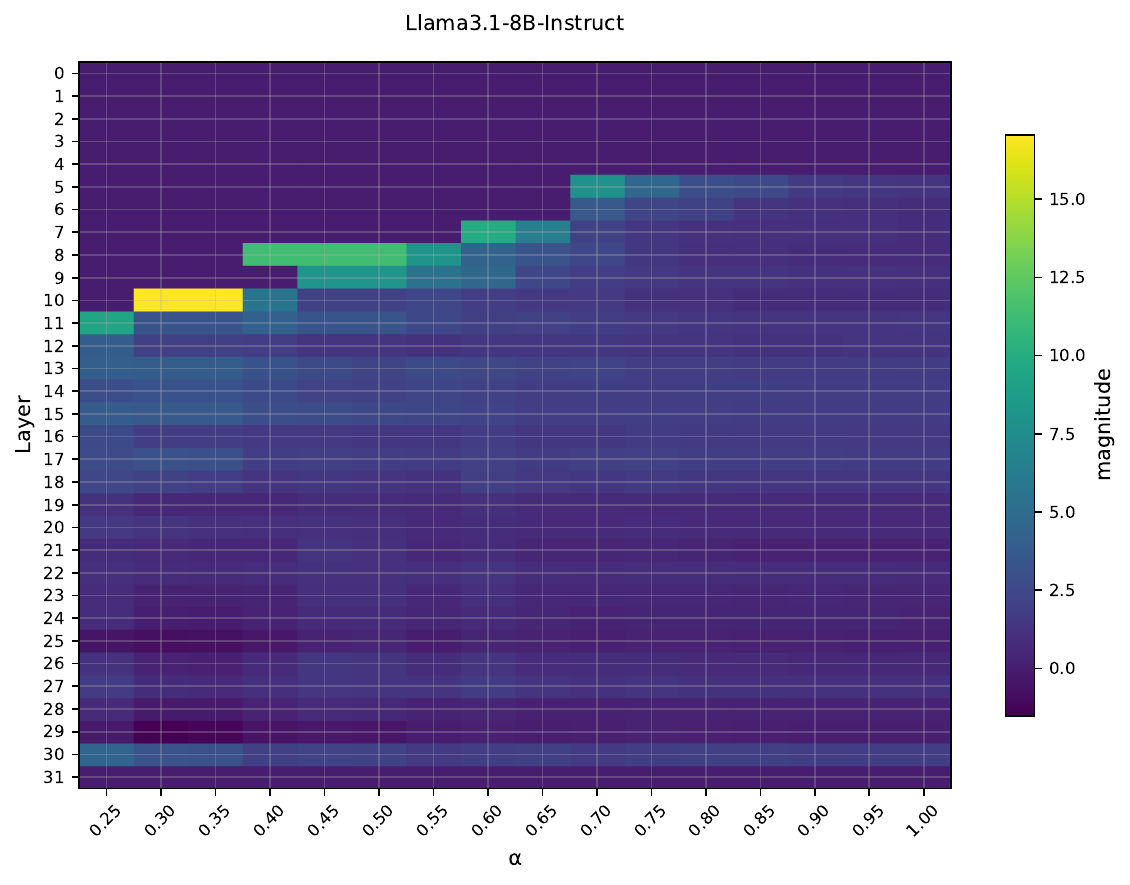}
				\subcaption{Llama3.1 APR+GWP}
				\label{subfig:11}
			\end{subfigure}
		\hfill
		\begin{subfigure}[b]{0.48\textwidth}
				\centering
				\includegraphics[width=\textwidth, keepaspectratio]{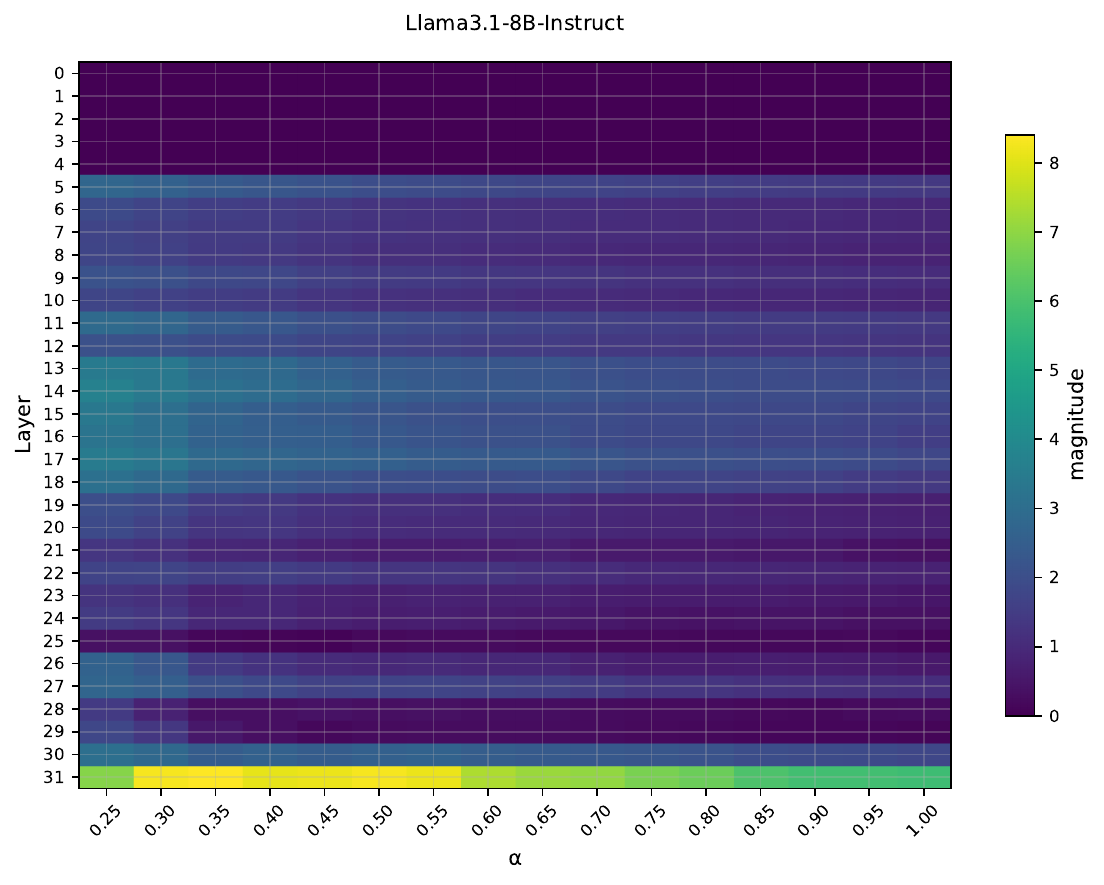}
				\subcaption{Llama3.1 APR+LWP}
				\label{subfig:12}
			\end{subfigure}
	\caption{Layer~\(\times\)~\(\alpha\) heatmap of average perturbation magnitudes \(R_l(\alpha)\) on MaliciousInstruct. The model is Llama3.1-8B-Instruct.}
	\label{appfig:heatmap on MI-llama}  
\end{figure}

\begin{figure}[htbp]
	\centering 
	\setlength{\tabcolsep}{2pt}  
		\begin{subfigure}[b]{0.48\textwidth}
				\centering
				\includegraphics[width=\textwidth, keepaspectratio]{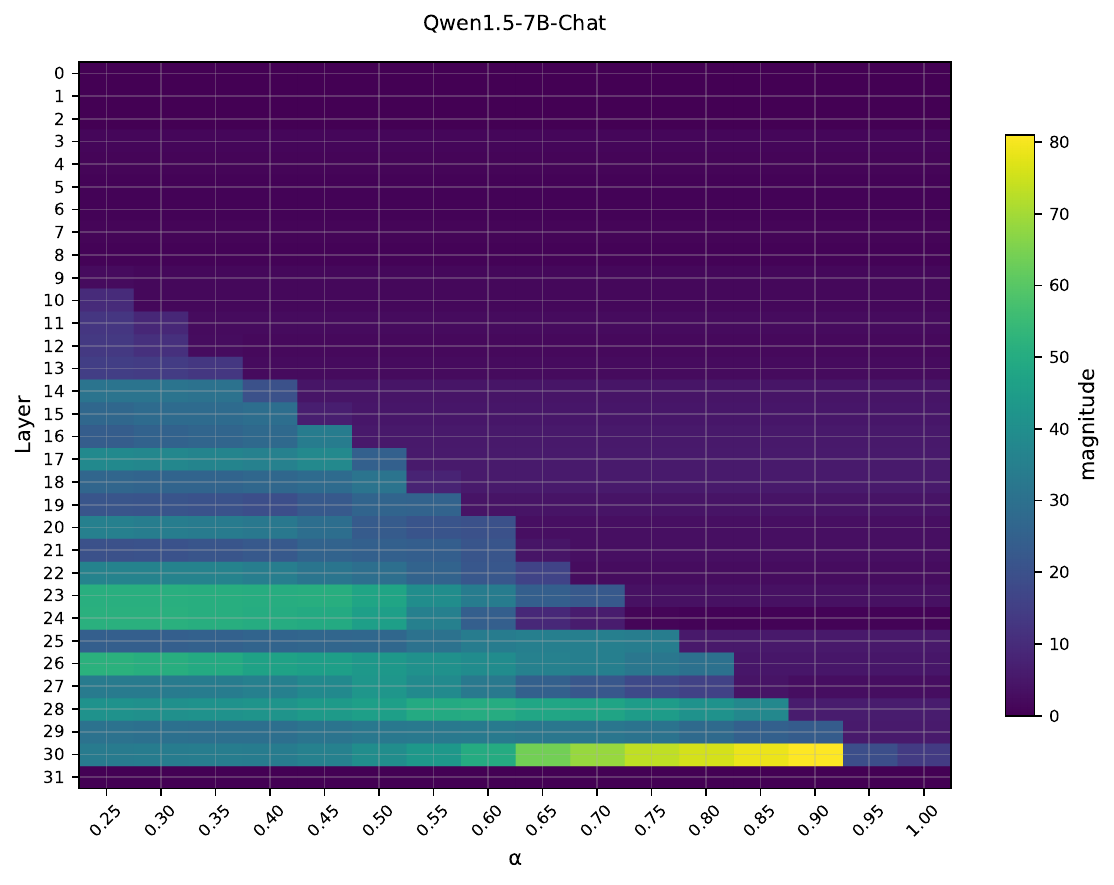}
				\subcaption{Qwen AIR+GWP}
				\label{subfig:13}
			\end{subfigure}
		\hfill
		\begin{subfigure}[b]{0.48\textwidth}
				\centering
				\includegraphics[width=\textwidth, keepaspectratio]{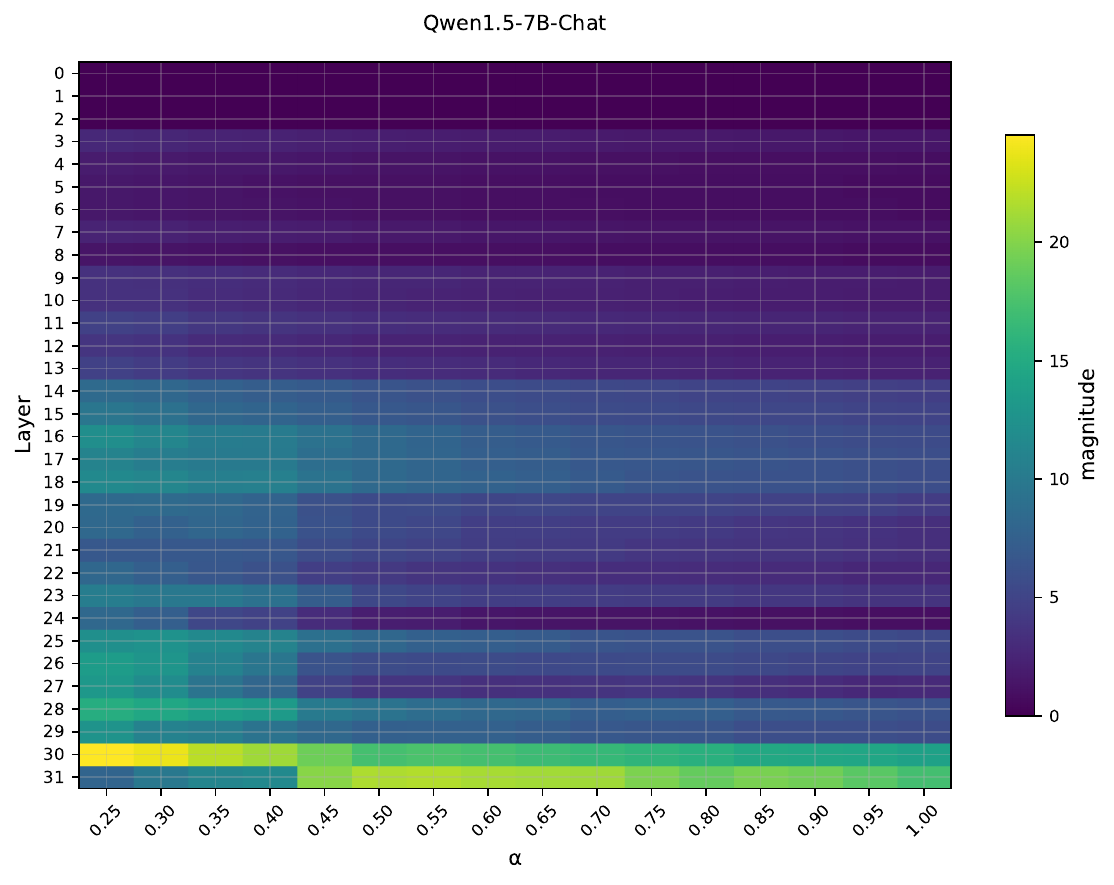}
				\subcaption{Qwen AIR+LWP}
				\label{subfig:14}
			\end{subfigure}
        \vspace{0.5cm}
		\begin{subfigure}[b]{0.48\textwidth}
				\centering
				\includegraphics[width=\textwidth, keepaspectratio]{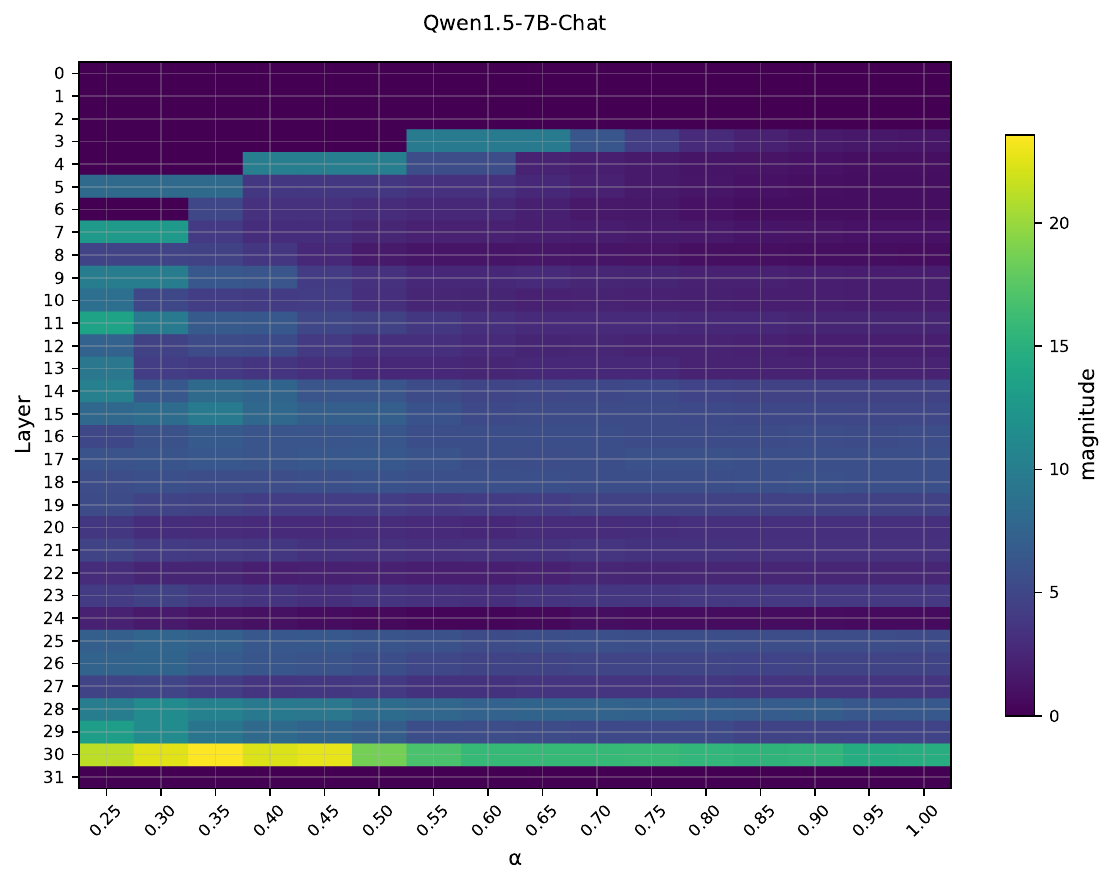}
				\subcaption{Qwen APR+GWP}
				\label{subfig:15}
			\end{subfigure}
		\hfill
		\begin{subfigure}[b]{0.48\textwidth}
				\centering
				\includegraphics[width=\textwidth, keepaspectratio]{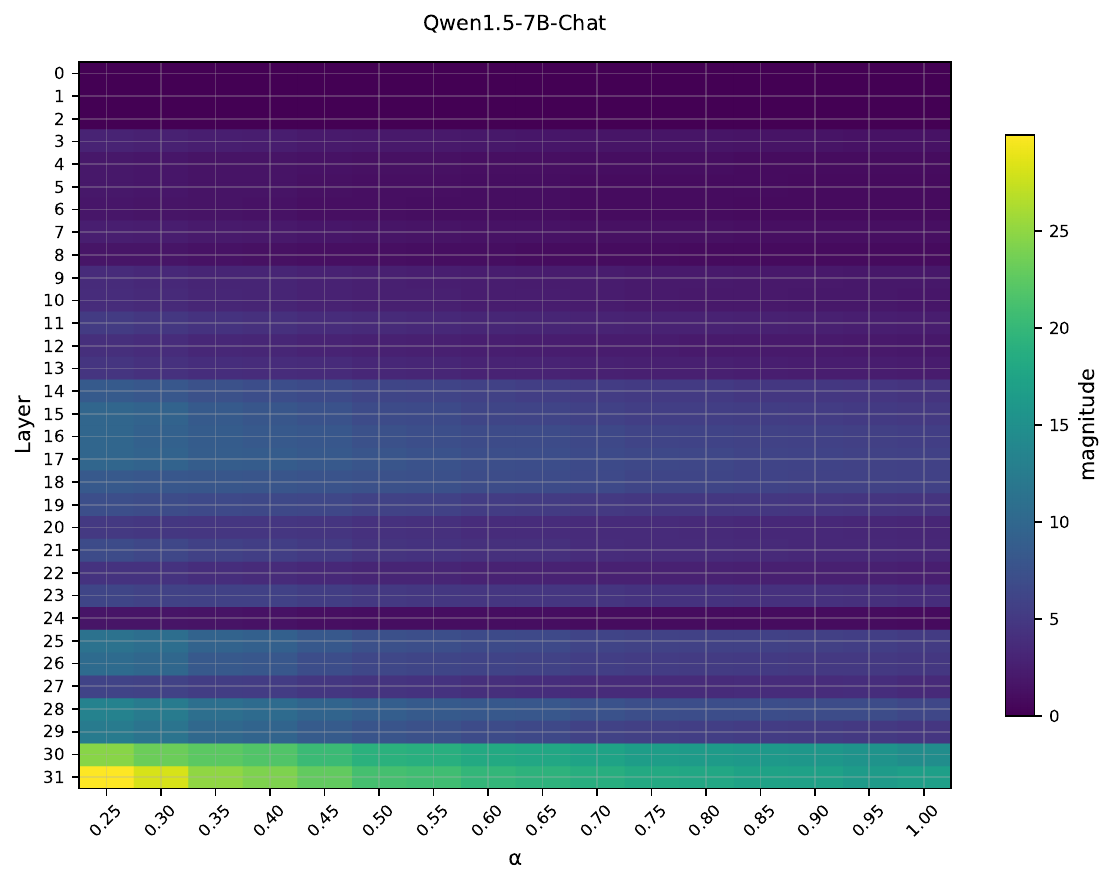}
				\subcaption{Qwen APR+LWP}
				\label{subfig:16}
			\end{subfigure}
	\caption{Layer~\(\times\)~\(\alpha\) heatmap of average perturbation magnitudes \(R_l(\alpha)\) on JailbreakBench. The model is Qwen1.5-7B-Chat.}
	\label{appfig:heatmap on JB-qwen}  
\end{figure}

\begin{figure}[htbp]
	\centering 
	\setlength{\tabcolsep}{2pt}  
		\begin{subfigure}[b]{0.48\textwidth}
				\centering
				\includegraphics[width=\textwidth, keepaspectratio]{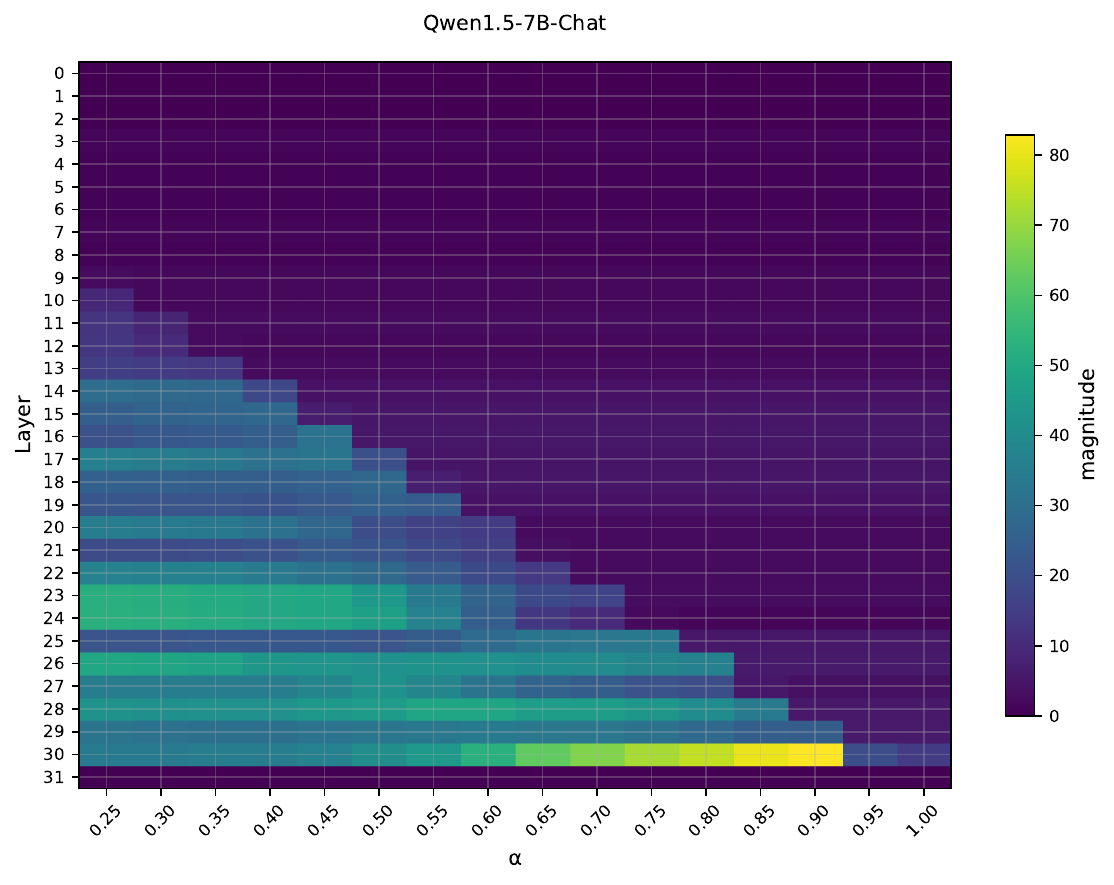}
				\subcaption{Qwen AIR+GWP}
				\label{subfig:13}
			\end{subfigure}
		\hfill
		\begin{subfigure}[b]{0.48\textwidth}
				\centering
				\includegraphics[width=\textwidth, keepaspectratio]{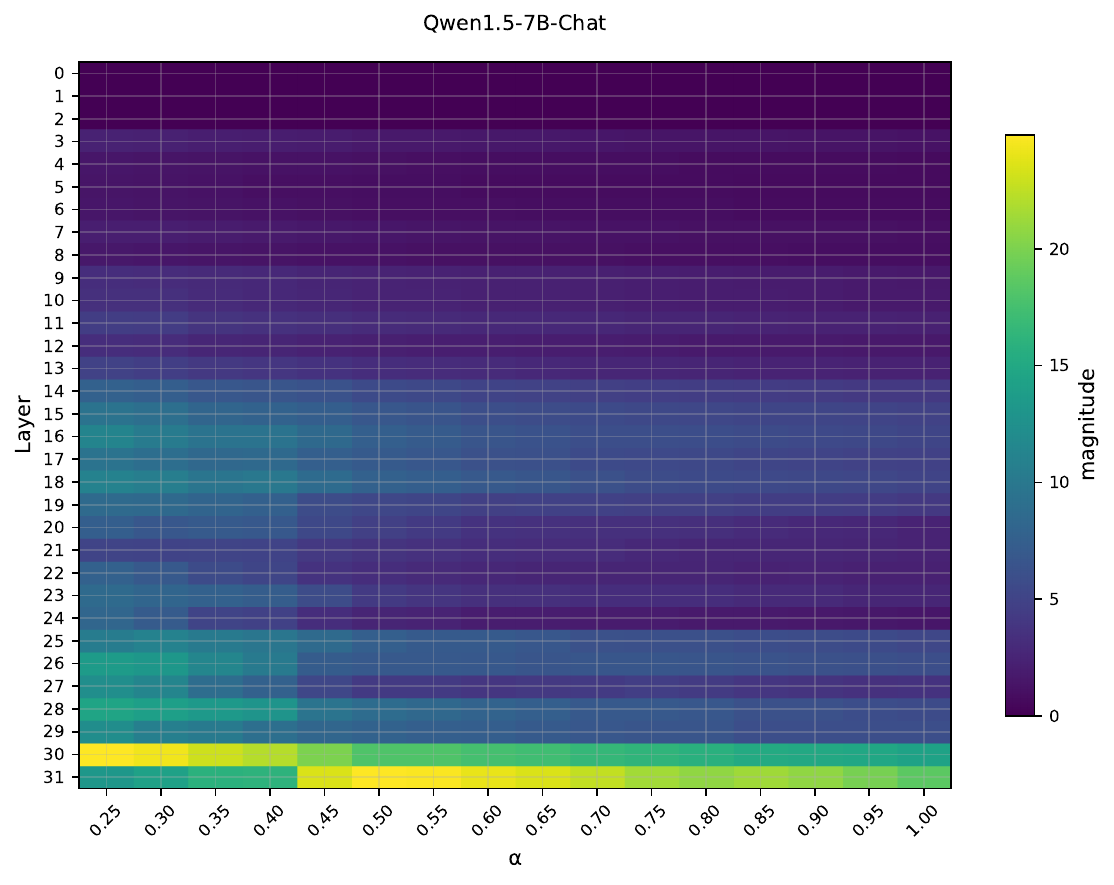}
				\subcaption{Qwen AIR+LWP}
				\label{subfig:14}
			\end{subfigure}
        \vspace{0.5cm}
		\begin{subfigure}[b]{0.48\textwidth}
				\centering
				\includegraphics[width=\textwidth, keepaspectratio]{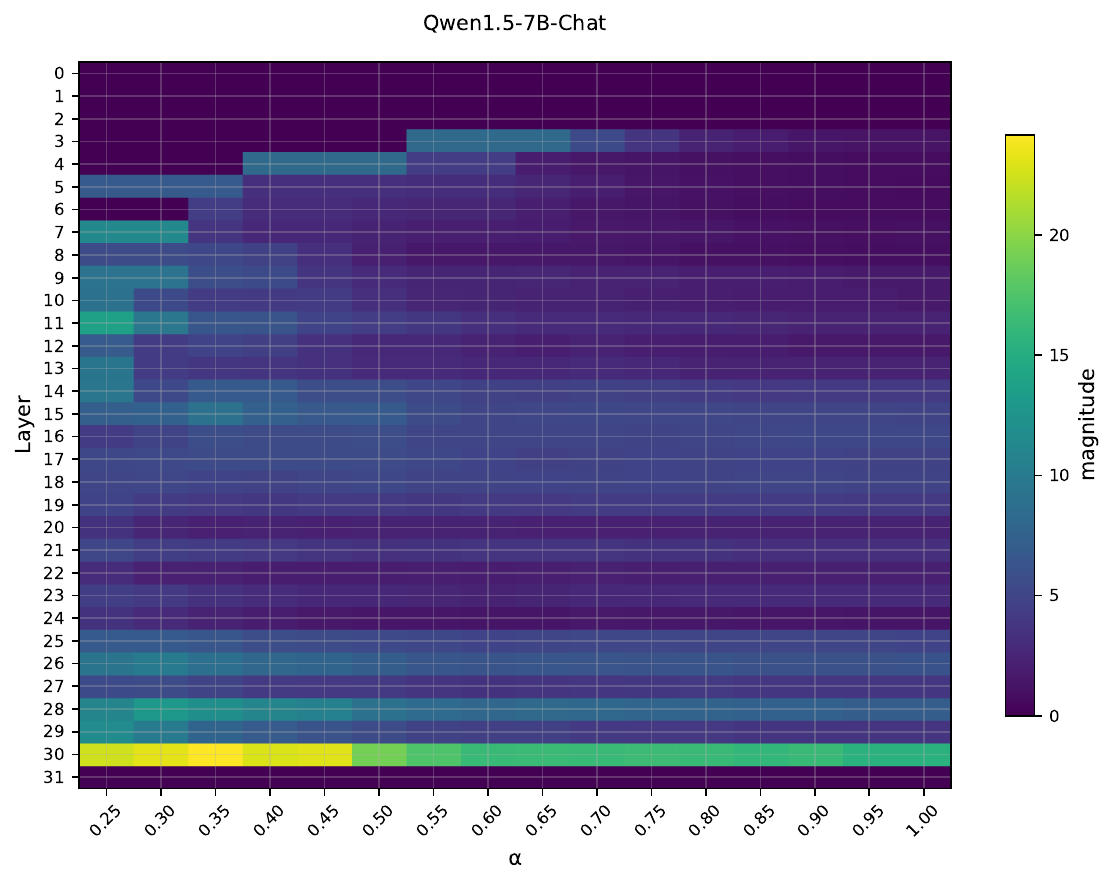}
				\subcaption{Qwen APR+GWP}
				\label{subfig:15}
			\end{subfigure}
		\hfill
		\begin{subfigure}[b]{0.48\textwidth}
				\centering
				\includegraphics[width=\textwidth, keepaspectratio]{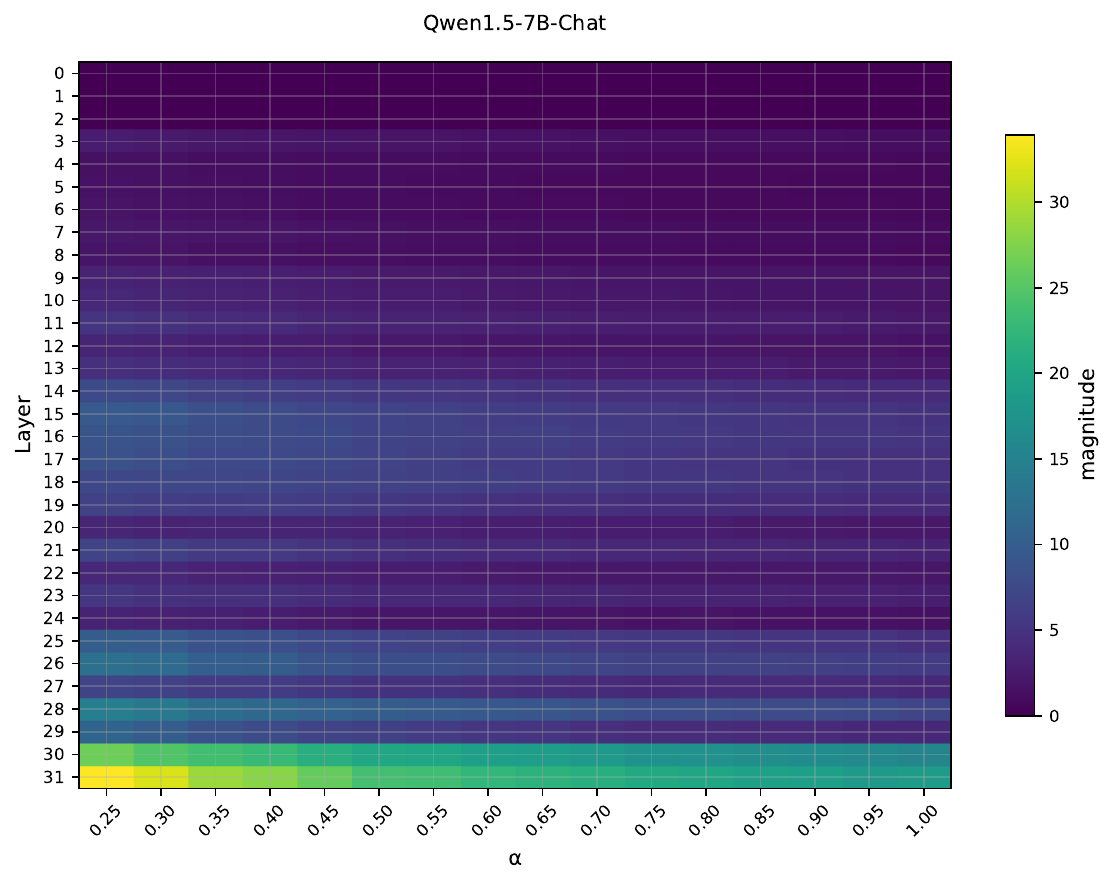}
				\subcaption{Qwen APR+LWP}
				\label{subfig:16}
			\end{subfigure}
	\caption{Layer~\(\times\)~\(\alpha\) heatmap of average perturbation magnitudes \(R_l(\alpha)\) on MaliciousInstruct. The model is Qwen1.5-7B-Chat.}
	\label{appfig:heatmap on MI-qwen}  
\end{figure}
\clearpage
\subsubsection*{C.4 Safety-Critical Head Localization}
Figure~\ref{fig:layer_head_heatmap} reveals distinct patterns in the distribution of safety-critical mechanisms of Deepseek, indicating a more diffuse safety mechanism.

\begin{figure}[ht!] 
	\centering 
	\begin{subfigure}[b]{0.48\textwidth}
		\centering
		\includegraphics[width=\textwidth]{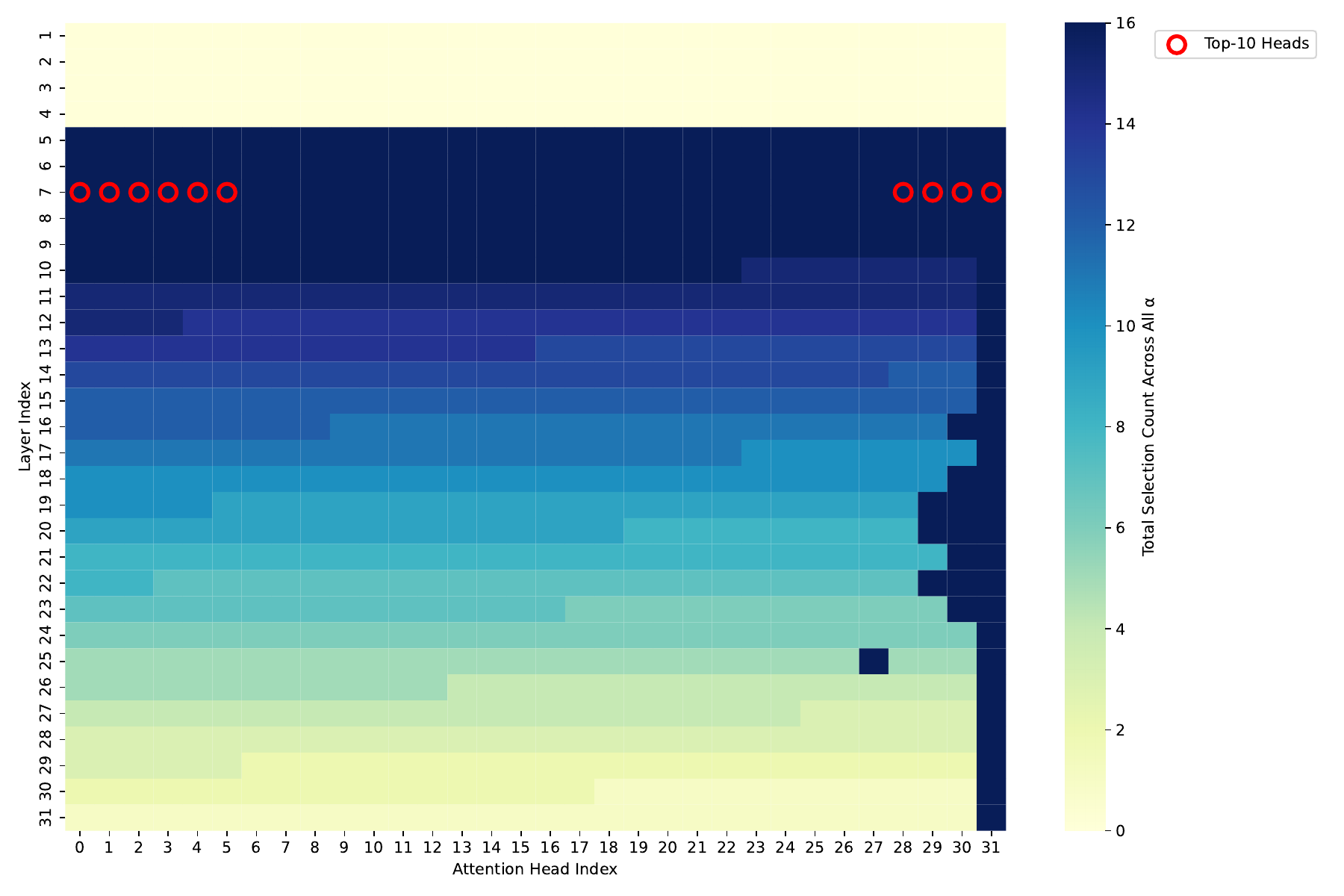}
		\caption{Llama3.1}
		\label{fig:heatmap-llama3.1}
	\end{subfigure}
	\hfill  
    \begin{subfigure}[b]{0.48\textwidth}
		\centering
		\includegraphics[width=\textwidth]{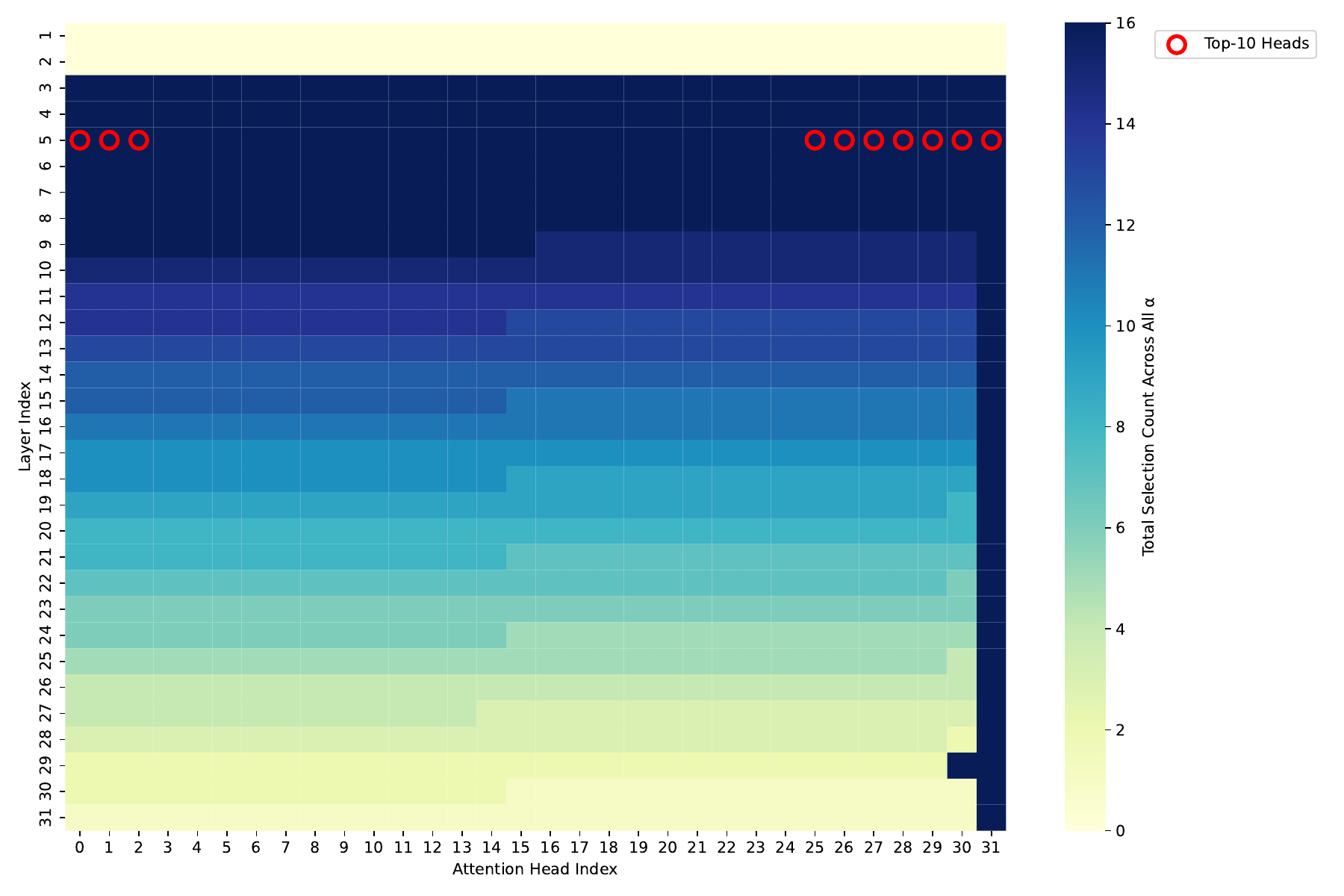}
		\caption{Qwen}
		\label{fig:heatmap-qwen}
	\end{subfigure}
    \vspace{0.5cm}    
    \begin{subfigure}[b]{0.48\textwidth}
		\centering
		\includegraphics[width=\textwidth]{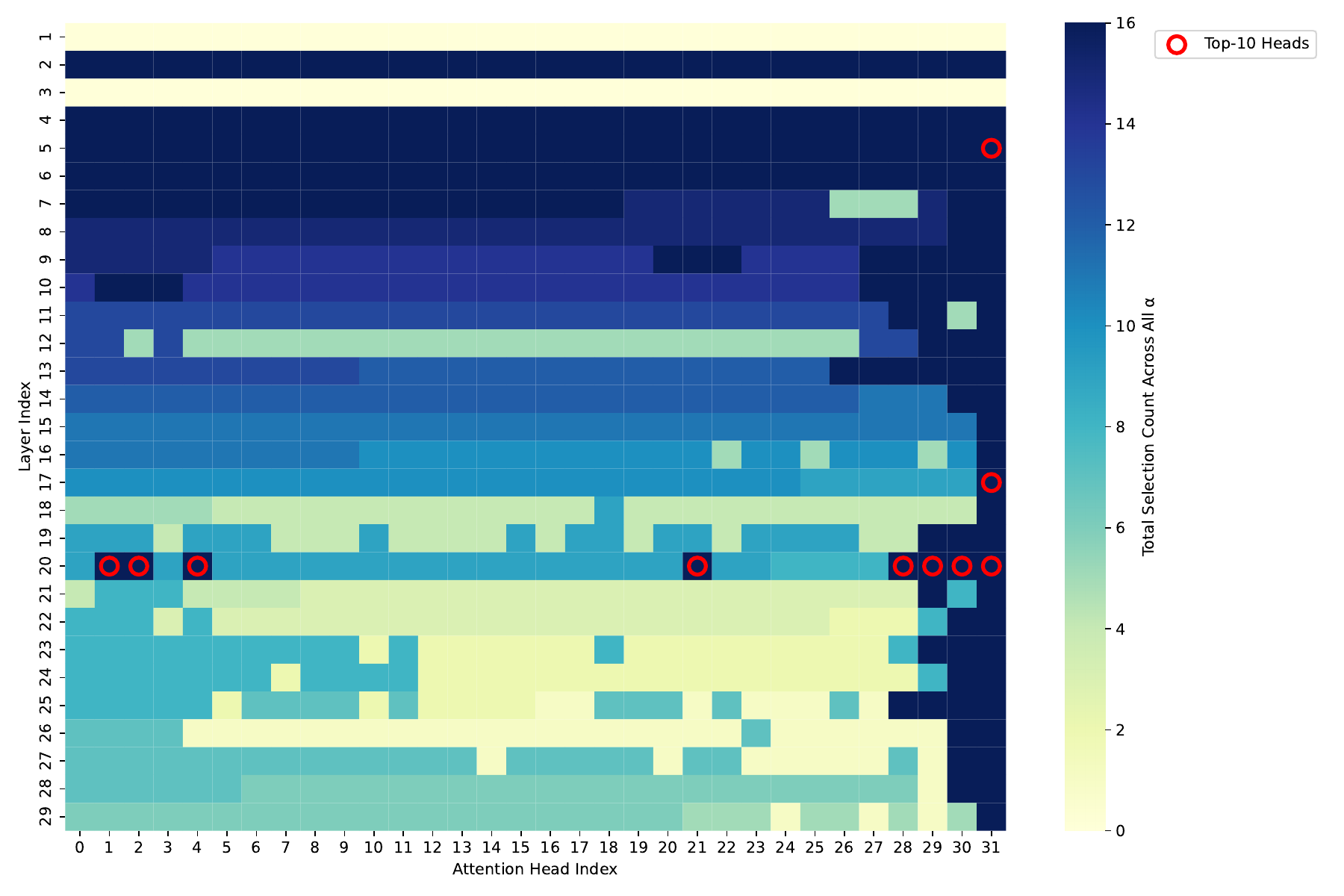}
		\caption{Deepseek}
		\label{fig:heatmap-deepseek}
	\end{subfigure}
    \caption{Spatial distribution of attention head selection frequencies across transformer layers for the AIR+GWP configuration. Color intensity indicates the selection frequency across all $\alpha$ values, while red circles mark the top-10 most frequently selected safety-critical heads.}
    \label{fig:layer_head_heatmap}
\end{figure}
\clearpage

\subsubsection*{C.5 Ablation Study on Structures.} 

\begin{figure}[ht]
	\centering
	\begin{subfigure}[b]{0.32\textwidth}
		\centering
		\includegraphics[width=\textwidth]{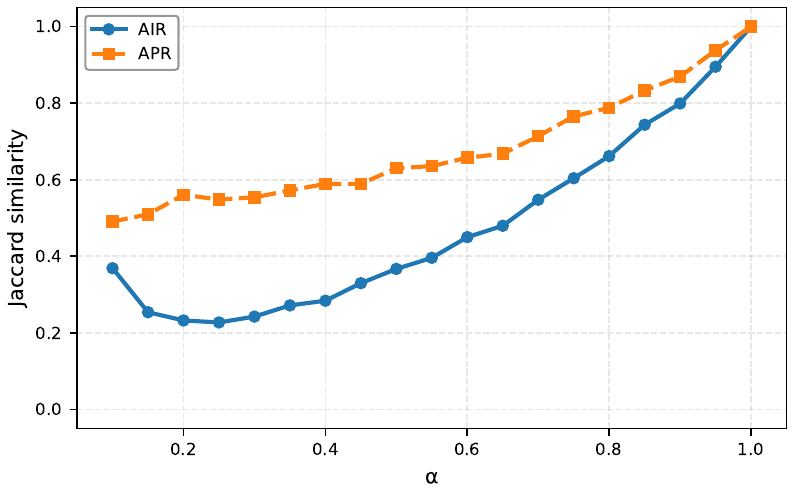}
		\caption{Llama3.1-8B-Instruct}
		\label{fig:jaccard:llama3.1-gl}
	\end{subfigure}
	\hfill
    \begin{subfigure}[b]{0.32\textwidth}
		\centering
		\includegraphics[width=\textwidth]{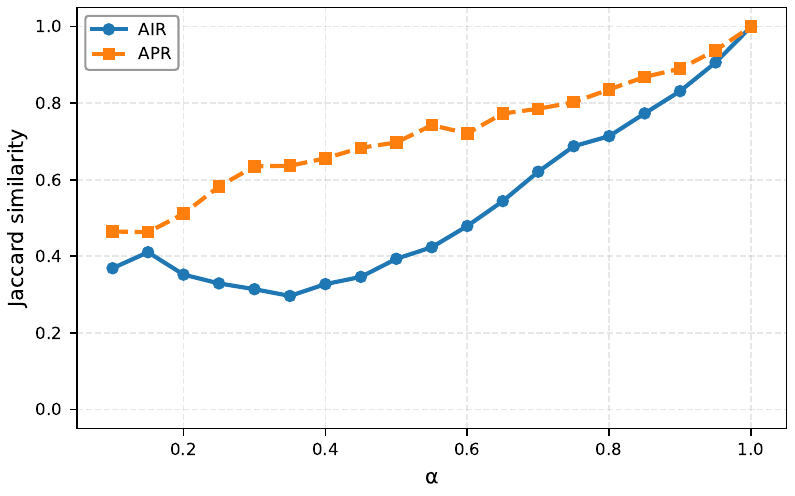}
		\caption{Deepseek-LLM-7B-Chat}
		\label{fig:jaccard:deepseek-gl}
	\end{subfigure}
	\hfill
	\begin{subfigure}[b]{0.32\textwidth}
		\centering
		\includegraphics[width=\textwidth]{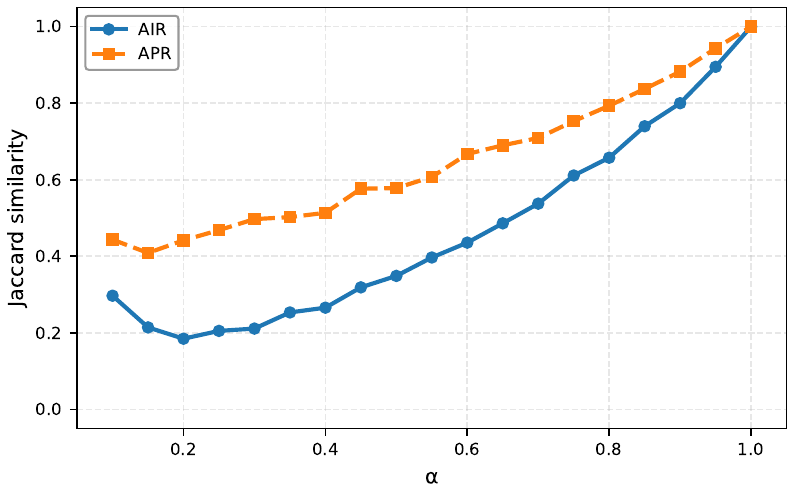}
		\caption{Qwen1.5-7B-Chat}
		\label{fig:jaccard:qwen-gl}
	\end{subfigure}
	\caption{Jaccard similarity comparison of AIR/APR.}
    \vspace{-0.1in}
	\label{fig:jaccard:GL}
\end{figure}

Figure~\ref{fig:jaccard:GL} evaluates the Jaccard similarity between attention-head sets generated by LWP and GWP for the two head-selection strategies. Jaccard Similarity is a metric used to measure the similarity between two sets, quantifying the ratio of the intersection of the two sets to their union. The differing behavior of APR and AIR arises from their fundamentally distinct approaches to measuring head importance. APR assesses each head's standalone predictive power for safety classification via linear probing, yielding nearly unique accuracy scores for individual heads. Since each head encodes distinct local features, their ability to predict safety labels rarely overlaps, eliminating ambiguity in ranking and ensuring consistent top-k head selection across runs. In contrast, AIR quantifies head importance by measuring performance drops upon head ablation, yielding distinct importance scores that reliably rank attention heads. Since each head's contribution to safety is uniquely captured through ablation impact—reflecting its specific role in safeguarding model behavior—the method eliminates ambiguity in selection, ensuring consistent top-k head identification across experimental runs.

\begin{figure}[ht]
	\centering
	\begin{subfigure}[b]{0.32\textwidth}
		\centering
		\includegraphics[width=\textwidth]{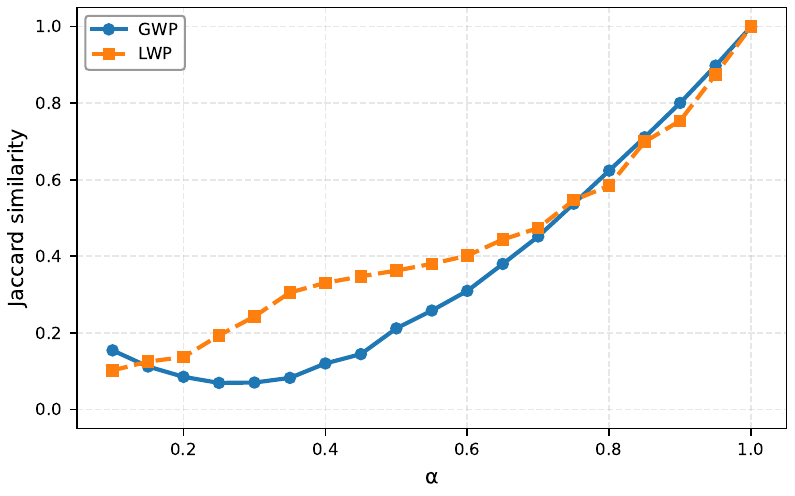}
		\caption{Llama3.1-8B-Instruct}
		\label{jaccard:llama3.1-as}
	\end{subfigure}
	\hfill
    \begin{subfigure}[b]{0.32\textwidth}
		\centering
		\includegraphics[width=\textwidth]{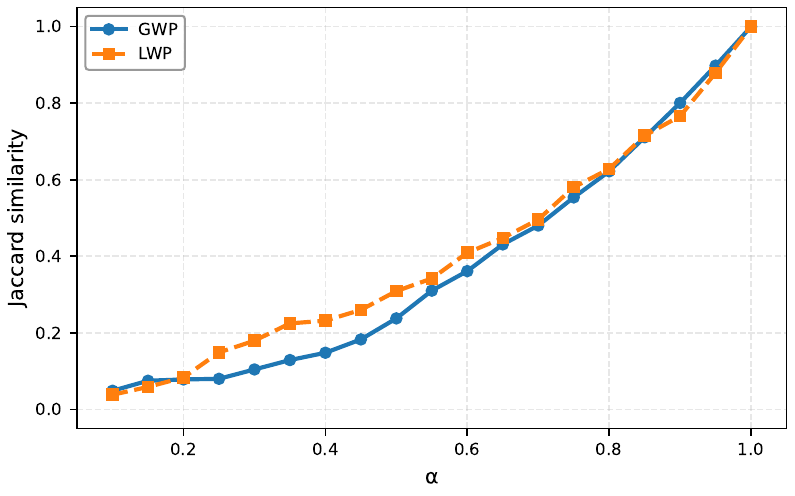}
		\caption{Deepseek-LLM-7B-Chat}
		\label{jaccard:deepseek-as}
	\end{subfigure}
	\hfill
	\begin{subfigure}[b]{0.32\textwidth}
		\centering
		\includegraphics[width=\textwidth]{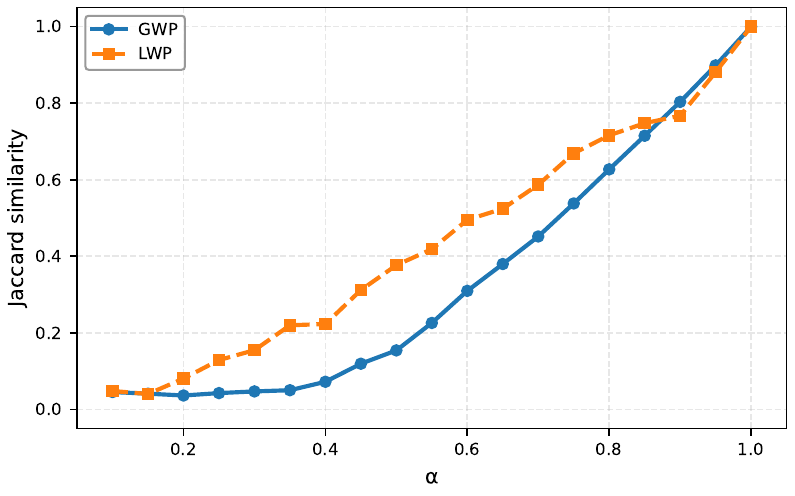}
		\caption{Qwen1.5-7B-Chat}
		\label{jaccard:qwen-as}
	\end{subfigure}
	\caption{Jaccard similarity comparison of LWP/GWP. Minor phase-dependent differences.}
    \vspace{-0.1in}
	\label{fig:jaccard:AS}
\end{figure}

Figure~\ref{fig:jaccard:AS} quantifies the Jaccard similarity between attention-head sets selected by APR and AIR across selection ratio \(\alpha\), focusing on the comparison between two perturbation strategies. 
In the low-$\alpha$ regime ($\alpha < 0.8$), LWP achieves slightly higher similarity than GWP. This improvement arises because LWP employs a layer-wise budget allocation strategy that independently identifies critical heads within each layer, aligning closely with APR's low stochasticity in layer-specific critical head localization. 
In contrast, GWP relies on a global ranking across all layers, which can allow non-critical heads from shallow layers to displace critical heads from deeper ones, thereby diminishing its alignment with the APR/AIR sets. 
In the mid-to-high $\alpha$ range ($0.8 \le \alpha \le 1.0$), where more attention heads are retained, GWP surpasses LWP in similarity. 
As the stochasticity of AIR diminishes in this regime, GWP's global unified ranking not only smooths the remaining minor randomness in AIR but also aligns more closely with APR's global prioritization of cross-layer critical head distribution. 
In contrast, LWP's layer-wise independent selection lacks such global integration, making it less effective in capturing APR's cross-layer head selection priorities.
This trend in Figure~\ref{fig:jaccard:AS} complements Figure~\ref{fig:jaccard:GL} by confirming that LWP and GWP need to adapt to  \(\alpha\) ranges based on the characteristics of APR and AIR to maximize the alignment between APR and AIR headsets.

%% file: sections/11-reference.bib
@Article{1Fei25,
	author = {Fei, Yunqiao and Fan, Jingchao and Zhou, Guomin},
	title = {Extracting fruit disease knowledge from research papers based on large language models and prompt engineering},
	journal = {Applied Sciences},
	volume = 15,
	year = 2025,
	number = 2,
	articleno = 628,
	url = {https://www.mdpi.com/2076-3417/15/2/628},
	issn = {2076-3417},
	doi = {10.3390/app15020628}
}

@InProceedings{2Guo25,
	author = {Guo, Shaopeng},
	booktitle = {2025 IEEE 7th International Conference on Communications, Information System and Computer Engineering (CISCE)},
	title = {DeepLegal-CN: research and application of a deepSeek-based large language model for the legal domain},
	year = 2025,
	pages = {944--947},
	doi = {10.1109/CISCE65916.2025.11065329},
	publisher = {IEEE},
	address = {Los Alamitos, CA}
}

@InProceedings{3Hu25,
	author = {Hu, Su and Zeng, Junyu and Lin, Ximing and Zhong, Xin and Liang, Fangyi and Chen, Peixi and Bai, Xuemei},
	booktitle = {2025 4th International Symposium on Computer Applications and Information Technology (ISCAIT)},
	title = {Accident LLM: a multimodal large language model of accident description based on inter-frame attention mechanism},
	year = 2025,
	pages = {345--350},
	doi = {10.1109/ISCAIT64916.2025.11010470},
	publisher = {IEEE},
	address = {Los Alamitos, CA}
}

@InProceedings{4Sinha25,
	author = {Sinha, Nishchay and Trivedi, Raghav and Mittapalli, Sanjit and Parihar, Anuj and Selvanambi, Ramani},
	booktitle = {2025 International Conference on Cognitive Computing in Engineering, Communications, Sciences and Biomedical Health Informatics (IC3ECSBHI)},
	title = {Targeting and automating recoveries from cybersecurity vulnerabilities using large language models},
	year = 2025,
	pages = {1369--1373},
	doi = {10.1109/IC3ECSBHI63591.2025.10990905},
	publisher = {IEEE},
	address = {Los Alamitos, CA}
}

@InProceedings{5andriushchenko25,
	author = {Andriushchenko, Maksym and Croce, Francesco and Flammarion, Nicolas},
	title = {Jailbreaking leading safety-aligned {LLM}s with simple adaptive attacks},
	booktitle = {The Thirteenth International Conference on Learning Representations},
	year = 2025,
	url = {https://openreview.net/forum?id=hXA8wqRdyV},
	publisher = {ICLR},
	address = {Virtual}
}

@InProceedings{6chen25,
	author = {Chen, Zhuowei and Zhang, Qiannan and Pei, Shichao},
	title = {Injecting universal jailbreak backdoors into {LLM}s in minutes},
	booktitle = {The Thirteenth International Conference on Learning Representations},
	year = 2025,
	url = {https://openreview.net/forum?id=aSy2nYwiZ2},
	publisher = {ICLR},
	address = {Virtual}
}

@InProceedings{7zhao-etal-2025-sql,
	author = {Zhao, Jiawei and Chen, Kejiang and Zhang, Weiming and Yu, Nenghai},
	title = {{SQL} injection jailbreak: a structural disaster of large language models},
	booktitle = {Findings of the Association for Computational Linguistics: ACL 2025},
	month = jul,
	year = 2025,
	address = {Vienna, Austria},
	publisher = {Association for Computational Linguistics},
	url = {https://aclanthology.org/2025.findings-acl.358/},
	doi = {10.18653/v1/2025.findings-acl.358},
	pages = {6871--6891},
	isbn = {979-8-89176-256-5}
}

@InProceedings{8huang-etal-2025-efficient,
	author = {Huang, Yihao and Wang, Chong and Jia, Xiaojun and Guo, Qing and Juefei-Xu, Felix and Zhang, Jian and Liu, Yang and Pu, Geguang},
	title = {Efficient universal goal hijacking with semantics-guided prompt organization},
	booktitle = {Proceedings of the 63rd Annual Meeting of the Association for Computational Linguistics (Volume 1: Long Papers)},
	month = jul,
	year = 2025,
	address = {Vienna, Austria},
	publisher = {Association for Computational Linguistics},
	url = {https://aclanthology.org/2025.acl-long.290/},
	doi = {10.18653/v1/2025.acl-long.290},
	pages = {5796--5816},
	isbn = {979-8-89176-251-0}
}

@InProceedings{9Ahi,
	author = {Ahi, Kiarash and Valizadeh, Saeed},
	booktitle = {2025 Silicon Valley Cybersecurity Conference (SVCC)},
	title = {Large language models (LLMs) and generative AI in cybersecurity and privacy: a survey of dual-use risks, AI-generated malware, explainability, and defensive strategies},
	year = 2025,
	pages = {1--8},
	doi = {10.1109/SVCC65277.2025.11133642},
	publisher = {IEEE},
	address = {Los Alamitos, CA}
}

@InProceedings{10singh,
	author = {Singh, Sonali and Namin, Akbar Siami},
	booktitle = {2024 IEEE International Conference on Big Data (BigData)},
	title = {Adversarial training of retrieval augmented generation to generate believable fake news},
	year = 2024,
	pages = {3589--3598},
	doi = {10.1109/BigData62323.2024.10825933},
	publisher = {IEEE},
	address = {Los Alamitos, CA}
}

@InProceedings{11park,
	author = {Park, Sungwon and Han, Sungwon and Xie, Xing and Lee, Jae-Gil and Cha, Meeyoung},
	title = {Adversarial style augmentation via large language model for robust fake news detection},
	year = 2025,
	isbn = {9798400712746},
	publisher = {Association for Computing Machinery},
	address = {New York, NY, USA},
	url = {https://doi.org/10.1145/3696410.3714569},
	doi = {10.1145/3696410.3714569},
	booktitle = {Proceedings of the ACM on Web Conference 2025},
	pages = {4024--4033},
	location = {Sydney NSW, Australia},
	series = {WWW '25}
}

@inproceedings{12panaitescu-liess-etal-2025-poisonedparrot,
    author = {Panaitescu-Liess, Michael-Andrei  and
	Pathmanathan, Pankayaraj  and
	Kaya, Yigitcan  and
	Che, Zora  and
	An, Bang  and
	Zhu, Sicheng  and
	Agrawal, Aakriti  and
	Huang, Furong},
	title = {{P}oisoned{P}arrot: Subtle Data Poisoning Attacks to Elicit Copyright-Infringing Content from Large Language Models},
	editor = {Chiruzzo, Luis  and
	Ritter, Alan  and
	Wang, Lu},
	booktitle = {Proceedings of the 2025 Conference of the Nations of the Americas Chapter of the Association for Computational Linguistics: Human Language Technologies (Volume 1: Long Papers)},
	month = {apr},
	year = {2025},
	address = {Albuquerque, New Mexico},
	publisher = {Association for Computational Linguistics},
	url = {https://aclanthology.org/2025.naacl-long.415/},
	doi = {10.18653/v1/2025.naacl-long.415},
	pages = {8173--8190},
	ISBN = {979-8-89176-189-6},
}

@InProceedings{13DBLP:conf/acl/ChenLSH0SH25,
	author = {Chen, Yulin and Li, Haoran and Sui, Yuan and He, Yufei and Liu, Yue and Song, Yangqiu and Hooi, Bryan},
	title = {Can indirect prompt injection attacks be detected and removed?},
	year = 2025,
	pages = {18189--18206},
	url = {https://aclanthology.org/2025.acl-long.890/},
	booktitle = {Proceedings of the 63rd Annual Meeting of the Association for Computational Linguistics (Volume 1: Long Papers)},
	publisher = {Association for Computational Linguistics},
	address = {Vienna, Austria},
	month = jul
}

@InProceedings{14zhang-etal-2025-defense,
	author = {Zhang, Ruiyi and Sullivan, David and Jackson, Kyle and Xie, Pengtao and Chen, Mei},
	editor = {Chiruzzo, Luis and Ritter, Alan and Wang, Lu},
	title = {Defense against prompt injection attacks via mixture of encodings},
	booktitle = {Proceedings of the 2025 Conference of the Nations of the Americas Chapter of the Association for Computational Linguistics: Human Language Technologies (Volume 2: Short Papers)},
	month = apr,
	year = 2025,
	address = {Albuquerque, New Mexico},
	publisher = {Association for Computational Linguistics},
	url = {https://aclanthology.org/2025.naacl-short.21/},
	doi = {10.18653/v1/2025.naacl-short.21},
	pages = {244--252},
	isbn = {979-8-89176-190-2}
}

@InProceedings{15yi2025probe,
	author = {Yi, Biao and Huang, Tiansheng and Chen, Sishuo and Li, Tong and Liu, Zheli and Chu, Zhixuan and Li, Yiming},
	title = {Probe before you talk: towards black-box defense against backdoor unalignment for large language models},
	booktitle = {The Thirteenth International Conference on Learning Representations},
	year = 2025,
	url = {https://openreview.net/forum?id=EbxYDBhE3S},
	publisher = {ICLR},
	address = {Virtual}
}

@InProceedings{15-1zhang-etal-2024-defending,
	author = {Zhang, Zhexin and Yang, Junxiao and Ke, Pei and Mi, Fei and Wang, Hongning and Huang, Minlie},
	editor = {Ku, Lun-Wei and Martins, Andre and Srikumar, Vivek},
	title = {Defending large language models against jailbreaking attacks through goal prioritization},
	booktitle = {Proceedings of the 62nd Annual Meeting of the Association for Computational Linguistics (Volume 1: Long Papers)},
	month = aug,
	year = 2024,
	address = {Bangkok, Thailand},
	publisher = {Association for Computational Linguistics},
	url = {https://aclanthology.org/2024.acl-long.481/},
	doi = {10.18653/v1/2024.acl-long.481},
	pages = {8865--8887}
}

@InProceedings{17DBLP,
	author = {Zou, Qingsong and Xiao, Jingyu and Li, Qing and Yan, Zhi and Wang, Yuhang and Xu, Li and Wang, Wenxuan and Gao, Kuofeng and Li, Ruoyu and Jiang, Yong},
	title = {QueryAttack: jailbreaking aligned large language models using structured non-natural query language},
	year = 2025,
	pages = {5725--5741},
	url = {https://aclanthology.org/2025.findings-acl.298/},
	booktitle = {Findings of the Association for Computational Linguistics: ACL 2025},
	publisher = {Association for Computational Linguistics},
	address = {Vienna, Austria},
	month = jul
}

@InProceedings{18saiem,
	author = {Saiem, Bijoy Ahmed and Shanto, MD Sadik Hossain and Ahsan, Rakib and Rashid, Md Rafi Ur},
	editor = {Zhao, Jin and Wang, Mingyang and Liu, Zhu},
	title = {{SequentialBreak}: large language models can be fooled by embedding jailbreak prompts into sequential prompt chains},
	booktitle = {Proceedings of the 63rd Annual Meeting of the Association for Computational Linguistics (Volume 4: Student Research Workshop)},
	month = jul,
	year = 2025,
	address = {Vienna, Austria},
	publisher = {Association for Computational Linguistics},
	url = {https://aclanthology.org/2025.acl-srw.37/},
	doi = {10.18653/v1/2025.acl-srw.37},
	pages = {548--579},
	isbn = {979-8-89176-254-1}
}

@InProceedings{19ha-etal-2025-one,
	author = {Ha, Junwoo and Kim, Hyunjun and Yu, Sangyoon and Park, Haon and Yousefpour, Ashkan and Park, Yuna and Kim, Suhyun},
	editor = {Che, Wanxiang and Nabende, Joyce and Shutova, Ekaterina and Pilehvar, Mohammad Taher},
	title = {One-Shot is enough: consolidating multi-turn attacks into efficient single-turn prompts for {LLM}s},
	booktitle = {Proceedings of the 63rd Annual Meeting of the Association for Computational Linguistics (Volume 1: Long Papers)},
	month = jul,
	year = 2025,
	address = {Vienna, Austria},
	publisher = {Association for Computational Linguistics},
	url = {https://aclanthology.org/2025.acl-long.805/},
	doi = {10.18653/v1/2025.acl-long.805},
	pages = {16489--16507},
	isbn = {979-8-89176-251-0}
}

@InProceedings{20li-etal-2025-revisiting,
	author = {Li, Tianlong and Wang, Zhenghua and Liu, Wenhao and Wu, Muling and Dou, Shihan and Lv, Changze and Wang, Xiaohua and Zheng, Xiaoqing and Huang, Xuanjing},
	editor = {Rambow, Owen and Wanner, Leo and Apidianaki, Marianna and Al-Khalifa, Hend and Di Eugenio, Barbara and Schockaert, Steven},
	title = {Revisiting jailbreaking for large language models: a representation engineering perspective},
	booktitle = {Proceedings of the 31st International Conference on Computational Linguistics},
	month = jan,
	year = 2025,
	address = {Abu Dhabi, UAE},
	publisher = {Association for Computational Linguistics},
	url = {https://aclanthology.org/2025.coling-main.212/},
	pages = {3158--3178}
}

@InProceedings{21li-etal-2025-cavgan,
	author = {Li, Xiaohu and Ning, Yunfeng and Bao, Zepeng and Xu, Mayi and Chen, Jianhao and Qian, Tieyun},
	editor = {Che, Wanxiang and Nabende, Joyce and Shutova, Ekaterina and Pilehvar, Mohammad Taher},
	title = {{CAVGAN}: unifying jailbreak and defense of {LLM}s via generative adversarial attacks on their internal representations},
	booktitle = {Findings of the Association for Computational Linguistics: ACL 2025},
	month = jul,
	year = 2025,
	address = {Vienna, Austria},
	publisher = {Association for Computational Linguistics},
	url = {https://aclanthology.org/2025.findings-acl.346/},
	doi = {10.18653/v1/2025.findings-acl.346},
	pages = {6664--6678},
	isbn = {979-8-89176-256-5}
}

@InProceedings{24vaswani,
	author = {Vaswani, Ashish and Shazeer, Noam and Parmar, Niki and Uszkoreit, Jakob and Jones, Llion and Gomez, Aidan N. and Kaiser, \L{}ukasz and Polosukhin, Illia},
	title = {Attention is all you need},
	year = 2017,
	isbn = {9781510860964},
	publisher = {Curran Associates Inc.},
	address = {Red Hook, NY, USA},
	booktitle = {Proceedings of the 31st International Conference on Neural Information Processing Systems},
	pages = {6000--6010},
	location = {Long Beach, California, USA},
	series = {NIPS'17}
}

@InProceedings{25DBLP:conf/nips/XuHCW24,
	author = {Xu, Zhihao and Huang, Ruixuan and Chen, Changyu and Wang, Xiting},
	title = {Uncovering safety risks of large language models through concept activation vector},
	year = 2024,
	url = {http://papers.nips.cc/paper_files/paper/2024/hash/d3a230d716e65afab578a8eb31a8d25f-Abstract-Conference.html},
	booktitle = {Proceedings of the 38th International Conference on Neural Information Processing Systems},
	publisher = {NeurIPS Foundation},
	address = {New Orleans, LA, USA},
	series = {NeurIPS'24}
}

@inproceedings{
	26Zhang,
	title={JBShield: Defending Large Language Models 
	from Jailbreak Attacks through Activated 
	Concept Analysis and Manipulation},
	author={Zhang, Shenyi and Zhai, Yuchen and Guo, Keyan and Hu, Hongxin and Guo, Shengnan and Fang, Zheng and Zhao, Lingchen and Shen, Chao and Wang, Cong and Wang, Qian},
	booktitle={Proceedings of the 
	34th USENIX Security Symposium},
	month = {aug},
	year={2025},
	pages = {8215-8234},
	address = {Seattle, WA, USA},
	doi={ 978-1-939133-52-6},
	url={ https://www.usenix.org/conference/usenixsecurity25/presentation/zhang-shenyi}
}

@InProceedings{27zhang,
	author = {Zhang, Hanyu and Wang, Xiting and Li, Chengao and Ao, Xiang and He, Qing},
	title = {Controlling large language models through concept activation vectors},
	year = 2025,
	booktitle = {Proceedings of the AAAI Conference on Artificial Intelligence},
	volume = 39,
	number = 24,
	pages = {25851--25859},
	doi = {10.1609/aaai.v39i24.34778},
	publisher = {AAAI Press},
	address = {Washington, DC, USA}
}

@inproceedings{28rimsky-etal-2024-steering,
	title = {Steering Llama 2 via Contrastive Activation Addition},
	author = {Rimsky, Nina  and
	Gabrieli, Nick  and
	Schulz, Julian  and
	Tong, Meg  and
	Hubinger, Evan  and
	Turner, Alexander},
	editor = {Ku, Lun-Wei  and
	Martins, Andre  and
	Srikumar, Vivek},
	booktitle = {Proceedings of the 62nd Annual Meeting of the Association for Computational Linguistics (Volume 1: Long Papers)},
	month = {aug},
	year = {2024},
	address = {Bangkok, Thailand},
	publisher = {Association for Computational Linguistics},
	url = {https://aclanthology.org/2024.acl-long.828/},
	doi = {10.18653/v1/2024.acl-long.828},
	pages = {15504--15522},
}

@inproceedings{29jin-etal-2025-internal,
	title = {Internal Value Alignment in Large Language Models through Controlled Value Vector Activation},
	author = {Jin, Haoran  and
	Li, Meng  and
	Wang, Xiting  and
	Xu, Zhihao  and
	Huang, Minlie  and
	Jia, Yantao  and
	Lian, Defu},
	editor = {Che, Wanxiang  and
	Nabende, Joyce  and
	Shutova, Ekaterina  and
	Pilehvar, Mohammad Taher},
	booktitle = {Proceedings of the 63rd Annual Meeting of the Association for Computational Linguistics (Volume 1: Long Papers)},
	month = {jul},
	year = {2025},
	address = {Vienna, Austria},
	publisher = {Association for Computational Linguistics},
	url = {https://aclanthology.org/2025.acl-long.1326/},
	doi = {10.18653/v1/2025.acl-long.1326},
	pages = {27347--27371},
	ISBN = {979-8-89176-251-0},
}

@misc{30qwen1.5,
	author = {Qwen Team},
	title = {Introducing Qwen1.5},
	url = {https://qwenlm.github.io/blog/qwen1.5/},
	month = feb,
	year = 2024,
	howpublished = {Technical Blog}
}

@misc{32llama3,
	author = {Grattafiori, Aaron and Dubey, Abhimanyu and Jauhri, Abhinav and Pandey, Abhinav and Kadian, Abhishek and Al-Dahle, Ahmad and Letman, Aiesha and Mathur, Akhil and Schelten, Alan and Vaughan, Alex and et al},
	title = {The Llama 3 herd of models},
	year = 2024,
	eprint = {2407.21783},
	archivePrefix = {arXiv},
	primaryClass = {cs.AI},
	url = {https://arxiv.org/abs/2407.21783},
	howpublished = {arXiv preprint}
}

@article{33Deepseek,
	author = {Bi, Xiao and Chen, Deli and Chen, Guanting and Chen, Shanhuang and Dai, Damai and Deng, Chengqi and Ding, Honghui and Dong, Kai and Du, Qiushi and Fu, Zhe and et al},
	title = {DeepSeek LLM: scaling open-source language models with longtermism},
	year = 2024,
	journal = {CoRR},
	volume = {abs/2401.02954},
	url = {https://doi.org/10.48550/arXiv.2401.02954},
	doi = {10.48550/arXiv.2401.02954}
}

@InProceedings{34chao2024jailbreakbench,
	author = {Chao, Patrick and Debenedetti, Edoardo and Robey, Alexander and Andriushchenko, Maksym and Croce, Francesco and Sehwag, Vikash and Dobriban, Edgar and Flammarion, Nicolas and Pappas, George J. and Tram{\`e}r, Florian and Hassani, Hamed and Wong, Eric},
	title = {JailbreakBench: an open robustness benchmark for jailbreaking large language models},
	booktitle = {The Thirty-eight Conference on Neural Information Processing Systems Datasets and Benchmarks Track},
	year = 2024,
	url = {https://openreview.net/forum?id=urjPCYZt0I},
	publisher = {NeurIPS Foundation},
	address = {New Orleans, LA, USA}
}

@InProceedings{35huang2024catastrophic,
	author = {Huang, Yangsibo and Gupta, Samyak and Xia, Mengzhou and Li, Kai and Chen, Danqi},
	title = {Catastrophic jailbreak of open-source {LLM}s via exploiting generation},
	booktitle = {The Twelfth International Conference on Learning Representations},
	year = 2024,
	url = {https://openreview.net/forum?id=r42tSSCHPh},
	publisher = {ICLR},
	address = {Virtual}
}

@inproceedings{36mu-etal-2025-stealthy,
	title = {Stealthy Jailbreak Attacks on Large Language Models via Benign Data Mirroring},
	author = {Mu, Honglin and He, Han and Zhou, Yuxin and Feng, Yunlong and Xu, Yang and Qin, Libo and Shi, Xiaoming and Liu, Zeming and Han, Xudong and Shi, Qi and et al.},
	editor = {Chiruzzo, Luis and Ritter, Alan and Wang, Lu},
	booktitle = {Proceedings of the 2025 Conference of the Nations of the Americas Chapter of the Association for Computational Linguistics: Human Language Technologies (Volume 1: Long Papers)},
	month = apr,
	year = {2025},
	address = {Albuquerque, New Mexico},
	publisher = {Association for Computational Linguistics},
	url = {https://aclanthology.org/2025.naacl-long.88/},
	doi = {10.18653/v1/2025.naacl-long.88},
	pages = {1784--1799},
	ISBN = {979-8-89176-189-6}
	}

@inproceedings{37ding-etal-2024-wolf,
	title = {A Wolf in Sheep{'}s Clothing: Generalized Nested Jailbreak Prompts can Fool Large Language Models Easily},
	author = {Ding, Peng and Kuang, Jun and Ma, Dan and Cao, Xuezhi and Xian, Yunsen and Chen, Jiajun and Huang, Shujian},
	editor = {Duh, Kevin and Gomez, Helena and Bethard, Steven},
	booktitle = {Proceedings of the 2024 Conference of the North American Chapter of the Association for Computational Linguistics: Human Language Technologies (Volume 1: Long Papers)},
	month = jun,
	year = {2024},
	address = {Mexico City, Mexico},
	publisher = {Association for Computational Linguistics},
	url = {https://aclanthology.org/2024.naacl-long.118/},
	doi = {10.18653/v1/2024.naacl-long.118},
	pages = {2136--2153}}

@inproceedings{38ramesh2024gpt,
	author={Ramesh, Govind  and Yao, Dou and Wei, Xu},
	title={GPT-4 Jailbreaks Itself with Near-Perfect Success Using Self-Explanation},
	year={2024},
	cdate={1704067200000},
	pages={22139-22148},
	url={https://aclanthology.org/2024.emnlp-main.1235},
	booktitle={EMNLP}
}

@inproceedings{
	39yuan2024gpt,
	title={{GPT}-4 Is Too Smart To Be Safe: Stealthy Chat with {LLM}s via Cipher},
	author={Yuan, Youliang and Jiao, Wenxiang and Wang, Wenxuan and Huang, Jen-tse and He, Pinjia and Shi, Shuming and Tu, Zhaopeng},
	booktitle={The Twelfth International Conference on Learning Representations},
	year={2024},
	url={https://openreview.net/forum?id=MbfAK4s61A}
}

@inproceedings{
	40schwinn2024soft,
	title={Soft Prompt Threats: Attacking Safety Alignment and Unlearning in Open-Source {LLM}s through the Embedding Space},
	author={Schwinn, Leo and Dobre, David and Xhonneux, Sophie and Gidel, Gauthier and G{\"u}nnemann, Stephan},
	booktitle={The Thirty-eighth Annual Conference on Neural Information Processing Systems},
	year={2024},
	url={https://openreview.net/forum?id=CLxcLPfARc}
}

@inproceedings{
	41hase2024smoothed,
	title={Smoothed Embeddings for Robust Language Models},
	author={Hase, Ryo and Rashid, Md Rafi Ur and Lewis, Ashley and Liu, Jing and Koike-Akino, Toshiaki and Parsons, Kieran and Wang, Ye},
	booktitle={Neurips Safe Generative AI Workshop 2024},
	year={2024},
	url={https://openreview.net/forum?id=GkBRng3Hl9}
}

@inproceedings{42zhou-etal-2024-alignment,
	title = {How Alignment and Jailbreak Work: Explain {LLM} Safety through Intermediate Hidden States},
	author = {Zhou, Zhenhong and Yu, Haiyang and Zhang, Xinghua and Xu, Rongwu and Huang, Fei and Li, Yongbin},
	editor = {Al-Onaizan, Yaser and Bansal, Mohit and Chen, Yun-Nung},
	booktitle = {Findings of the Association for Computational Linguistics: EMNLP 2024},
	month = nov,
	year = {2024},
	address = {Miami, Florida, USA},
	publisher = {Association for Computational Linguistics},
	url = {https://aclanthology.org/2024.findings-emnlp.139/},
	doi = {10.18653/v1/2024.findings-emnlp.139},
	pages = {2461--2488}}

@INPROCEEDINGS{43zhou,
	author={Zhou, Shide and Li, Tianlin and Wang, Kailong and Huang, Yihao and Shi, Ling and Liu, Yang and Wang, Haoyu},
	booktitle={2025 IEEE/ACM 47th International Conference on Software Engineering (ICSE)}, 
	title={Understanding the Effectiveness of Coverage Criteria for Large Language Models: A Special Angle from Jailbreak Attacks}, 
	year={2025},
	volume={},
	number={},
	pages={730-742},
	doi={10.1109/ICSE55347.2025.00209}}

@inproceedings{45scalena-etal-2024-multi,
	title = {Multi-property Steering of Large Language Models with Dynamic Activation Composition},
	author = {Scalena, Daniel and Sarti, Gabriele and Nissim, Malvina},
	editor = {Belinkov, Yonatan and Kim, Najoung and Jumelet, Jaap and Mohebbi, Hosein and Mueller, Aaron and Chen, Hanjie},
	booktitle = {Proceedings of the 7th BlackboxNLP Workshop: Analyzing and Interpreting Neural Networks for NLP},
	month = nov,
	year = {2024},
	address = {Miami, Florida, US},
	publisher = {Association for Computational Linguistics},
	url = {https://aclanthology.org/2024.blackboxnlp-1.34/},
	doi = {10.18653/v1/2024.blackboxnlp-1.34},
	pages = {577--603},}

@inproceedings{46nguyen-etal-2025-multi,
	title = {Multi-Attribute Steering of Language Models via Targeted Intervention},
	author = {Nguyen, Duy and Prasad, Archiki and Stengel-Eskin, Elias and Bansal, Mohit},
	editor = {Che, Wanxiang and Nabende, Joyce and Shutova, Ekaterina and Pilehvar, Mohammad Taher},
	booktitle = {Proceedings of the 63rd Annual Meeting of the Association for Computational Linguistics (Volume 1: Long Papers)},
	month = jul,
	year = {2025},
	address = {Vienna, Austria},
	publisher = {Association for Computational Linguistics},
	url = {https://aclanthology.org/2025.acl-long.1007/},
	doi = {10.18653/v1/2025.acl-long.1007},
	pages = {20619--20634},
	ISBN = {979-8-89176-251-0}}

@inproceedings{47deng-etal-2025-unveiling,
	title = {Unveiling Language-Specific Features in Large Language Models via Sparse Autoencoders},
	author = {Deng, Boyi and Wan, Yu and Yang, Baosong and Zhang, Yidan and Feng, Fuli},
	editor = {Che, Wanxiang and Nabende, Joyce and Shutova, Ekaterina and Pilehvar, Mohammad Taher},
	booktitle = {Proceedings of the 63rd Annual Meeting of the Association for Computational Linguistics (Volume 1: Long Papers)},
	month = jul,
	year = {2025},
	address = {Vienna, Austria},
	publisher = {Association for Computational Linguistics},
	url = {https://aclanthology.org/2025.acl-long.229/},
	doi = {10.18653/v1/2025.acl-long.229},
	pages = {4563--4608},
	ISBN = {979-8-89176-251-0}}

@inproceedings{
	48marks2025sparse,
	title={Sparse Feature Circuits: Discovering and Editing Interpretable Causal Graphs in Language Models},
	author={Marks, Samuel and Rager, Can and Michaud, Eric J and Belinkov, Yonatan and Bau, David and Mueller, Aaron},
	booktitle={The Thirteenth International Conference on Learning Representations},
	year={2025},
	url={https://openreview.net/forum?id=I4e82CIDxv}
}

@inproceedings{
	49frankle2018the,
	title={The Lottery Ticket Hypothesis: Finding Sparse, Trainable Neural Networks},
	author={Frankle, Jonathan and Carbin, Michael},
	booktitle={International Conference on Learning Representations},
	year={2019},
	url={https://openreview.net/forum?id=rJl-b3RcF7},
}

@inproceedings{50voita-etal-2019-analyzing,
	title = {Analyzing Multi-Head Self-Attention: Specialized Heads Do the Heavy Lifting, the Rest Can Be Pruned},
	author = {Voita, Elena and Talbot, David and Moiseev, Fedor and Sennrich, Rico and Titov, Ivan},
	editor = {Korhonen, Anna and Traum, David and M{`a}rquez, Llu{'i}s},
	booktitle = {Proceedings of the 57th Annual Meeting of the Association for Computational Linguistics},
	month = jul,
	year = {2019},
	address = {Florence, Italy},
	publisher = {Association for Computational Linguistics},
	url = {https://aclanthology.org/P19-1580/},
	doi = {10.18653/v1/P19-1580},
	pages = {5797--5808},}

@inproceedings{51kovaleva-etal-2019-revealing,
	title = {Revealing the Dark Secrets of {BERT}},
	author = {Kovaleva, Olga and Romanov, Alexey and Rogers, Anna and Rumshisky, Anna},
	editor = {Inui, Kentaro and Jiang, Jing and Ng, Vincent and Wan, Xiaojun},
	booktitle = {Proceedings of the 2019 Conference on Empirical Methods in Natural Language Processing and the 9th International Joint Conference on Natural Language Processing (EMNLP-IJCNLP)},
	month = nov,year = {2019},
	address = {Hong Kong, China},
	publisher = {Association for Computational Linguistics},
	url = {https://aclanthology.org/D19-1445/},
	doi = {10.18653/v1/D19-1445},
	pages = {4365--4374}}

@inproceedings{
57autodanturbo,
title={Auto{DAN}-Turbo: A Lifelong Agent for Strategy Self-Exploration to Jailbreak {LLM}s},
author={Liu, Xiaogeng and Li, Peiran and Suh, G. Edward and Vorobeychik, Yevgeniy and Mao, Zhuoqing and Jha, Somesh and McDaniel, Patrick and Sun, Huan and Li, Bo and Xiao, Chaowei},
booktitle={The Thirteenth International Conference on Learning Representations},
year={2025},
url={https://openreview.net/forum?id=bhK7U37VW8}
}

@inproceedings{
     58autodan,
      title={AutoDAN: Generating Stealthy Jailbreak Prompts on Aligned Large Language Models},
      author={Liu, Xiaogeng and Xu, Nan and Chen, Muhao and Xiao, Chaowei},
      booktitle={The Twelfth International Conference on Learning Representations},
      year={2024},
      url={https://openreview.net/forum?id=7Jwpw4qKkb}
}

@misc{59PAIR,
      title={Jailbreaking Black Box Large Language Models in Twenty Queries}, 
      author={Chao, Patrick and Robey, Alexander and Dobriban, Edgar and Hassani, Hamed and Pappas, George J. and Wong, Eric},
      year={2023},
      eprint={2310.08419},
      archivePrefix={arXiv},
      primaryClass={cs.LG}
}

@misc{60GCG,
      title={Universal and Transferable Adversarial Attacks on Aligned Language Models}, 
      author={Zou, Andy and Wang, Zifan and Carlini, Nicholas and Nasr, Milad and Kolter, J. Zico and Fredrikson, Matt},
      year={2023},
      eprint={2307.15043},
      archivePrefix={arXiv},
      primaryClass={cs.CL},
      url={https://arxiv.org/abs/2307.15043}, 
}

@article{61wang2025comprehensive,
  title={A comprehensive survey in llm (-agent) full stack safety: Data, training and deployment},
  author={Wang, Kun and Zhang, Guibin and Zhou, Zhenhong and Wu, Jiahao and Yu, Miao and Zhao, Shiqian and Yin, Chenlong and Fu, Jinhu and Yan, Yibo and Luo, Hanjun and others},
  journal={arXiv preprint arXiv:2504.15585},
  year={2025}
}

@inproceedings{
63ouyang2022training,
title={Training language models to follow instructions with human feedback},
author={Ouyang, Long and Wu, Jeffrey and Jiang, Xu and Almeida, Diogo and Wainwright, Carroll and Mishkin, Pamela and Zhang, Chong and Agarwal, Sandhini and Slama, Katarina and Gray, Alex and et al},
booktitle={Advances in Neural Information Processing Systems},
editor={Oh, Alice H. and Agarwal, Alekh and Belgrave, Danielle and Cho, Kyunghyun},
year={2022},
url={https://openreview.net/forum?id=TG8KACxEON}
}

@misc{64bai2022constitutionalaiharmlessnessai,
      title={Constitutional AI: Harmlessness from AI Feedback}, 
      author={Bai, Yuntao and Kadavath, Saurav and Kundu, Sandipan and Askell, Amanda and Kernion, Jackson and Jones, Andy and Chen, Anna and Goldie, Anna and Mirhoseini, Azalia and McKinnon, Cameron and et al},
      year={2022},
      eprint={2212.08073},
      archivePrefix={arXiv},
      primaryClass={cs.CL},
      url={https://arxiv.org/abs/2212.08073}, 
}

@inproceedings{
65han2025heads,
title={Heads up! Large Language Models Can Perform Tasks Without Your Instruction via Selective Attention Head Masking},
author={Han, Senyu and Zeng, Hongchuan and Yu, Kai and Chen, Lu},
booktitle={Forty-second International Conference on Machine Learning},
year={2025},
url={https://openreview.net/forum?id=x2Dw9aNbvw}
}

@inproceedings{66zhang2023tell,
author = {Zhang, Qingru and Singh, Chandan and Liu, Liyuan and Liu, Xiaodong and Yu, Bin and Gao, Jianfeng and Zhao, Tuo},
title = {Tell Your Model Where to Attend: Post-hoc Attention Steering for LLMs},
booktitle = {ICLR 2024},
year = {2023},
month = {November},
}

@article{67probing,
    title = {Probing Classifiers: Promises, Shortcomings, and Advances},
    author = {Belinkov, Yonatan},
    journal = {Computational Linguistics},
    volume = {48},
    number = {1},
    month = {mar},
    year = {2022},
    address = {Cambridge, MA},
    publisher = {MIT Press},
    url = {https://aclanthology.org/2022.cl-1.7/},
    pages = {207--219}
}

@inproceedings{68cao,
    title = "Defending Against Alignment-Breaking Attacks via Robustly Aligned {LLM}",
    author = "Cao, Bochuan and Cao, Yuanpu and Lin, Lu and Chen, Jinghui",
    booktitle = "Proceedings of the 62nd Annual Meeting of the Association for Computational Linguistics (Volume 1: Long Papers)",
    month = aug,
    year = "2024",
    address = "Bangkok, Thailand",
    publisher = "Association for Computational Linguistics",
    url = "https://aclanthology.org/2024.acl-long.568/",
    pages = "10542--10560",
}

@inproceedings{69zhao,
    title = "Defending Large Language Models Against Jailbreak Attacks via Layer-specific Editing",
    author = "Zhao, Wei and Li, Zhe and Li, Yige and Zhang, Ye and Sun, Jun",
    booktitle = "Findings of the Association for Computational Linguistics: EMNLP 2024",
    month = nov,
    year = "2024",
    address = "Miami, Florida, USA",
    publisher = "Association for Computational Linguistics",
    url = "https://aclanthology.org/2024.findings-emnlp.293/",
    pages = "5094--5109",
}

@inproceedings{70liu2024formalizing,
  title={Formalizing and Benchmarking Prompt Injection Attacks and Defenses},
  author={Liu, Yupei and Jia, Yuqi and Geng, Runpeng and Jia, Jinyuan and Gong, Neil Zhenqiang},
  booktitle={Proceedings of the 33rd USENIX Security Symposium},
  pages={1831--1847},
  year={2024},
  organization={USENIX Association}
}

@inproceedings{71li2024evaluating,
  title={Evaluating the Instruction-Following Robustness of Large Language Models to Prompt Injection},
  author={Li, Zekun and Peng, Baolin and He, Pengcheng and Yan, Xifeng},
  booktitle={Proceedings of the 2024 Conference on Empirical Methods in Natural Language Processing (EMNLP)},
  pages={557--568},
  year={2024},
  organization={Association for Computational Linguistics}
}

@inproceedings{72arditi2024refusal,
  title = {Refusal in Language Models Is Mediated by a Single Direction},
  author = {Arditi, Andy and Obeso, Oscar and Syed, Aaquib and Paleka, Daniel and Panickssery, Nina and Gurnee, Wes and Nanda, Neel},
  booktitle = {Advances in Neural Information Processing Systems (NeurIPS 2024)},
  year = {2024},
  note = {NeurIPS 2024; arXiv:2406.11717}
}

@article{73yu2024refusal,
  title = {Robust LLM safeguarding via refusal feature adversarial training},
  author = {Yu, Lei and Do, Virginie and Hambardzumyan, Karen and Cancedda, Nicola},
  year = {2024},
  note = {arXiv:2409.20089}
}

@inproceedings{74NEURIPS2023_fd661313,
 author = {Wei, Alexander and Haghtalab, Nika and Steinhardt, Jacob},
 booktitle = {Advances in Neural Information Processing Systems},
 pages = {80079--80110},
 publisher = {Curran Associates, Inc.},
 title = {Jailbroken: How Does LLM Safety Training Fail?},
 url = {https://proceedings.neurips.cc/paper_files/paper/2023/file/fd6613131889a4b656206c50a8bd7790-Paper-Conference.pdf},
 volume = {36},
 year = {2023}
}

@inproceedings{75NEURIPS2023_ac662d74,
 author = {Zhou, Chunting and Liu, Pengfei and Xu, Puxin and Iyer, Srinivasan and Sun, Jiao and Mao, Yuning and Ma, Xuezhe and Efrat, Avia and Yu, Ping and YU, LILI and Zhang, Susan and Ghosh, Gargi and Lewis, Mike and Zettlemoyer, Luke and Levy, Omer},
 booktitle = {Advances in Neural Information Processing Systems},
 pages = {55006--55021},
 publisher = {Curran Associates, Inc.},
 title = {LIMA: Less Is More for Alignment},
 url = {https://proceedings.neurips.cc/paper_files/paper/2023/file/ac662d74829e4407ce1d126477f4a03a-Paper-Conference.pdf},
 volume = {36},
 year = {2023}
}

@inproceedings{76ICLR2024_83b7da3e,
 author = {Qi, Xiangyu and Zeng, Yi and Xie, Tinghao and Chen, Pin-Yu and Jia, Ruoxi and Mittal, Prateek and Henderson, Peter},
 booktitle = {International Conference on Learning Representations},
 editor = {B. Kim and Y. Yue and S. Chaudhuri and K. Fragkiadaki and M. Khan and Y. Sun},
 pages = {30988--31043},
 title = {Fine-tuning Aligned Language Models Compromises Safety, Even When Users Do Not Intend To!},
 url = {https://proceedings.iclr.cc/paper_files/paper/2024/file/83b7da3ed13f06c13ce82235c8eedf35-Paper-Conference.pdf},
 volume = {2024},
 year = {2024}
}

@article{77inan2023llama,
  title={Llama Guard: LLM-based Input-Output Safeguard for Human-AI Conversations},
  author={Inan, Hakan H and Lin, Nelson F and Yu, Huan and Wallace, Eric and Sridhar, Sathyapriya and Wu, Lemeng and Wang, Tianle and Sivaprasad, Abhinav and Mathew, Bin and Tsvetkov, Yulia and others},
  journal={arXiv preprint arXiv:2312.06674},
  year={2023}
}

@inproceedings{78reimers2019sentence,
  title={Sentence-BERT: Sentence Embeddings using Siamese BERT-Networks},
  author={Reimers, Nils and Gurevych, Iryna},
  booktitle={Proceedings of the 2019 Conference on Empirical Methods in Natural Language Processing and the 9th International Joint Conference on Natural Language Processing (EMNLP-IJCNLP)},
  pages={3982--3992},
  year={2019},
  publisher={Association for Computational Linguistics},
  url={https://aclanthology.org/D19-1410/}
}

@inproceedings{79sahara,
    author = {Zhou, Z. and Yu, H. and Zhang, X. and Xu, R. and Huang, F. and Wang, K. and Liu, Y. and Fang, J. and Li, Y.},
    title = {On the Role of Attention Heads in Large Language Model Safety},
    booktitle = {The Thirteenth International Conference on Learning Representations (ICLR)},
    year = {2025},
    address = {Virtual},
    url = {https://openreview.net/forum?id=h0Ak8A5yqw}
}
